\documentclass[10pt,journal]{IEEEtran}
\ifCLASSINFOpdf
\else
\fi

 \usepackage[ruled]{algorithm2e}
 \usepackage{algorithmic}
 \usepackage{subfigure}
 \usepackage{epsfig}
\usepackage{url}
 \newcommand{\co}[1]{}
 \usepackage{amsthm}
 \newtheorem{theorem}{Theorem}[section]
 \newtheorem{lemma}[theorem]{Lemma}
 \newtheorem{corollary}[theorem]{Corollary}
 \newtheorem{definition}[theorem]{Definition}

\newtheorem{defn}[theorem]{Definition}

\begin{document}
%
\title{HybridNN: Supporting Network Location Service on Generalized Delay Metrics}
%
%
%
\author{Yongquan~Fu,
Yijie~Wang, and~Ernst~Biersack,
\thanks{Yongquan Fu and Yijie Wang are with National Key Laboratory for Parallel and
Distributed Processing, College of Computer Science, University of
Defense Technology.}
\thanks{Ernst Biersack is with Networking and Security Department, Eurecom.}
}
\markboth{Journal of \LaTeX\ Class Files,~Vol.~6, No.~1, January~2007}%
{Shell \MakeLowercase{\textit{et al.}}: Bare Demo of IEEEtran.cls
for Journals}
%



\maketitle

\begin{abstract}
Distributed Nearest Neighbor Search (DNNS) locates service nodes
that have shortest interactive delay towards requesting hosts.
DNNS provides an important service for large-scale latency
sensitive networked applications, such as VoIP, online network
games, or interactive network services on the cloud. Existing work
assumes the delay to be symmetric, which does not generalize to
applications that are sensitive to one-way delays, such as the
multimedia video delivery from the servers to the hosts. We
propose a relaxed inframetric model for the network delay space
that does not assume the triangle inequality and delay symmetry to
hold. We prove that the DNNS requests can be completed efficiently
if the delay space exhibits modest inframetric dimensions, which
we can observe empirically. Finally, we propose a DNNS method
named HybridNN (\textit{Hybrid} \textit{N}earest \textit{N}eighbor
search) based on the inframetric model for fast and accurate DNNS.
For DNNS requests, HybridNN chooses closest neighbors accurately
via the inframetric modelling, and scalably by combining delay
predictions with direct probes to a pruned set of neighbors.
Simulation results show that HybridNN locates nearly optimally the
nearest neighbor. Experiments on PlanetLab show that HybridNN can
provide accurate nearest neighbors that are close to optimal with
modest query overhead and maintenance traffic.
\end{abstract}

\section{Introduction}
\label{intro}

Latency-sensitive applications, such as P2P based VoIP and IPTV
\cite{DBLP:journals/cacm/RodriguesD10}, interactive network
services on the cloud (e.g., Office Live Workspace
\cite{MSOffice}, Google Maps \cite{GoogleMap}), online network
games, need to transmit data from geo-distributed servers (called
a service node) in real-time to many hosts quickly. High
transmission delays reduce the Quality of Experience (QoE) of
users \cite{Agboma:2008:QQM:1497185.1497210}, which lead to
significant business losses
\cite{DBLP:journals/ccr/GreenbergHMP09}. For instance, Google
reports that its revenue decreases by 20\% when the latency of
showing search results increases by 500 ms; similarly, Amazon
claims that its sales amount decreases by 1\% if the page-response
latency increases by 100 ms
\cite{DBLP:journals/ccr/GreenbergHMP09}.

Since there are hundreds or thousands of service nodes that
provide identical services to hosts, there is an increasing push
for service providers to route real-time data to a host from
geo-distributed servers that are nearest to that host. For
example, Google routes users' search queries to
geographical-nearby servers \cite{DBLP:conf/imc/KrishnanMSJKAG09};
Akamai redirects hosts' content requests to replica servers mainly
based on proximity conditions \cite{Su:2006:DBA:1151659.1159962};
CoralCDN \cite{coralCDN} uses OASIS
\cite{DBLP:conf/nsdi/FreedmanLM06} and DONAR
\cite{DBLP:conf/sigcomm/WendellJFR10} to select proxy servers near
to end hosts based on geographic distances. However, selecting
nearest servers to hosts are still far from standard due to
several challenges.

\co{Donnybrook disseminates game updates via proximity-aware
multicast; }

First, \textit{selecting nearest servers must prove to be
reliable, since service providers need to ensure the QoE fairly
for all hosts}. Selecting nearest servers using proximity
coordinates
\cite{Guyton95locatingnearby,DBLP:conf/icdcs/CostaCRK04} or
geographic distances \cite{DBLP:conf/nsdi/FreedmanLM06} suffer
from the mismatch between the estimated delays and real-world
delays \cite{DBLP:conf/imc/KrishnanMSJKAG09}, which makes the
selection accuracy hard to be predicted. On the other hand,
selecting nearest servers using distributed search such as
Meridian \cite{DBLP:conf/sigcomm/WongSS05} or OASIS
\cite{DBLP:conf/nsdi/FreedmanLM06} avoid such mismatch problems
using direct probes, but may terminate at service nodes that are
much worse than the nearest ones, since the search is easily
trapped into local minima due to the clustering
\cite{DBLP:conf/imc/VishnumurthyF08} and Triangle Inequality
Violations (TIV) \cite{Lumezanu:2009:TIV:1644893.1644914}
properties of the delay space.

\co{Since  . Besides, new servers may be added and existing
servers may be unavailable due to the planned maintenance or
crashes. As a result, }

Second, \textit{selecting nearest servers must be aware of
unidirectional delays whenever possible}. Since routing on the
Internet is asymmetric \cite{Pathak:2008:MSI:1791949.1791975}, the
delays from servers to hosts may deviate those in the reverse
direction in several times. Furthermore, One-Way Delay (OWD)
measurements become increasingly practical due to the advance of
measurement techniques such as OWAMP \cite{rfc4656} or Reverse
Traceroute \cite{DBLP:conf/nsdi/Katz-BassettMASSWAK10}. However,
delay optimizations using Round Trip Time (RTT) ignores such delay
asymmetry. For multimedia streaming, application-level multicast,
or more generalized applications where data flows in one
directions, such agnostics of unidirectional delays degrades the
effectiveness of selected servers, as shown in Fig
\ref{fig:illustrationGeneralizedDelay}.

  \begin{figure}[tp]
  \leavevmode \centering \setlength{\epsfxsize}{0.5\hsize}
  \epsffile{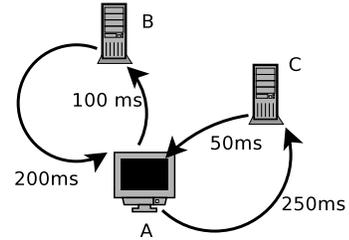}
  \caption{Illustrating the RTT and OWDs. Suppose $B$ and $C$
are two servers that are able to supply short videos to host $A$.
If we use the RTT metric to minimize the delay of video delivery,
we may arbitrarily choose any of them to send videos to host $A$
based on the RTT metric, since the RTT between $A$, $B$ and that
between $A$, $C$ are all 300 ms. However, since the video files
are transmitted from servers to hosts, the OWDs from servers to
hosts become more important
\cite{Pathak:2008:MSI:1791949.1791975}. We can see that the OWD
from server $C$ to host $A$ is four times less than that from
server $B$ to host $A$. Therefore, choosing server $C$ to serve
host $A$ significantly minimizes the content transmission delay
for host $A$, which is feasible only when we use the OWD metric
for delay optimizations.}
  \label{fig:illustrationGeneralizedDelay}
\end{figure}

Third, \textit{selecting nearest servers must find good tradeoff
between the response time and timeliness}. The response time lasts
several seconds for server selections using on-demand probing such
as Meridian \cite{DBLP:conf/sigcomm/WongSS05} or OASIS
\cite{DBLP:conf/nsdi/FreedmanLM06}. However, the response time
degrades the QoE of users in latency-sensitive applications, such
as online workspace, online music. OASIS caches nearest servers
for each IP prefix using in-advance probes once a week, which has
better response time. However, the cached server selections tend
be suboptimal, since the delays vary due to routing dynamics or
server workloads \cite{Paxson:1997:EIP:263105.263155}, and service
nodes may be added or removed dynamically. Therefore, it is
difficult to find good tradeoff between response time and the
timeliness of server selections.

\co{the nearest servers for hosts vary dynamically, since  As a
result, the selection process should be refreshed frequently
enough to adapt to such dynamics, which inevitably incurs
measurement costs. However, the bandwidth costs for locating
nearest servers need to be shortened as low as possible in order
to limit the service expenses. to be useful only in background.
 must adopt a background

, e.g., interactive workspace, Google maps

 Since the QoE of users
may degrade even the response time is within several seconds. For
instance, We leave the third challenge can be resolved using
in-advance server selections similar as OASIS
\cite{DBLP:conf/nsdi/FreedmanLM06}. }

The goal of this paper is to provide new algorithms to address the
first two challenges. To this end, we develop a general enough
delay model that captures the major statistics of the delay space,
including: TIV, delay dynamics and asymmetry of delays. This
papers makes three contributions.

First, we analytically demonstrate that we can find approximately
nearest servers quickly by iteratively searching closer nodes to
the host using sampled nodes from proximity regions of each node.
However, the analytical method requires a large number of samples,
which does not scale well.

Second, we introduce a novel distributed algorithm, named
HybridNN, that finds nearest service nodes for any machine on the
 Internet (called a target). This algorithm derives from our analytical
method, which preserves the accuracy and speediness of the
analytical method. However, HybridNN has better dynamic adaptation
and reduced measurement costs.

\noindent (i) \textbf{Dynamic adaptation}. A practical DNNS
algorithm needs to proactively maintain moderate service nodes as
samples for DNNS queries, irrespective of the system dynamics.
HybridNN dynamically maintains such neighbors using a concentric
ring used in Meridian \cite{DBLP:conf/sigcomm/WongSS05} or OASIS
\cite{DBLP:conf/nsdi/FreedmanLM06}. However, HybridNN has two
improvements:
\begin{itemize}
    \item The maximum number of nodes stored per ring is derived from the
lower bounds of required samples in the analytical method, which
implies that HybridNN requires the lowest possible number of
samples that has the same accuracy guarantee as the analytical
method.

    \item HybridNN proposes a biased sampling based concentric ring
maintenance scheme, in order to sample enough nodes for each ring.
Specifically, different from previous neighbor discoveries based
on a gossip protocol, we also periodically discover a small number
of nearest nodes and farthest nodes to each node as neighbors in
the concentric ring. This is because given a concentric ring, the
innermost and outermost rings contain only a few neighbors
compared to other rings, which are hardly to be sampled using a
gossip based neighbor discovery protocol.
\end{itemize}

\noindent (ii) \textbf{Reducing measurement costs}. HybridNN
adopts scalable delay predictions to reduce the measurement costs.
\begin{itemize}
    \item HybridNN maintains the concentric rings using estimated
pairwise delays with the revision \cite{DBLP:conf/imc/WangZN07} of
the Vivaldi network coordinate
\cite{DBLP:conf/sigcomm/DabekCKM04}, which significantly reduces
the maintenance overhead of HybridNN compared to Meridian.

    \item HybridNN selects candidate neighbors that are close to the target
using delay predictions. Since delay predictions are only
approximations of real-world delays, HybridNN also uses a small
number of delay probes to avoid being misled by inaccurate delay
predictions. Interestingly, although the network coordinate
distances are symmetric, we empirically find that our hybrid delay
measurement approach provides the accurate nearest next-hop
neighbor for both symmetric and asymmetric delay data sets. This
is because we replace inaccurate coordinate distances with direct
probes using the error indicator of Vivaldi coordinate, which
relieves the mismatch between symmetric coordinate distances and
asymmetric delays.
\end{itemize}

\co{ , which reduce the probing costs of DNNS queries

 Specifically, HybridNN combines estimated delays with a small
number of delay probes, in order to reduce the measurement costs
with delay estimations whenever possible and avoid being affected
by the inaccurate delay estimations. To do that, HybridNN select
samples at each search step using (i) HybridNN selects a small
number of nearest candidate neighbors from those candidate
neighbors whose coordinates are accurate enough; (ii) HybridNN
finds candidate neighbors having erroneous coordinates, using the
error indicators of coordinates provided by Vivaldi; (iii)
HybridNN finds candidate neighbors having erroneous delay
estimations caused by TIV, using a proposed heuristic of locating
TIV using Vivaldi.}

Third, we validate our algorithm using real-world delay data sets
and PlanetLab deployments. Through simulation study, we show that
HybridNN finds servers close to optimal for symmetric and
asymmetric delay data sets. In fact, in more than 95\% of cases,
HybridNN locates the ground-truth nearest servers for the targets.
Furthermore, most queries terminate within four search hops, which
implies that HybridNN can return the search results fast. Using
PlanetLab deployments, we confirm that HybridNN can locate
accurate nearest servers with low query loads and control
overhead, with moderate query time that improves Meridian in more
than 15\% of cases.

\section{System Model}
\label{sysModel}

\subsection{Problem Definition}

In this section, we formally define the nearest server location
problem. Let $V$ denote a set of service nodes and hosts. Let a
distance function $d$ denote the pairwise delays between node
pairs in $V$. Let $N$ be the number of service nodes.

Our objective is to minimize the serving delays of
latency-sensitive applications by finding a service node for a
requesting host with the minimum delay. As discussed in the
previous section, we expect a generalized delay optimization
scenario where the delay may be symmetric or asymmetric according
to the problem context and measurement tools. Furthermore, the
service nodes may be added or removed, which causes system churns.
As a result, we need to locate the service node that is closest to
the target from dynamic service nodes.

We study a distributed approach to realize our objective, since
the centralized approach has several well-known weaknesses,
including: it requires global delay measurements that is hard to
obtain for dynamic service nodes; it incurs the single point of
failures. On the other hand, the distributed approach avoids such
weaknesses through collaborations of service nodes. Specifically,
we formulate the \textbf{Distributed Nearest Neighbor Search}
(DNNS) as:
\begin{defn}
(Distributed Nearest Neighbor Search): For a set of dynamic
service nodes, given any target $T$ on the Internet, the objective
of the \textit{Distributed Nearest Neighbor Search} is to find one
service node that has the smallest delay to $T$, based on the
distributed collaboration of service nodes. \label{DNNSDef}
\end{defn}
The definition of DNNS is not novel, since existing research on
closest server discovery
\cite{1019369,Waldvogel02efficienttopology-aware,DBLP:conf/icdcs/CostaCRK04,DBLP:conf/sigcomm/WongSS05,DBLP:conf/nsdi/FreedmanLM06,DBLP:conf/sigcomm/WendellJFR10}
has formulated the similar problem. Intuitively, DNNS consists of
multiple steps. At each step, a current service node $P$ tries to
locate a new service node that is closer to the target $T$ than
node $P$. The flowchart of a sample DNNS query is shown in
Fig~\ref{fig:illustrationDN2S}. When a host $T$ accesses a
networked service, the local service client module creates a DNNS
query to locate the nearest service machine to the client $T$. The
query message is firstly forwarded to the bootstrap machine of the
DNNS service (Step 1). Then our DNNS query system will forward the
query message recursively until locating a nearest service machine
(Step $2 \rightarrow 3$). Finally, our system returns the contact
addresses of the found service nodes to host $T$ (Step 4).

\co{ However, the DNNS definition serves as the foundation for our
theoretical analysis and algorithm design for the problem of
nearest neighbor search.}

  \begin{figure}[tp]
  \leavevmode \centering \setlength{\epsfxsize}{0.5\hsize}
  \epsffile{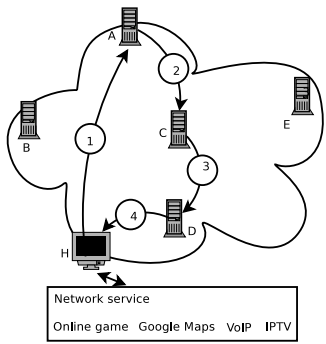}
  \caption{A DNNS query service substrate for network services.}
  \label{fig:illustrationDN2S}
\end{figure}

\subsection{Key DNNS Requirements}

To be useful for latency-sensitive applications, we identify key
goals for the DNNS:
\begin{itemize}
    \item \textbf{Accurate}, we need to find a service node with the lowest interactive
    time in order to increase the Quality of Experience of users.
    \item \textbf{Fast}, we need to obtain the nearest service node with low query periods. Otherwise, long query time makes the DNNS less
    attractive for server redirections in latency-sensitive applications.
    \item \textbf{Scalable}, the DNNS process should incur low bandwidth costs with increasing system size.
    \item \textbf{Resilient to churns}, the DNNS process should find accurate results when the service nodes crash or new service nodes are added.
\end{itemize}

\subsection{Discussion}
\label{cacheDiscuss}

Since the DNNS process may last several seconds due to on-demand
probing, performing DNNS for each query from hosts may even hurt
the Quality of Experience of users, which is significant for small
Web objects. For example, Google typically returns responses in
less than 0.4 seconds; however, such low response periods are
difficult to be realized when applying the DNNS process before
returning the responses.

Therefore, in order to realize a practical nearest server
redirection service, we need to proactively run DNNS for each host
and redirect hosts' requests using cached DNNS results, in order
to achieve millisecond-level response time. For example, OASIS
\cite{DBLP:conf/nsdi/FreedmanLM06} shows that it is feasible to
cache DNNS queries of IP prefixes for server redirections without
reducing the DNNS accuracy.

We do not study how to organize cache results in this paper;
instead, we assume that a DNNS caching service exists to map
hosts' requests to nearest servers using cached DNNS queries. Our
focus is to realize an accurate, scalable and resilient DNNS
system with low DNNS query periods. Since if the DNNS query last
long periods, then crawling DNNS for every IP prefix will be less
efficient.

\co{ In this paper,  cached DNNS results will be obsolete due to
the delay variations and system churns. The completion time of
DNNS queries determine the effectiveness of cached DNNS results, }

\co{ However, we can see that there exist a tradeoff between the
bandwidth costs and the timeliness for cached DNNS results. The
cached results may be obsolete if we 'lazily' perform DNNS queries
within long intervals, due to the delay variations and system
churns. On the other hand, if we frequently perform DNNS queries,
the bandwidth costs increase accordingly.

However, the caching approach is useless if hosts issue DNNS
queries to un-cached targets.}

\section{Related Work}

First, for the theoretical computer science field, research on the
nearest neighbor search mainly focuses on designing efficient
algorithms in the metric space
\cite{Hjaltason:2003:ISS:958942.958948,nn_survey,Chavez:2001:SMS:502807.502808,citeulike:8175150}.
However, applying algorithms in the metric space into DNNS is
inappropriate, since the delay space violates the triangle
inequality that is required by the metric space model
\cite{DBLP:conf/imc/WangZN07}.

On the other hand, for the network system field, research on
nearest neighbor search can be classified into centralized and
distributed approaches according to the communication patterns of
the search process.

\subsection{Centralized Approaches}

The centralized scheme uses a centralized sorting process to
select nearest neighbors for target nodes. However, the
centralized approach does not scale well with increasing system
size, since collecting and transmitting the distance measurements
easily cause performance bottlenecks, which degrades the service
availability.

Guyton et al.~\cite{Guyton95locatingnearby} pioneer the research
on finding the closest server replica in a centralized manner.
They use the Hotz's
 metric \cite{hotzMetric} to represent pairwise hop distances using $O(N)$
measurements to landmark nodes, where $N$ denotes the number of
server replicas. However, smaller hop distances do not mean the
shorter delays, because one hop may pass continents or a data
center. Later Carter and Crovella
\cite{Carter:1997:SSU:839292.843086,Carter19992529} combine the
RTT and available bandwidth measurements to dynamically select
optimal server replica with minimal response time. However, the
dynamic server selection approach does not scale well due to the
quadric measurement costs. Netvigator
\cite{DBLP:journals/ccr/SharmaXBL06} collects RTT values from
hosts to landmarks and milestone nodes based on the Traceroute
measurements, and estimates nearest servers based on local
clustering. However, Netvigator does not guarantee the estimation
accuracy, and may get obsolete results since Netvigator does not
perform active measurements. Different from Netvigator, CRP
\cite{Su:2008:RNP:1439271.1439581} leverage the dynamic
association of nodes with replica servers from CDNs to determine
the proximity between end hosts. CRP incurs low maintenance costs
similar as Netvigator. However, CRP does not guarantee the
accuracy. iPlane
\cite{DBLP:conf/nsdi/MadhyasthaKAKV09,DBLP:conf/osdi/MadhyasthaIPDAKV06}
constructs a synthetic topology structure for the Internet. iPlane
estimates the nearest servers using the approximated delays on the
synthetic topology. However, in order to provide services for
hosts spanning geo-distributed places, iPlane consumes heavy
bandwidth costs to perform active measurements.

\co{by measurements from the landmark machines on the PlanetLab}

\subsection{Distributed Approaches}

The DNNS approach iteratively selects closer nodes using
distributed nearest neighbor search by local measurements towards
a small set of neighbors, which reduces the network measurement
overhead and is more scalable than the centralized approach.
Existing DNNS methods fall into four families based on their
search rules: (i) Bin based DNNS; (ii) Topology based DNNS; (iii)
Greedy search based DNNS; (iv) Ring search based DNNS.

\textbf{Bin based DNNS}. Ratnasamy et al.~\cite{1019369} assign
nodes into "bins" based on the ordered sequence of RTT
measurements to landmarks, and declare nodes are close to each
other in the same bin. However, the bin approach does not
guarantee the accuracy, and fails when the landmarks crashes.

\textbf{Topology based DNNS}. Tiers \cite{Banerjee02scalablepeer}
locates the nearest nodes by a top-down approach with a
hierarchical clustering tree, but may cause load imbalance for
nodes near the root of the tree. Besides, Tiers do not guarantee
the search accuracy since the tree does not strictly preserve the
pairwise proximity.

\textbf{Greedy search based DNNS}. Mithos
\cite{Waldvogel02efficienttopology-aware} iteratively locates
proximate neighbors with $O(N)$ hops by a gradient descent based
protocol in the overlay construction, but terminates earlier
before locating the real nearest nodes due to the limited
diversity in the neighbor set. PIC
\cite{DBLP:conf/icdcs/CostaCRK04} iteratively locates nearest
neighbors at each search step in terms of the coordinate distance.
However, PIC is prone to be trapped into the local minima since
the coordinate distance only approximates the delays. DONAR
\cite{DBLP:conf/sigcomm/WendellJFR10} redirects host requests to
optimal server relicas by considering the network proximity, the
routing optimization and server loads. DONAR uses geographic
distances as the proximity metric in order to reduce measurement
costs. However, DONAR may find suboptimal server replicas for
delay minimizations since the delay values are not consistent with
the geographic distances.

\textbf{Ring search based DNNS}. Our work is closely related to
Meridian \cite{DBLP:conf/sigcomm/WongSS05}, which seeks
approximately nearest nodes in $\log \left( N \right)$ steps.
Meridian \cite{DBLP:conf/sigcomm/WongSS05} maintains a loosely
connected overlay using a gossip based peer finding scheme. The
neighbors are organized in concentric rings with exponentially
increasing radii. For a DNNS request, Meridian iteratively locates
one next-hop node that is $\beta$ ($\beta <1$) times closer to the
target $T$ than the current Meridian node. Compared to other
families of DNNS, Meridian is more accurate by using rings of
neighbors that promote the diversity of neighbor sets
\cite{DBLP:conf/sigcomm/WongSS05}. However, several studies have
identified that Meridian may fail to find the closest service node
due to the last-hop clustering of servers
\cite{DBLP:conf/imc/VishnumurthyF08}, and TIV of the network delay
space \cite{DBLP:conf/imc/WangZN07}. Similar as Meridian, OASIS
\cite{DBLP:conf/nsdi/FreedmanLM06} organize neighbors as
concentric rings for each service node, and iteratively search
nearest service node for the request host in terms of the
geographic distances. OASIS reduces the delay measurement costs in
Meridian through the static geographic coordinates, and has low
response time using in-advance probes. However, OASIS does not
guarantee the accuracy of the search results, since selecting the
geographically closest servers may incur high delays
\cite{DBLP:conf/imc/KrishnanMSJKAG09}.

To address these problems, two adjustments are proposed:  (i)
explicitly finding the clustering subsets based on the structure
of IP addresses \cite{DBLP:conf/imc/VishnumurthyF08} or, (ii)
adding additional neighbors for DNNS that may not be chosen due to
the TIVs \cite{DBLP:conf/imc/WangZN07}. However, finding the
clusters of nodes sharing identical last hops becomes insufficient
when the service nodes spread over nearby subnets, which may still
mislead the DNNS queries due to no forwarding nodes closer enough
to the target. Furthermore, finding all neighbors that are
affected by the TIVs is challenging since calculating the TIVs for
decentralized service nodes is very difficult; besides, adding
additional neighbors for DNNS also increases the query overhead.
Due to the limitations of modifications for Meridian, significant
challenges remain in DNNS. We focus on tackling these challenges
in this paper.

\section{Data Sets}
\label{datset}

Our empirical data sets include four publicly available real-world
RTT data sets, covering the delay measurements between wide-area
DNS servers and those between end hosts \cite{HybridNNReport}.
(i) \textbf{DNS3997}. A RTT matrix collected between 3997 DNS
servers by Zhang et al.~\cite{DBLP:journals/ton/ZhangNNRDW10}
using the King method \cite{Gummadi2002_637203}. The matrix is
symmetric in that $d_{ij}=d_{ji}$, for any pair of items $i$ and
$j$, where $d$ denotes the delay matrix. (ii) \textbf{Host479}. A
RTT delay matrix based on RTT measurements that last 15-day
periods between the Vuze BitTorrent clients
\cite{DBLP:conf/infocom/ChoffnesSB10}. Host479 is asymmetric,
where in over 40\% of the cases delay pairs $d_{AB}$ and $d_{BA}$
in Host479 differ more than 4 times. This is because RTT
measurements between node pairs are not synchronized and delay
results are affected by varying queueing delays at end hosts
\cite{DBLP:conf/infocom/ChoffnesSB10}. (iii) \textbf{DNS1143}. A
RTT matrix between 1143 DNS servers collected by the MIT P2PSim
project \cite{P2PSimData} using the King method
\cite{Gummadi2002_637203}. The matrix is symmetric in that
$d_{ij}=d_{ji}$, for any pair of items $i$ and $j$, where $d$
denotes the delay matrix. (iv) \textbf{DNS2500}. A RTT matrix
between 2500 DNS servers by the Meridian project
\cite{DBLP:conf/sigcomm/WongSS05} using the King method. The
matrix is also symmetric.

\co{ which cause the RTT measurement from host $A$ to host $B$ at
time $t_0$ to be different from that in the reverse direction at
time $t_1$ ($t_0 \neq t_1$) (missing percentage 13.77\%), (missing
percentage 5\%) }

Since obtaining the one-way delays between large-scale nodes is
extremely difficult, we use Host479 as an asymmetric delay data
set. However, we do not claim that our experiments on Host479 are
the same as those on the one-way delay metric.

\co{ Future detailed study on the one-way delay data sets are
required. measurement facility to show the generalization of our
DNNS methods on the  due to the time synchronization

 Finally, due to space limits, we only report results on
DNS3997 and Host479 data sets. Empirical results on DNS1143 and
DNS2500 are similar. A few words about the notations. }


\section{A Generalized Delay Model for the Delay Space}
\label{RIMSec}

In this section, we present a simple and general enough delay
model for the delay space. Our model captures the important
characteristics of the delays, including TIV, dynamics and
asymmetry of RTTs and OWDs. In the next section, we will analyze
the DNNS problem on our model.

Assuming that we select a node $P$ in $V$ as the center of a ball,
and choose a positive real number $r$ as the radius of the ball,
then we call a \textbf{closed ball} $B_P(r)$ as the set of nodes
whose delays to node $P$ are not larger than $r$, i.e.,
$B_P(r)=\left\{ {v |d(P,v) \le r, P,v \in V} \right\}$.
Furthermore, the \textbf{volume} of a ball is the number of nodes
covered by the ball. Besides, we define the cover relation of
different set of nodes as follows:
\begin{definition}[Cover]
Let $S$ and $\Omega$ be two sets of nodes, if $\Omega \subseteq
S$, then the set $S$ is said to \textit{cover} the set $\Omega$.
\label{coverDef}
\end{definition}

\co{provides us to analyze the DNNS faithfully to the ground-truth
delays; s in order to analyze the feasibility of DNNS}

\subsection{Definition}

We first state the requirements for a delay model suitable for
RTTs and OWDs used for delay minimizations. (i) The delay model
should relax the symmetry requirements, since the OWDs are
asymmetric due to routing asymmetry \cite{1577769}. Besides,
although RTT is symmetric by only accounting for the delays on the
routing paths, real-world RTT measurements may be asymmetric due
to variations of queueing delays at end hosts or un-synchronized
measurements \cite{Choffnes:2010:PTE:1764873.1764880}. (ii) The
delay model $d$ should allow TIV to exist, since the RTT metric
exhibits TIV \cite{Lumezanu:2009:TIV:1644893.1644914}. (iii) The
delay model $d$ should allow dynamic delays, since the delay
varies from time to time \cite{Paxson:1997:EIP:263105.263155}.

Therefore, inspired by the inframetric model
\cite{DBLP:conf/infocom/FraigniaudLV08} that allows the TIVs, we
extend the inframetric model to a relaxed inframetric model that
relaxes the symmetry requirement, where the distance function $d$
satisfies:

\begin{definition}[Relaxed Inframetric Model] Let a distance function $d: V \times V \to {\Re ^{+}}$ be a relaxed $\rho$-inframetric ($\rho >1$), if $d$ satisfies the following conditions for any pair of nodes $u$ and $v$: (1) if $d(u,v)$=0, then $u$=$v$; (2) $d(u,v) \le \rho \max \left\{ {d(u,w), d(v,w)} \right\}$, for any arbitrary node $w$ satisfying $w \notin \left\{ {u,v} \right\}$.
\label{directionInframetric}
\end{definition}

\textbf{Pros of the Relaxed Inframetric Model}: The condition (2)
in Def \ref{directionInframetric} states a generalized relation of
any directed triple from $V$, which has two beneficial properties:
\begin{itemize}
    \item \textbf{TIV-adaptive}. Intuitively, smaller $\rho$ implies that
three edges are closer to each other; while larger $\rho$ implies
that one edge is significantly larger than any of the other two
edges, which may introduce a TIV. Therefore, similar as the
inframetric model, the relaxed inframetric model naturally allows
the occurrence of TIVs.

\item \textbf{Dynamics-adaptive}.  The inframetric model allows
the delay variations by varying the inframetric parameter $\rho$
to describe the relations of updated triples. Therefore, both
inframetric model and the relaxed inframetric model are able to
model variations of triples due to delay variations.

\item \textbf{Asymmetry-aware}. The relaxed inframetric model
allows the asymmetry in the delay space, which generalizes to RTTs
and OWDs. As a result, we are able to analyze DNNS on symmetric
and asymmetric delays through the relaxed inframetric model.

\end{itemize}

\co{ Therefore, we can use the inframetric model to characterize
the dynamic delay space. }

Having shown the advantages of the relaxed inframetric model, next
we discuss the statistical property of the inframetric parameter
$\rho$.

First, the seminal work states that if the delay space obeys the
triangle inequality, then $\rho$ must be smaller or equal than 2
\cite{DBLP:conf/infocom/FraigniaudLV08}. However, when $\rho$ is
smaller than 2, there may exist TIVs. For example, given a triple
with pairwise RTTs ${3,1,1.8}$, we can see that the inframeter
parameter $\rho$ is approximately 1.67 but there also exists a TIV
in the triple. Therefore, we can see that $\rho \leq 2$ is only a
\textit{necessary} but not a \textit{sufficient} condition for no
TIVs.

Second, we find that the inframetric parameter $\rho$ is quite low
for most triples. First, the 95th percentiles of all data sets of
$\rho$ are below $2.5$. Low inframetric parameter $\rho$ means the
largest edges in triples are not too much larger than the other
edges of the triples. Second, among the triples whose $\rho$ are
bigger than 2, their $\rho$ values are around 3 on average.
Therefore, selecting $\rho$=3 is reasonable to model most of the
triples.

\subsection{Dimensions on the Relaxed Inframetric Model}
\label{growthDimensionSec}

Having introduced the definition of the relaxed inframetric model,
now we analyze the growth dimension of the relaxed inframetric
model, which is the ratio of the number of nodes covered by two
closed balls with the identical center and varying radii
\cite{DBLP:conf/infocom/FraigniaudLV08,510013}.

The growth dimension is important for efficient DNNS. As shown by
Karger and Ruhl \cite{510013}, assuming that the growth dimension
is low, each node $P$ can uniformly sample a modest number of
nodes to locate a node that is closer to any other node in $V$.
Therefore, we can recursively find nodes closer to the target
based on the above sampling procedure, which helps the design of
the DNNS algorithms.  However, since Karger and Ruhl assumes the
triangle inequality to hold \cite{510013}, we cannot immediately
apply their DNNS results into the relaxed inframetric model.
Accordingly, we need new proof techniques for DNNS analysis.

The growth dimension for the inframetric space
\cite{DBLP:conf/infocom/FraigniaudLV08} is defined as follows:
\begin{definition}[\textit{Growth \cite{DBLP:conf/infocom/FraigniaudLV08}}]For a $\rho$-inframetric model, for any $r \in {\Re ^{+}}$ and $P \in V$, if $\left| {{B_P}\left( {\rho r} \right)} \right| \le {\gamma _g}\left| {{B_P}\left( r \right)} \right|$, where ${\gamma _g} \in {\Re ^{+}}$, the $\rho$-inframetric model is said to have a growth ${\gamma _g} \ge 1$.
\label{growthDef}
\end{definition}

\co{
\begin{definition}[\textit{Doubling \cite{DBLP:conf/infocom/FraigniaudLV08}}]For a $\rho$-inframetric model, for any $r \in {\Re ^{+}}$ and $P \in V$, if ${B_P}\left( {\rho r} \right) \subseteq { \cup _{{v_i} \in I}}{B_{{v_i}}}\left( r \right)$, where $I \subseteq V$, $\left| I \right| \le {\gamma _d}$ and ${\gamma _d} \in {\Re ^{+}}$, the $\rho$-inframetric model is said to have a doubling ${\gamma _d} \ge 1$.
\label{doublingDef}
\end{definition}
}

The growth dimension ${\gamma _g}$ on the inframetric model
generalizes the growth definition in the metric space which
assumes the triangle inequality to hold
\cite{510013,DBLP:journals/ton/ZhangNNRDW10}. Therefore, the
growth ${\gamma _g}$ inherits the intuitive meanings of the growth
definition in the metric space. Specifically,  low growth ${\gamma
_g}$ means that the number of nodes covered by the closed ball
$B_P(\rho r)$ is comparable to the number of nodes covered by the
closed ball $B_P(r)$. Therefore, when we expand a ball around a
node $P \in V$, we can see that new nodes in $V$ "come into view"
at a constant rate \cite{510013}.

\co{ However, since the DNNS results for the growth dimension in
the metric space depends on the triangle inequality to hold
\cite{510013}, we cannot immediately apply these DNNS results into
the relaxed inframetric model. Accordingly, we need new proof
techniques. }

Finally, based on Def \ref{growthDef}, the infimum of the growth
dimension ${\gamma _g}$ equals the ratio of the volume between
$B_P(\rho r)$ and $B_P(r)$ for any node $P$ and radius $r$. Since
we are interested in the infimum, when we refer to the growth of
the inframetric space, we mean the infimum accordingly.

\co{ For example, the delay space for nodes in network
applications such as Internet or Peer-to-Peer networks has low
growth dimension.}

 \co{ Similarly, the doubling dimension is the number
of balls of radius $r$ whose union covers the ball $B_P(\rho r)$.
The doubling dimension generalizes the growth dimension
\cite{DBLP:conf/infocom/FraigniaudLV08}.

The doubling dimension has an intuitive explanation for designing
accurate DNNS algorithms. If the doubling dimension is low, then
we can sample a small number of nodes to ensure that there exists
at least one node and the target are covered by a small-radius
ball, i.e., we can find a node that is closer to the target than
the current node. Accordingly, we can recursively find nodes
closer to the target based on the above sampling procedure, which
helps the design of the DNNS algorithms.

Besides, we can see that the growth has the infimum, but do not
constrain the maximal values. Consequently, from now on, when we
refer to the growth of the inframetric space, we mean the infimum.
}

Next, we empirically evaluate the growth dimension of the delay
space with respect to the radius $r$ and the inframetric parameter
$\rho$. Our evaluation complements the seminal work on the growth
in the inframetric model \cite{DBLP:conf/infocom/FraigniaudLV08}
using symmetric and asymmetric data sets. Recall that computing
the growth is trivial by comparing the volumes of the balls with
identical centers and varying radii.

Fig~\ref{fig:MedGrowth} shows the median and 90th percentile
growth values for varying radii. The median growth of most data
sets is relatively small, and declines quickly with increasing
radii for most data sets except for Host479. For Host479, the
median growth may increase as the radii increase. On the other
hand, the 90th percentile growth shows divergent dynamics for
different data sets, revealing "\textit{M}"-shape dynamics,
indicating that a small fraction of growth values may increase or
decrease with increasing radii.

Furthermore, by selecting different percentages of nodes for the
statistics, Fig~\ref{fig:MedGrowth} shows that the median growth
is less sensitive to the sample size compared to the magnitudes of
radii; while the 90th percentile growth becomes relatively more
sensitive to the sample size.

\begin{figure}[tp]
     \centering
           \subfigure[DNS1143.]
        {
          \setlength{\epsfxsize}{.44\hsize}
          \epsffile{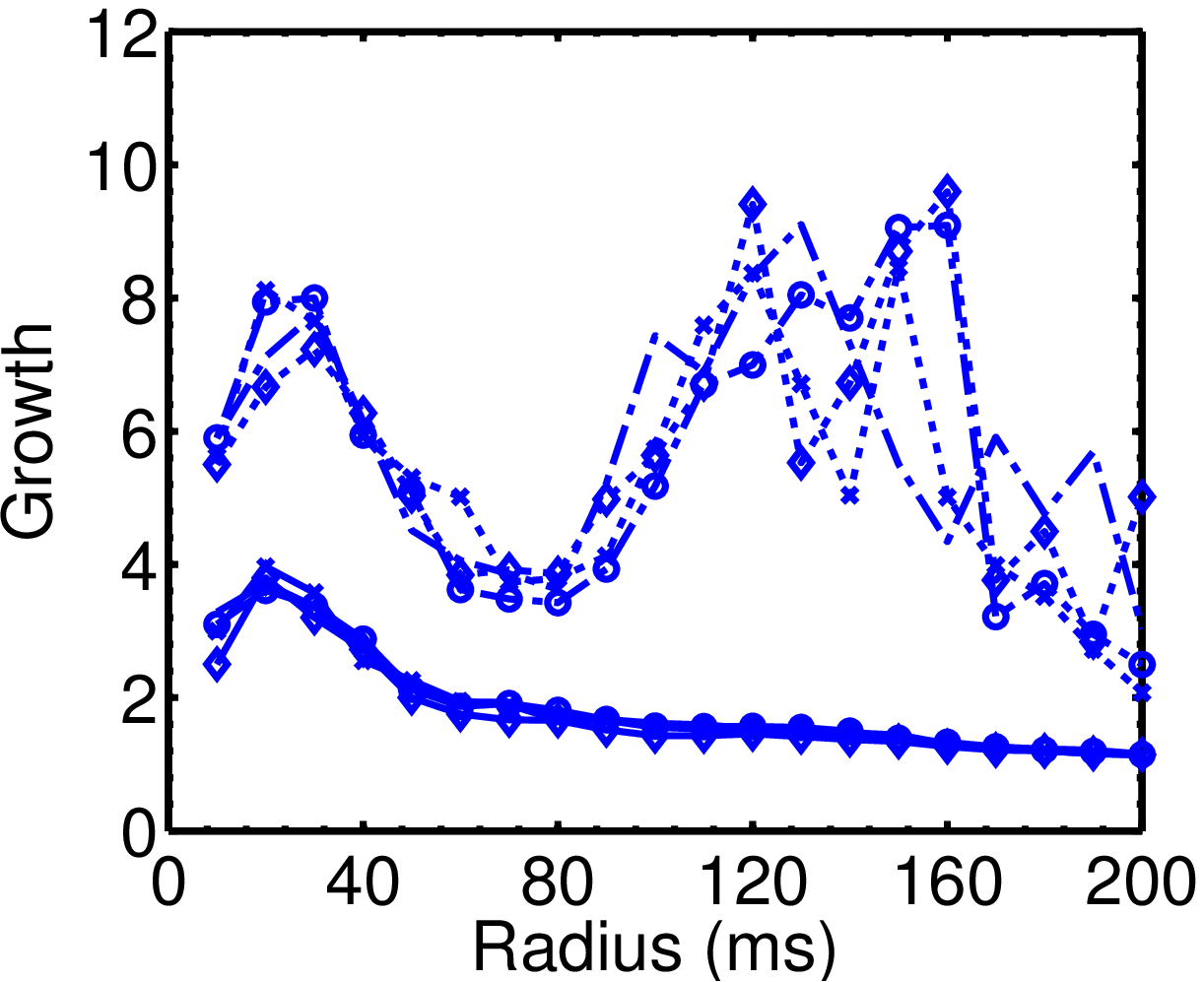}
         }
               \subfigure[DNS2500.]
        {
          \setlength{\epsfxsize}{.44\hsize}
          \epsffile{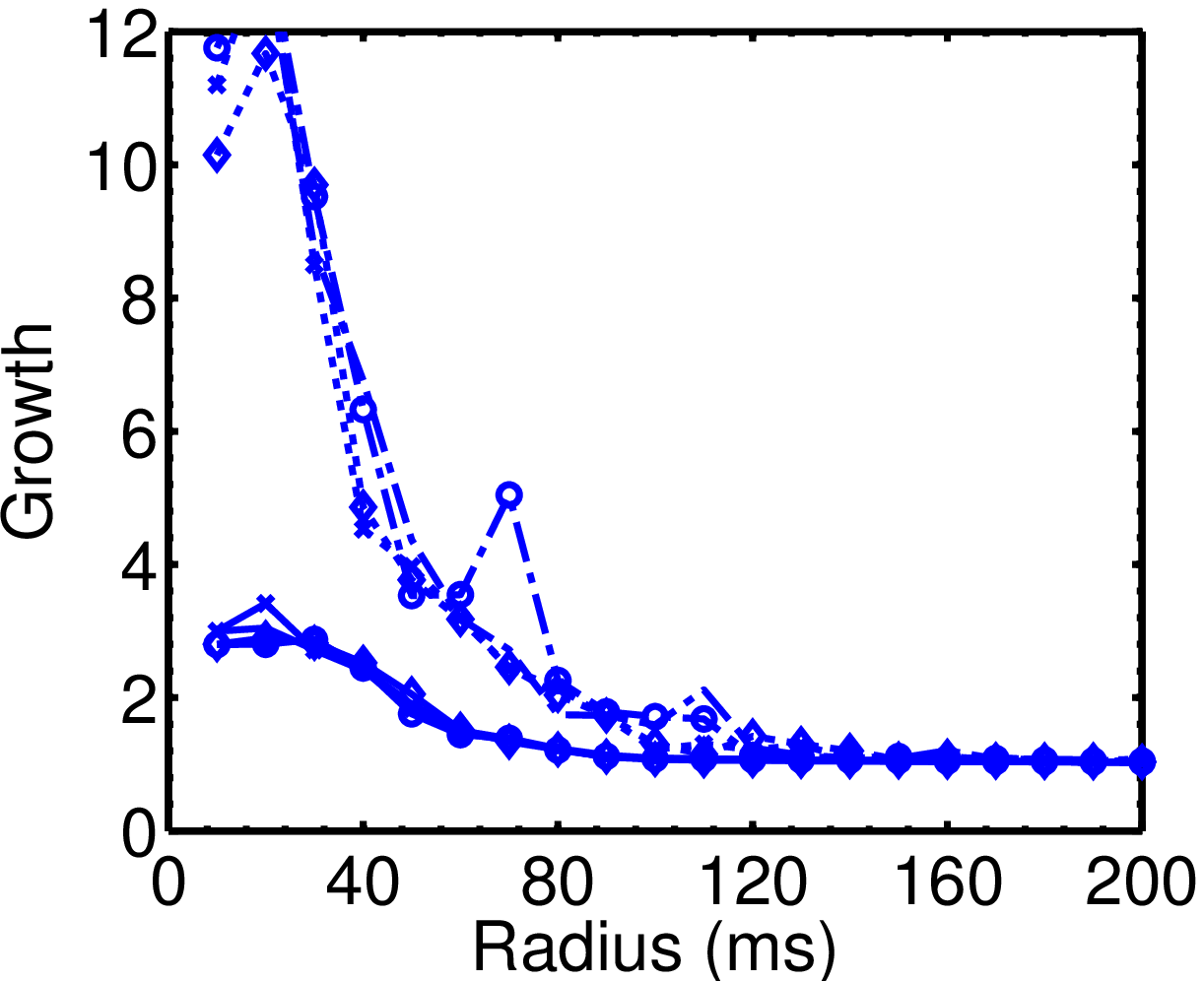}
         }
               \subfigure[DNS3997.]
        {
          \setlength{\epsfxsize}{.44\hsize}
          \epsffile{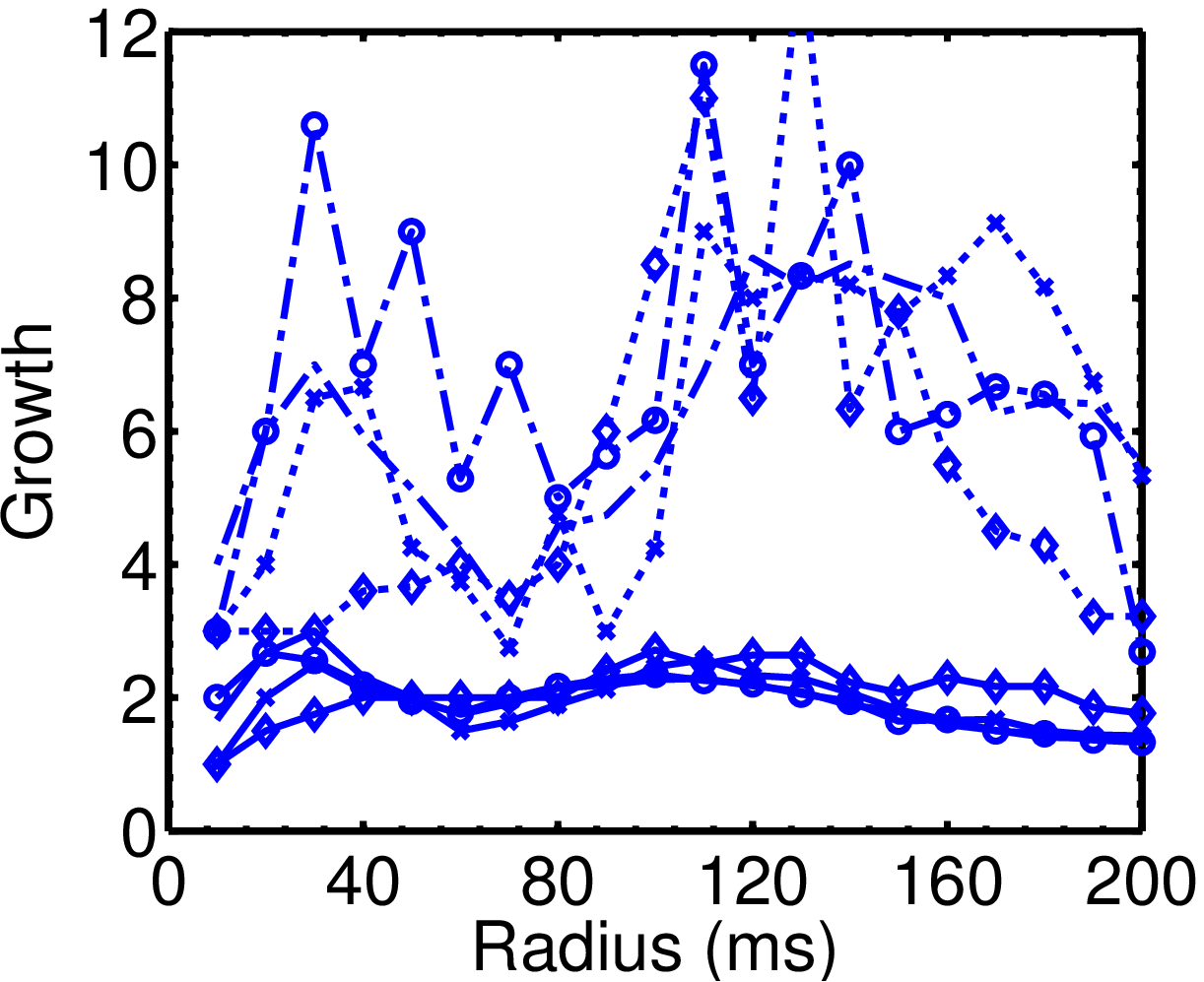}
         }
              \subfigure[Host479.]
        {
          \setlength{\epsfxsize}{.44\hsize}
          \epsffile{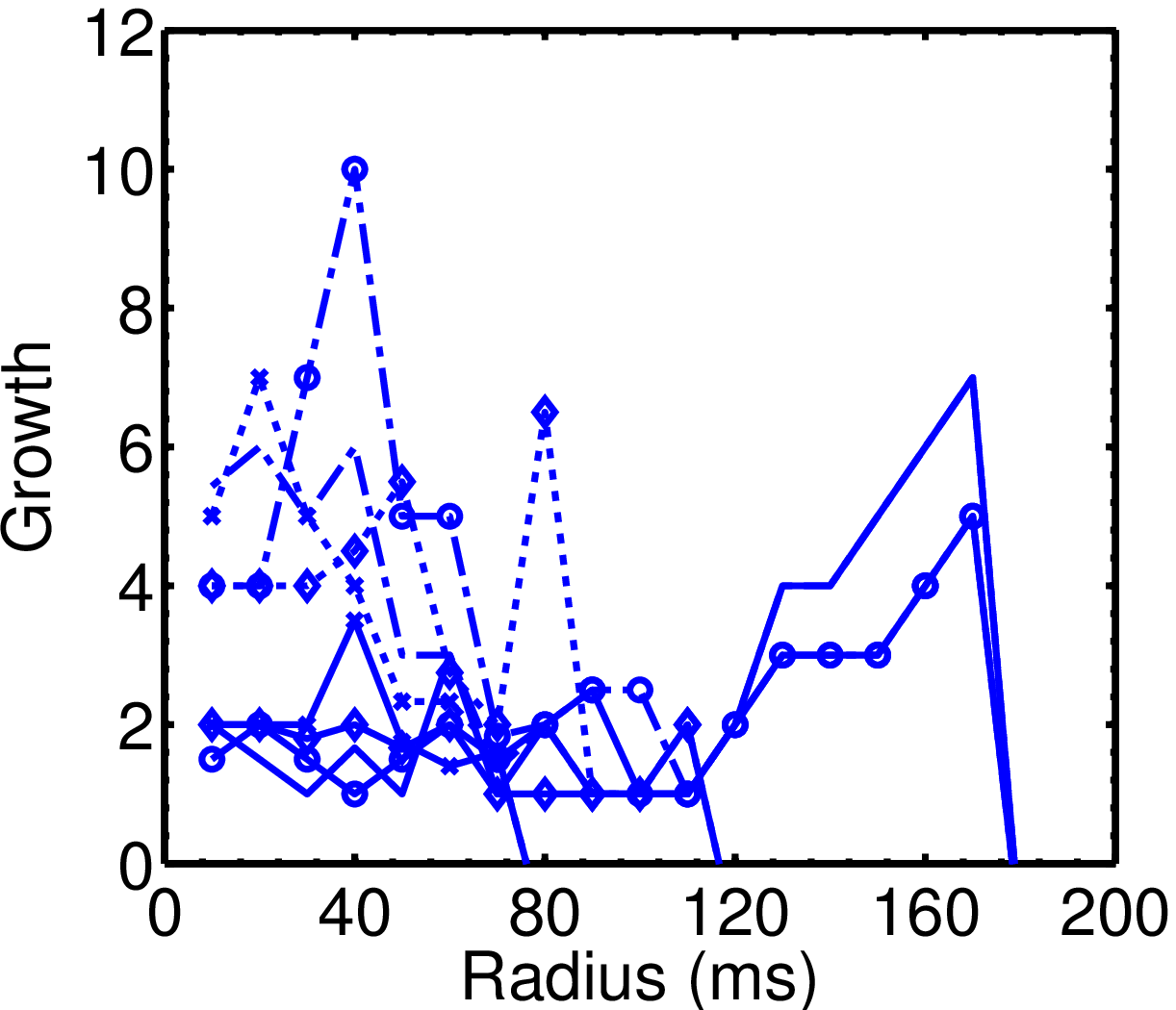}
         }

\caption{The statistics of the median and 90-th percentile growth
${\gamma _g}$ for $\rho =3$; \texttt{-$\diamondsuit$-} denotes
median values computed from sampled 20\% nodes; \texttt{-x-}
denotes median values computed from sampled 50\% nodes;
\texttt{-o-} denotes median values computed from sampled 75\%
nodes; - represents median values computed from  all nodes;
\newline \texttt{$\cdots$$\diamondsuit$$\cdots$} denotes
90-percentile values computed from sampled 20\% nodes;
\texttt{$\cdots$x$\cdots$} denotes 90-percentile values computed
from sampled 50\% nodes; \texttt{-.$o$-.} denotes 90-percentile
values computed from sampled 75\% nodes; \texttt{-.-} represents
90-percentile values computed from  all nodes.}
\label{fig:MedGrowth}
\end{figure}

\co{ Fig~\ref{fig:growthSta} shows the median of the growth as a
function of the variations of the inframetric parameter $\rho$ and
the radius $r$. Generally, the median growth values decrease
quickly with increasing radii for a fixed inframetric parameter
$\rho$ value. When the radii are smaller than 40, the growth
values increase as the inframetric parameter $\rho$ increases;
after the radii exceed 40 ms, the growth does not increase
significantly with increasing $\rho$.}

In summary, the growth metric ${\gamma _g}$ of the delay space is
quite low. Furthermore, with increasing radius, the growth
${\gamma _g}$ decreases to 2 quickly on average. However,
sometimes the growth values increase for increasing radius, which
means that there are many nodes that have similar distances to
each other. This usually corresponds to cases where the center of
the ball is a node on the edge of a cluster, where nodes in the
same cluster have smaller distances compared to those to other
nodes not in the same cluster.

\section{Efficient DNNS on the Relaxed Inframetric Model}
\label{theoryDNNS}

In this section, using the relaxed inframetric model presented in
Sec \ref{RIMSec}, we analyze how to design an efficient DNNS using
localized operations suitable for distributed systems. Proofs are
omitted due to space limits, which can be found in the full report
\cite{HybridNNReport}.

Our major result is that it is feasible to design an accurate and
fast DNNS algorithm for the relaxed inframetric mode, at the
expense of sampling enough candidate servers from the proximity
region of each node. We construct a simple DNNS process satisfying
our major result. However, the simple DNNS process incurs
relatively high measurement costs due to the sampling conditions,
which will be improved in the next section.

\subsection{Sampling Conditions to Locate Closer Nodes To Targets}
\label{alphaCalc}

In this section, We analyze samples required to locate a node
closer to a target than the current node based on the growth
dimension in Sec \ref{growthDimensionSec}. The sampling conditions
serves as the basis for the efficient DNNS algorithmic design.

Our results show that we can sample a server closer to the target
using bounded samples at each node. In order to obtain a node that
is $\beta$ ($\beta \in (0,1]$) times closer to the target than the
current node, we need to uniformly sample enough neighbors from
the proximity region of each current node.

Without loss of generality, assume that a node $P$ needs to locate
a node $Q$ that is $\beta$ ($\beta \le 1$) times closer to a
target $T$, which implies that $d_{QT} \le {\beta \times d_{PT}}$.
Let $d_{PT} = r$. We can see that node $Q$ must be covered by the
ball ${B_P}\left( {\rho r} \right)$, since ${d_{PQ}} \le \rho \max
\left\{ {{d_{PT}},{d_{QT}}} \right\} = \rho r$.
Fig~\ref{fig:growthIllust} shows an example of sampling a node
closer to the target $T$ in the closed ball ${B_P}\left( {\rho r}
\right)$ in the growth dimension.

\begin{figure}
  \leavevmode \centering \setlength{\epsfxsize}{.44\hsize}
  \epsffile{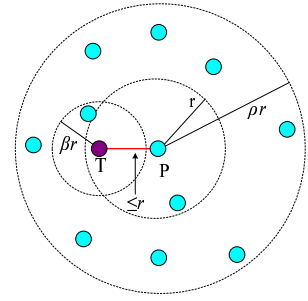}
  \caption{Sampling closer nodes to a target $T$ from ${{B_P}\left( {\rho r} \right)}$ in the $\rho$-inframetric model with growth ${\gamma _g}$.}
  \label{fig:growthIllust}
\end{figure}

We first quantify the volume differences of balls with identical
centers but different radii.
\begin{lemma}
Given a $\rho$-inframetric with growth ${\gamma _{g}}\geq1$, for
any $x\geq\rho$, $r>0$ and any node $P$, the volume of a ball
$B_P(r)$ is at most $x^{\alpha}$ smaller than that of the ball
$B_P(xr)$, where ${\log _\rho }{\gamma _{g}}  \le \alpha  \le
2{\log _\rho }{\gamma _{g}}$. \label{grid_growth}
\end{lemma}

\co{
\begin{lemma}
Given a $\rho$-inframetric space with doubling ${\gamma _{d}} \geq
1$, for any $x\geq\rho$, $r>0$ and any node $P$, the volume of a
ball $B_P(r)$ is at most $x^{\alpha}$ smaller than that of the
ball $B_P(xr)$, where $\frac{1}{4}\log _\rho {\gamma _{d}} \le
\alpha \le \log _\rho N$, and $N$ is the total number of nodes.
\label{grid_doubling}
\end{lemma}
}

Lemma \ref{grid_growth} states that the volume differences of the
balls with identical centers and different radii are bounded by
$x^{\alpha}$, where $x$ is the multiplicative ratio between
different radii, and the parameter $\alpha$ lies in a bounded
interval.

We calculate $\alpha$ by varying the radius $r$ and the
multiplicative ratio $x$ as shown in Fig \ref{fig:alphaStat}. We
can see that $\alpha$ is mostly below 1, and decreases close to 0
quickly with increasing radius $r$ or multiplicative ratio $x$.
Therefore, the volume difference $x^{\alpha}$ scales
\textbf{sub-linearly} in most cases. On the other hand, for small
radius $r$ or low multiplicative ratio $x$, the volume difference
$x^{\alpha}$ may scale \textbf{ultra-linearly}.

\begin{figure}[tp]
     \centering
      \subfigure[DNS1143.]
        {
          \setlength{\epsfxsize}{.44\hsize}
          \epsffile{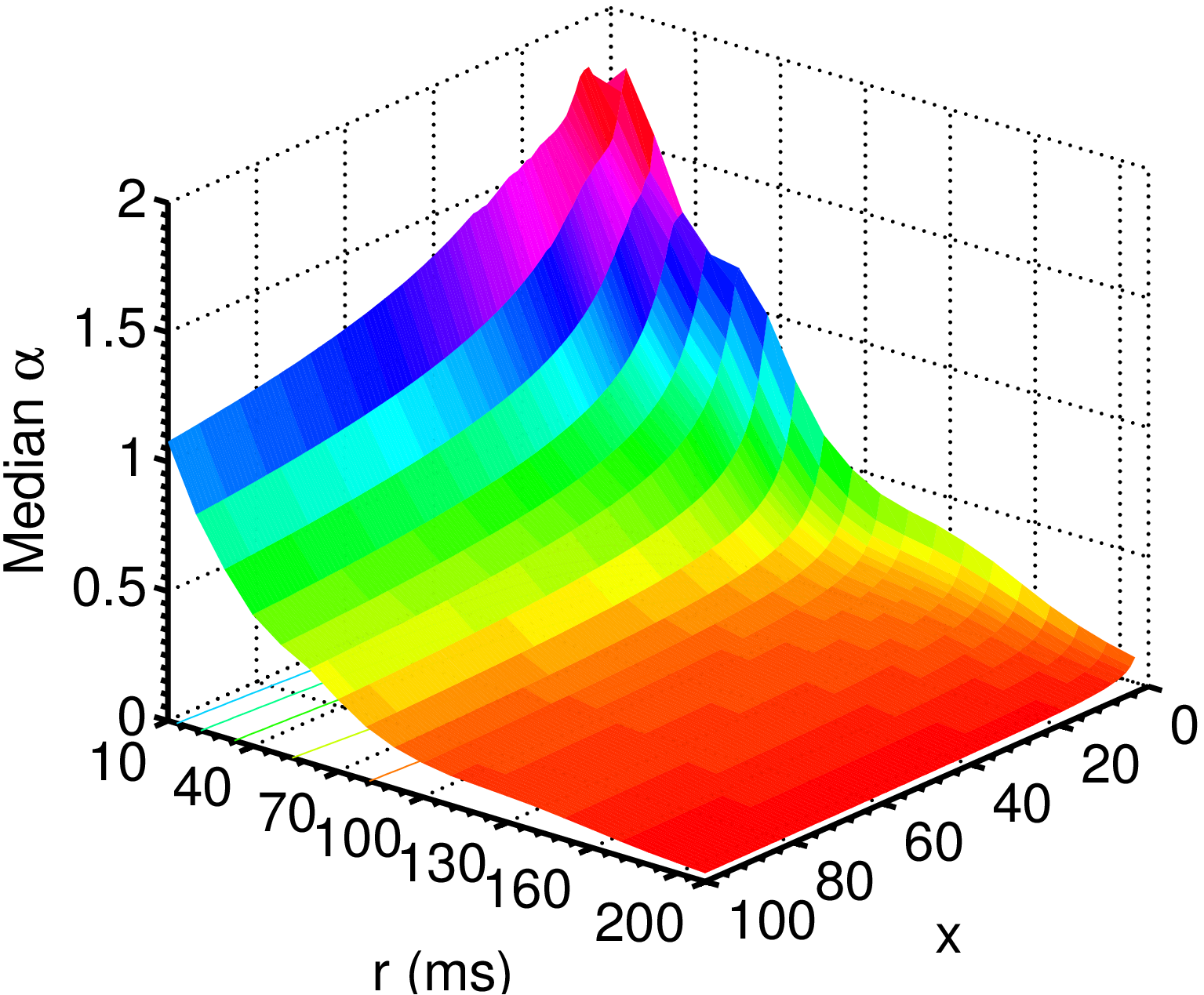}
         }
               \subfigure[DNS2500.]
        {
          \setlength{\epsfxsize}{.44\hsize}
          \epsffile{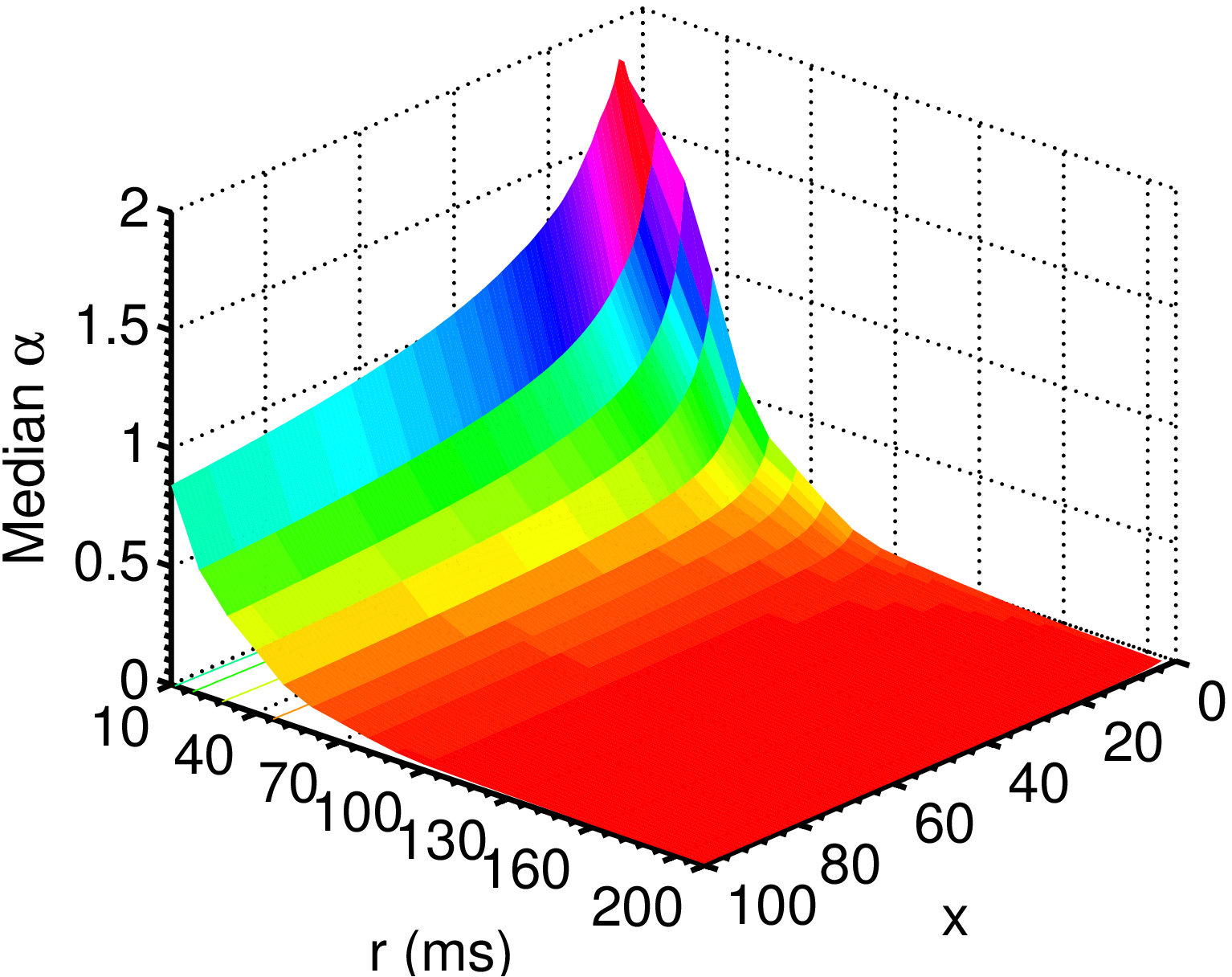}
         }
               \subfigure[DNS3997.]
        {
          \setlength{\epsfxsize}{.44\hsize}
          \epsffile{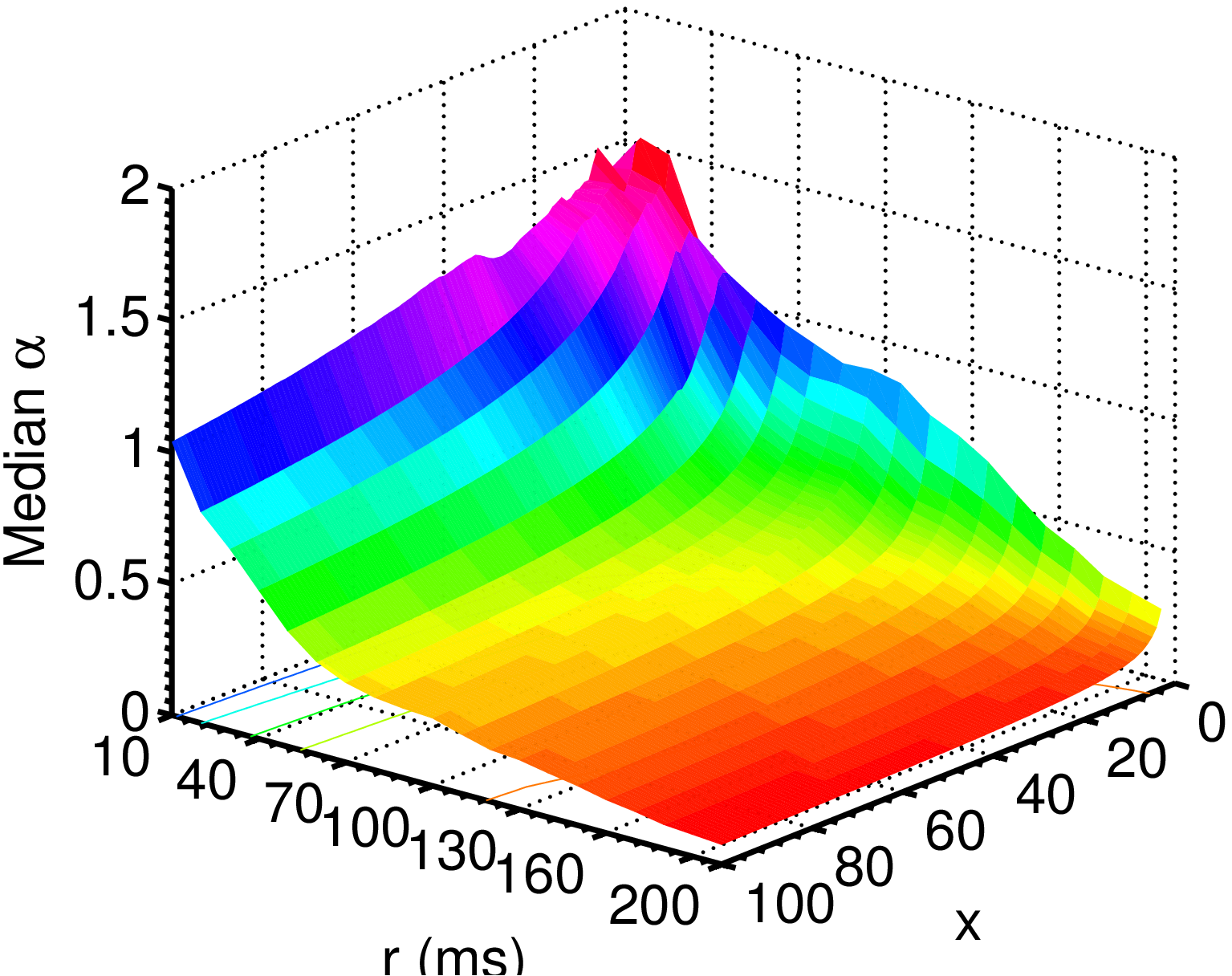}
         }
                 \subfigure[Host479.]
        {
          \setlength{\epsfxsize}{.44\hsize}
          \epsffile{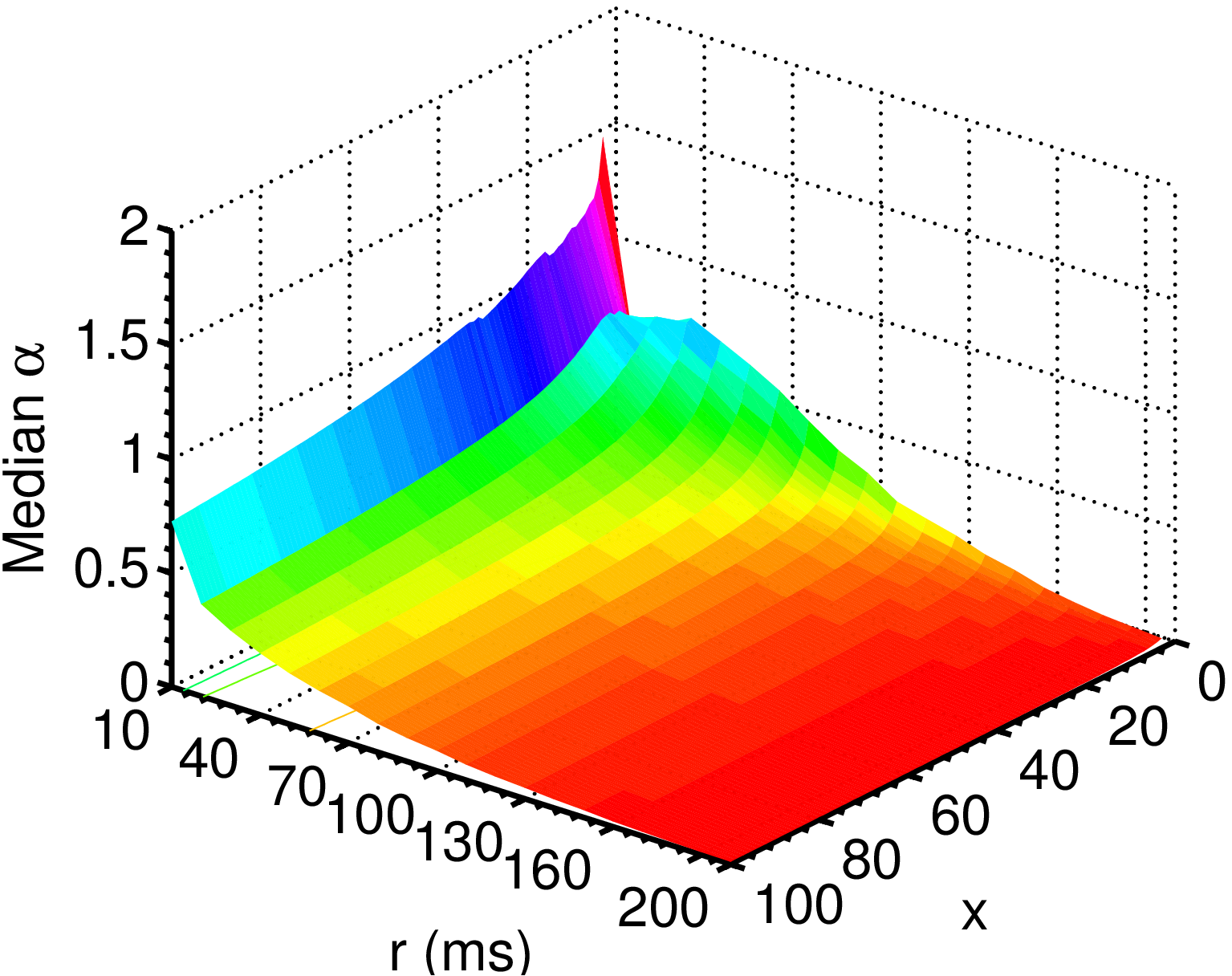}
         }
     \caption{Median $\alpha$ as function of the radius $r$ and the multiplicative ratio $x$.}
     \label{fig:alphaStat}
\end{figure}

Furthermore, we also characterize the inclusion relation of balls
with different centers, which generalizes the inclusions of balls
around a node pair in the metric space \cite{510013}. Lemma
\ref{sandwich} lays the foundation for uniform sampling nodes to
perform DNNS on the inframetric model.

\begin{lemma}
(Sandwich lemma) For any pair of node $p$ and $q$, and $d_{pq}
\leq r$, then
\[{B_q}\left( r \right) \subseteq {B_p}\left( {\rho r} \right) \subseteq {B_q}\left( {{\rho ^2}r} \right)\]
\label{sandwich}
\end{lemma}

Using Lemma \ref{grid_growth} and \ref{sandwich}, we can quantify
the size of sampled neighbors, to assure that at least one
neighbor lies in the closed ball ${B_T}\left( {\beta r} \right)$.

\begin{theorem}
(Sampling efficiency in the growth dimension) For a
$\rho$-inframetric model with growth ${\gamma _g} \geq 1$, for a
service node $P$, and a DNNS target $T$ satisfying ${d_{PT}} \le
r$, when selecting $3{\left( {\frac{{{\rho ^2}}}{\beta }}
\right)^\alpha }$ nodes uniformly at random from ${B_P}\left(
{\rho r} \right)$ with replacement, with probability of at least
95\%, one of these nodes will lie in ${B_T}\left( {\beta r}
\right)$, where ${\log _\rho }{\gamma _g} \le \alpha  \le 2{\log
_\rho }{\gamma _g}$ and $\beta<1$. \label{thm:SamplingGrowth}
\end{theorem}

Since $\alpha$ and $\rho$ are determined by the delay space, we
can see that the number of samples decreases with increasing delay
reduction threshold $\beta$. As $\beta$ approaches 1, the number
of required samples becomes approximately $3{\left( {\frac{{{\rho
^2}}}{\beta }} \right)^\alpha } \approx 3{\rho ^{2\alpha }} \in
\left[ {3\gamma _g^2,3\gamma _g^4} \right]$ based on Lemma
\ref{grid_growth}.

\co{ , whose bounds are the power functions of the growth ${\gamma
_g}$

Theorem \ref{thm:SamplingGrowth} says that

 Furthermore, since the lower bound of $\alpha$ increases with
increasing ${\gamma _g}$, the required number of nodes to sample
also goes up with increasing of ${\gamma _g}$.

As a result, since ${\gamma _g}$ and $\rho$ are determined by the
delay space, in order to reduce the sampling overhead, the delay
threshold $\beta$ should not be too small. As $\beta$ approaches
1, the number of required samples becomes approximately $3{\left(
{\frac{{{\rho ^2}}}{\beta }} \right)^\alpha } \approx 3{\rho
^{2\alpha }} \in \left[ {3\gamma _g^2,3\gamma _g^4} \right]$,
whose bounds are the power functions of the growth ${\gamma _g}$.
}

\co{ of uniformly sampling nodes in ${B_P}\left( {\rho r}
\right)$}

\co{Nevertheless,  And finding a tighter interval with the
doubling dimension is an interesting problem.

Although the growth dimension has a tighter $\alpha$ than the
doubling dimension, $\alpha$ has the same formulation in the
growth dimension and the doubling dimension. For any node $P$,
$\alpha$ is the minimal value satisfying $\frac{{\left|
{{B_P}\left( {xr} \right)} \right|}}{{\left| {{B_P}\left( r
\right)} \right|}} \le {x^\alpha }$ based on Lemma
\ref{grid_growth} and \ref{grid_doubling}, i.e.,
\begin{equation}\label{alphaDef}
\alpha = {\log _x}\left( {\frac{{\left| {{B_P}\left( {xr} \right)}
\right|}}{{\left| {{B_P}\left( r \right)} \right|}}} \right)
\end{equation}
where $r
>0$, $x \geq \rho$, and $\rho$ is the inframetric value of
the delay data set. }

\subsection{DNNS on the Inframetric Model}
\label{DN2SGrowth}

In this section, we present the analysis of DNNS on the
Inframetric model. We will show the search accuracy, search
periods and search costs related to a DNNS process. We prove that,
by recursively following such sampling conditions, we can locate a
server that is $1/\beta$-approximation to the optimal: the delay
from the found server to the target is not bigger than $1/\beta$
times that from the nearest server to the target.

First, we review the goal of each DNNS step using the sampling
conditions in Sec \ref{alphaCalc}. Assume that a node $P$ wants to
locate a node that is $\beta$ times closer to a target $T$. The
goal of the current DNNS step is to locate a node $\beta$ times
closer to the target than the current node $P$. To that end,
Theorem \ref{thm:SamplingGrowth} shows that we need to sample up
to $3{\left( {\frac{{{\rho ^2}}}{\beta }} \right)^\alpha }$ nodes
uniformly at random from ${B_P}\left( {\rho r} \right)$ with
replacement.

Based on the sampling condition in Theorem
\ref{thm:SamplingGrowth}, performing DNNS in the growth dimension
can be formulated into a simple DNNS procedure in Definition
\ref{samplingProcedure}.

\begin{definition}[A simple DNNS method in the inframetric model]
sampling $3{\left( {\frac{{{\rho ^2}}}{\beta }} \right)^\alpha }$
neighbors from the closed ball ${B_P}\left( {\rho
{d_{\mathit{PT}}}} \right)$ at each intermediate node $P$,
forwarding the DNNS request to a next-hop node $\beta$ times
closer to the target than the node $P$, and stopping at a local
minima when we can not find such a next-hop node.
\label{samplingProcedure}
\end{definition}

Furthermore, we can quantify the efficiency of found neighbors
based on the above DNNS procedure by Corollary
\ref{corrolaryGrowthDN2S}. As a result, we can locate an
approximately optimal nearest neighbor for a target $T$ when
$\beta$ approaches one. Furthermore, the number of required search
steps is a logarithm function of the ratio $\Delta$ of the maximum
delay to the minimum delay in the delay space, indicating that the
DNNS queries can complete quickly.

\begin{definition}[$\omega$-approximation]For a DNNS request with target $T$, a found nearest neighbor $A$ is a $\omega$-approximation, if the delay between $A$ to $T$ is smaller than $\omega {d_*}$, where $d_*$ is the delay between the real nearest neighbor to $T$.
\end{definition}

\begin{corollary}
For a relaxed inframetric model with growth ${\gamma _g}$,
according to the DNNS process in Definition
\ref{samplingProcedure}, the found nearest neighbor is a
$\frac{1}{\beta }$-approximation, and the number of search steps
is smaller than ${{{\log }_{\frac{1}{\beta }}}\Delta }$, where
$\Delta$ is the ratio of the maximum delay to the minimum delay of
all pairwise delays. \label{corrolaryGrowthDN2S}
\end{corollary}

\subsection{Limitations of Theoretical Results}
\label{limitations}

To find a better next-hop neighbor without missing any closer
nodes, based on the DNNS analysis in the inframetric model in
Sec~\ref{DN2SGrowth}, we should sample approximately $3{\left(
{\frac{{{\rho ^2}}}{\beta }} \right)^\alpha }$ nodes whose delays
to current node $P$ are not larger than ${\rho {d_{PT}}}$.
However, the number of the candidate neighbors may be quite high,
as shown in Fig \ref{fig:NumSamplesoutput}. We can see that  the
number of required samples exceeds 100 accordingly, for $\beta$
below 0.4 or $\alpha$ above 1. Such high number of samples implies
that we need extremely large number of samples for continuing the
DNNS query.

\co{ , where $\beta$ is the delay reduction threshold, $\alpha$ is
bounded as defined in Lemma \ref{grid_growth}, ${d_{PT}}$
represents the delay from the current node $P$ to the target $T$.
}
%

\co{ Fig \ref{fig:NumSamplesoutput} depicts the number of sampled
neighbors $3{\left( {\frac{{{\rho ^2}}}{\beta }} \right)^\alpha }$
by varying the delay reduction threshold $\beta$ from the interval
$[0.1,1]$.

Therefore, the simple DNNS method has limitations on the
measurement costs. }

On the other hand, the number of samples decreases with decreasing
$\alpha$ or with increasing $\beta$. When $\alpha$ is below 1, the
number of samples is below 33 if the delay reduction threshold
$\beta$ is above 0.8. As a result, we can see that we need to
choose a large $\beta$ in order to reduce the number of samples,
since the median values of $\alpha$ are mostly no more than 1 from
Fig \ref{fig:alphaStat}.

\begin{figure}[tp]
\leavevmode \centering \setlength{\epsfxsize}{.6\hsize}
\epsffile{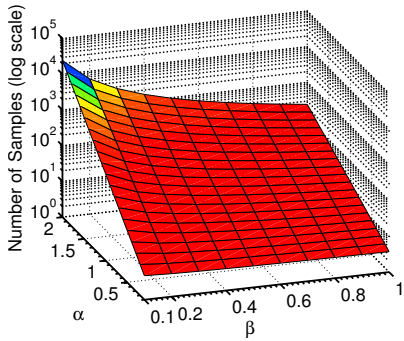} \caption{The
number of sampled neighbors $3{\left( {\frac{{{\rho ^2}}}{\beta }}
\right)^\alpha }$  by varying the volume difference parameter
$\alpha$ from the interval $[0,2]$ based on the analysis in Sec
\ref{alphaCalc} and the delay reduction threshold $\beta$. We set
the inframetric parameter $\rho$ to be 3 to represent most
triples. } \label{fig:NumSamplesoutput}
\end{figure}

\subsection{Comparison with Previous Inframetric Study}

Our relaxed inframetric model is inspired by the seminal study on
the inframetric model \cite{DBLP:conf/infocom/FraigniaudLV08} that
assumes the symmetry of the distance function. We extend the
inframetric model study for the Internet delays in four aspects:
\begin{itemize}
\co{
    \item  We prove that it is
NP-complete to compute the doubling metric for the inframetric and
the relaxed inframetric model. Besides, we prove the efficiency of
calculating the doubling metric on the relaxed inframetric model
using the greedy set cover technique.\item Moreover, we provide
detailed statistics of the doubling metric using diverse data
sets. On the contrary, the seminal work uses a randomized approach
to approximate the doubling metric on the inframetric model, which
does not provide accuracy guarantees. }

\item We extend the inframetric model to allow both symmetric and
asymmetric distance functions, which generalizes the RTTs and OWDs
that are important for latency-sensitive applications.

\item We clearly show the relation between inframetric parameter
$\rho$ and the TIV. The inframetric parameter $\rho \le 2$ is a
necessary but not sufficient condition for no TIVs.

\item We formulate the DNNS problem on the relaxed inframetric
model and propose a simple DNNS method that finds approximately
nearest neighbor for any target using at most logarithmic search
hops. Interestingly, our simple DNNS method works on both
symmetric and asymmetric delay metrics.

\end{itemize}

\co{\item We give detailed statistics for the inframetric
modelling parameter $\rho$ and the growth metric for symmetric and
asymmetric delay data sets.}

\section{Realizing a Practical DNNS}
\label{practicalDNNS}

\subsection{Overcoming Limitations of the Simple DNNS Method}
\label{limitations}

Recall that the measurement costs limits the usefulness of the
simple DNNS method defined in Def \ref{samplingProcedure} from Sec
\ref{limitations}. Besides, in the distributed system context,
since each service node does not have the global view of the delay
space, sampling enough neighbors from the closed ball centered at
each service node is difficult. We discuss design principles to
tackle these two difficulties in this section.

\subsubsection{Reduce Measurement Costs}

We reduce the measurement costs in two complementary approaches:
(i) Given that the number of required samples of the simple DNNS
method depend on varying parameters, we seek to modify the
parameters to obtain the lower bound of the required number of
samples. (ii) Given that network coordinates can be used for delay
estimations, we avoid complete measurements from selected samples
to the target using delay estimations.

First, recall that the number of samples for the simple DNNS
method increases quickly with decreasing delay reduction threshold
$\beta$. Therefore, to reduce the number of samples, we should set
the delay reduction threshold $\beta$ to be close to 1. On the
other hand, since the approximation ratio of the simple DNNS
method is ${1/{\beta}}$, we can see that large $\beta$ also leads
to better approximations of nearest neighbors. As a result, we set
$\beta$ to 1 in order to reduce the number of samples and obtain
the best approximation accuracy.

\co{ For example, the number of required samples becomes quite
large (over 100) when the delay reduction threshold $\beta$
decreases from 1 to 0.1.

Note that setting $\beta$ to be 1 does not prevent us to find the
candidate neighbor that is closest to the target at each search
step.

Besides, measurement packet losses also prolong the query periods,
since we need to wait for the completion of the measurements from
all selected neighbors.
 }

Second, although we reduce the number of samples using modified
$\beta$, we still need delay measurements between selected samples
to the targets, which consume the bandwidth costs and CPU loads of
service nodes.  Therefore, we hope to reduce the required delay
measurements while obtaining the sample that is closest to the
target. To that end, we use \textit{delay estimations based on
network coordinates} to reduce the delay measurement costs.
However, since the delay estimations incur errors due to the
embedding distortions of network coordinates, simply using delay
estimations to find the nearest neighbors becomes less reliable.
Instead, we issue delay measurements when the delay estimations
are inaccurate, so as to avoid the inaccurate delay estimations.

\co{ Furthermore, since the delay reduction threshold $\beta$ also
determines when to terminate a DNNS query, which leads to local
minima caused by the clustering in the delay space. In order to
avoid such local minima, we avoid $\beta$ to be fixed during the
DNNS process. Instead, we allow to continue the DNNS process
adaptively. Therefore, we set $\beta$ to be 1.}

\subsubsection{Sample Enough Neighbors For Continuing DNNS Query}

Based on the simple DNNS method, each DNNS service has to maintain
enough neighbors covering different delay ranges in the delay
space, in order to find the nearest neighbor to any target.
Therefore, each node has to maximize its diversity in the neighbor
set.

Gossip based neighbor management is frequently used for existing
DNNS methods. For example, Meridian
\cite{DBLP:conf/sigcomm/WongSS05} and OASIS
\cite{DBLP:conf/nsdi/FreedmanLM06} use an anti-entropy gossip
protocol to discover neighbors, and store neighbors using rings of
neighbors called concentric rings. However, during our
experiments,  the innermost and outermost rings in the concentric
ring often find no or only few neighbors compared to the capacity
of the ring, while the rest of rings with radii lying in the
middle portion of the delay distributions are filled with too many
neighbors, leading to frequent ring management events, incurring
heavy computation and communication overhead.

We explain the insufficiency of the gossip process in details.
Assuming that we know the complete delay matrix, for each node, we
compute the percent of mapped nodes for each ring, which serves as
an upper bound of sampled neighbors for that ring. Then we can
analyze whether the distributions of mapped nodes in concentric
rings affect the gossip process. As shown in Fig
\ref{fig:RTTDist}, we can see that most nodes are mapped into a
few number of rings, whose delay ranges lie in the middle portion
of the delay distributions. However, only quite a few nodes are
mapped into the innermost and outermost rings, which result in a
skewed distribution of mapped nodes for the concentric rings. As a
result, since the gossip process adopts the uniform sampling
approach, the gossip process will inevitably sample insufficient
neighbors from those rings that have too few mapped nodes.

Accordingly, to improve the concentric ring maintenance, we need
to sample enough neighbors that lie in different delay ranges. To
that end, we propose to find nearest neighbors and farthest
neighbors for each service node, in order to fill the innermost
and outermost rings in the concentric ring.

\co{the low-density regions correspond to the head portion and the
tail portion in the delay distribution}

\begin{figure}[tp]
     \centering
                    \subfigure[DNS1143.]
        {
          \setlength{\epsfxsize}{.44\hsize}
          \epsffile{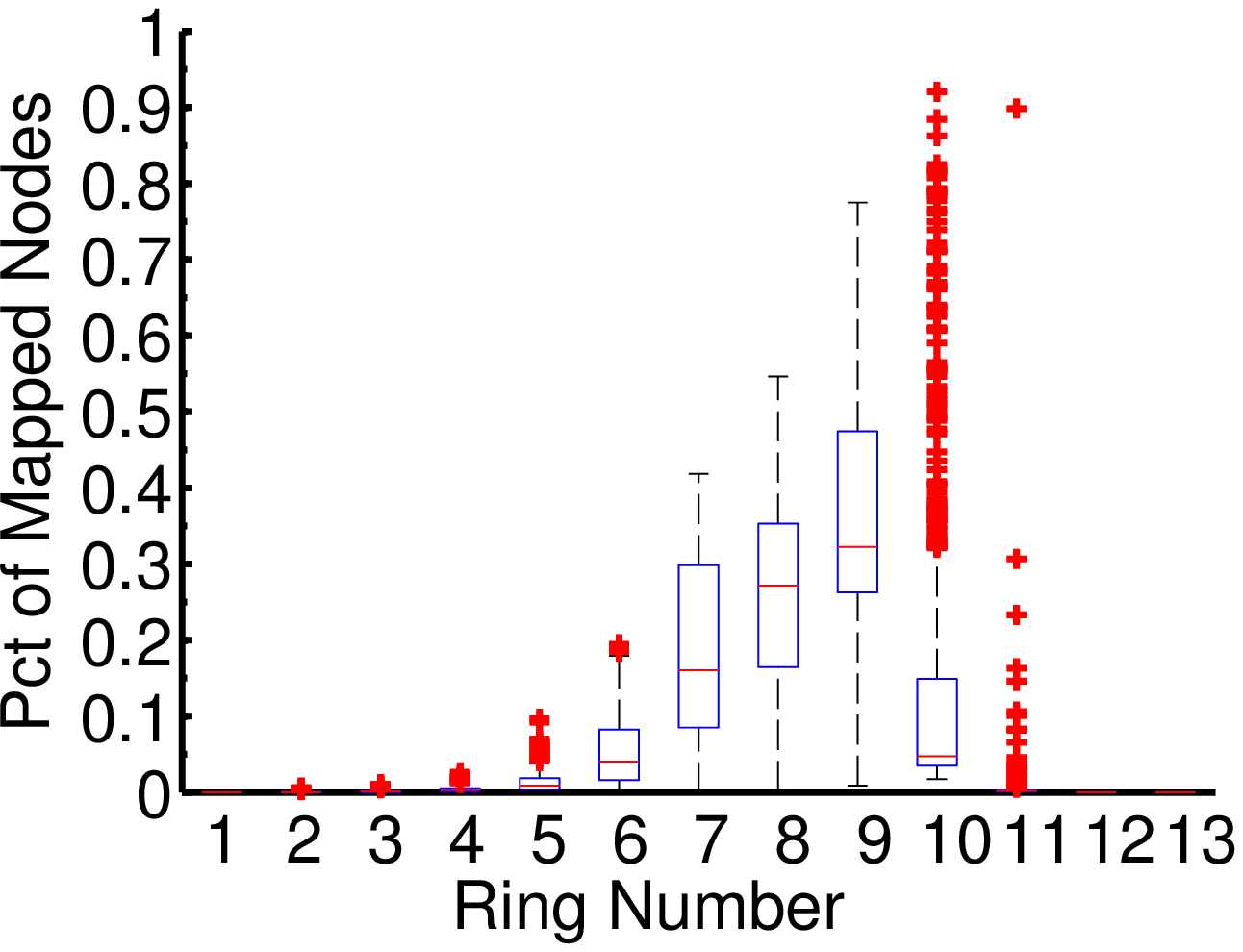}
         }
              \subfigure[DNS2500.]
        {
          \setlength{\epsfxsize}{.44\hsize}
          \epsffile{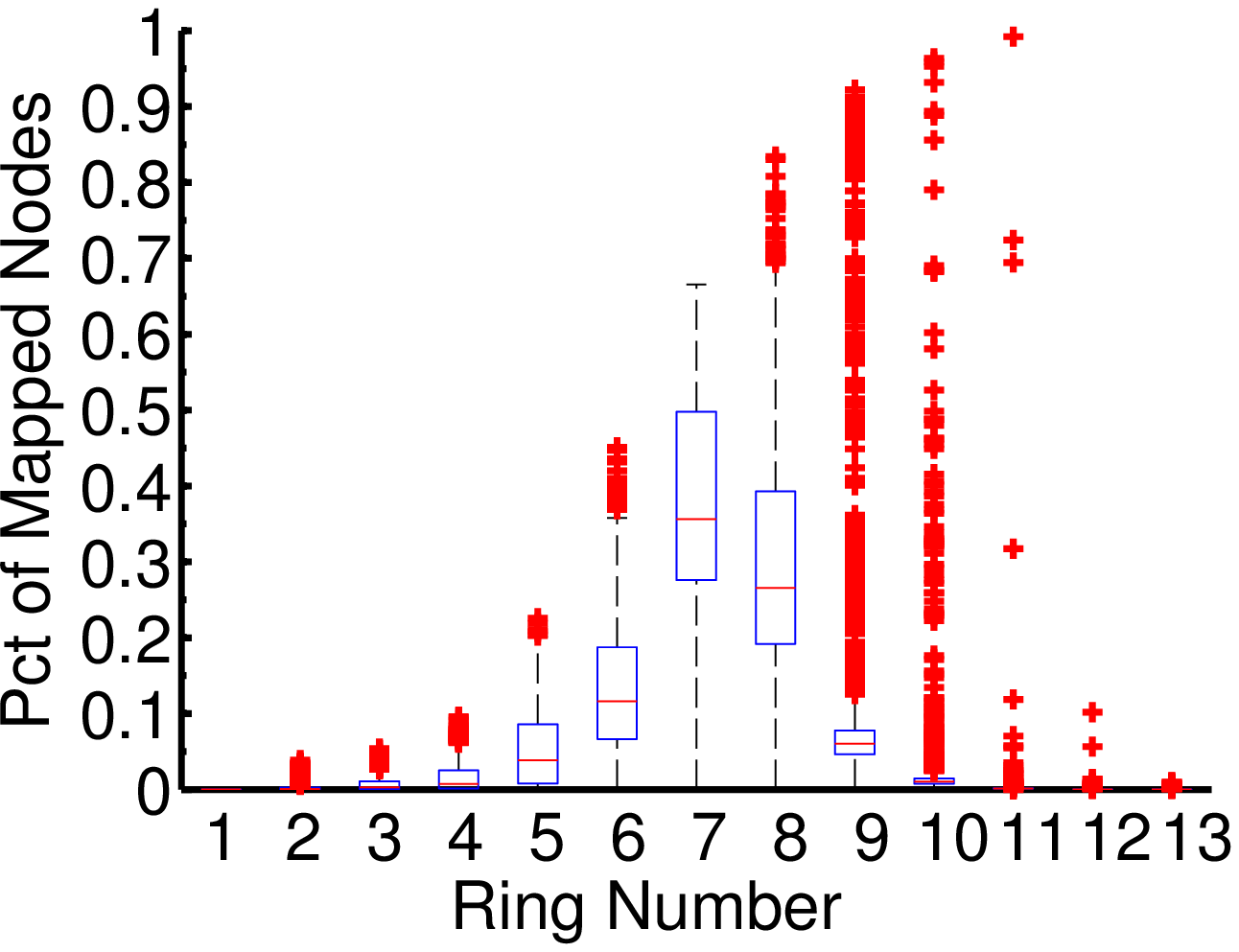}
         }
               \subfigure[DNS3997.]
        {
          \setlength{\epsfxsize}{.44\hsize}
          \epsffile{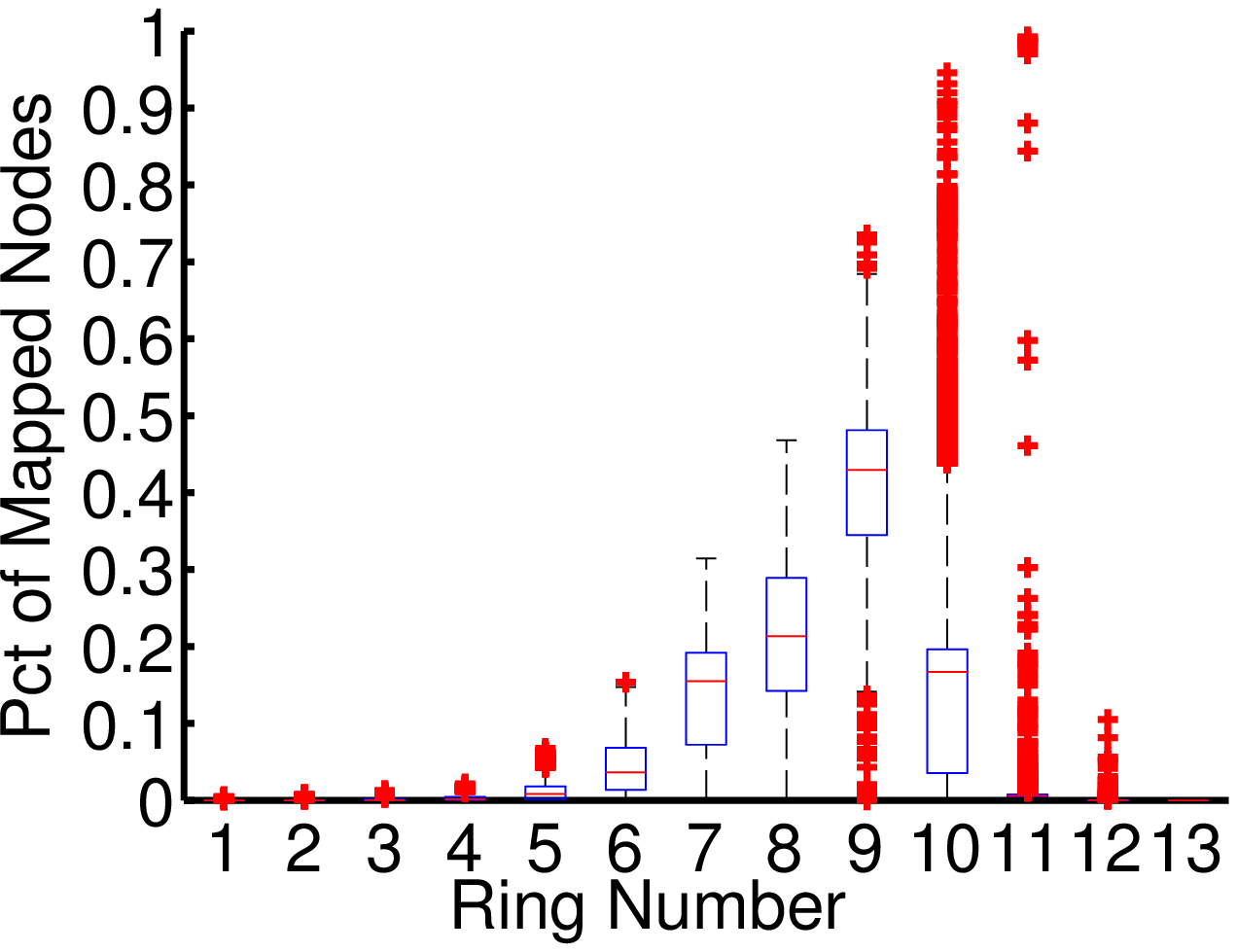}
         }
              \subfigure[Host479.]
        {
          \setlength{\epsfxsize}{.44\hsize}
          \epsffile{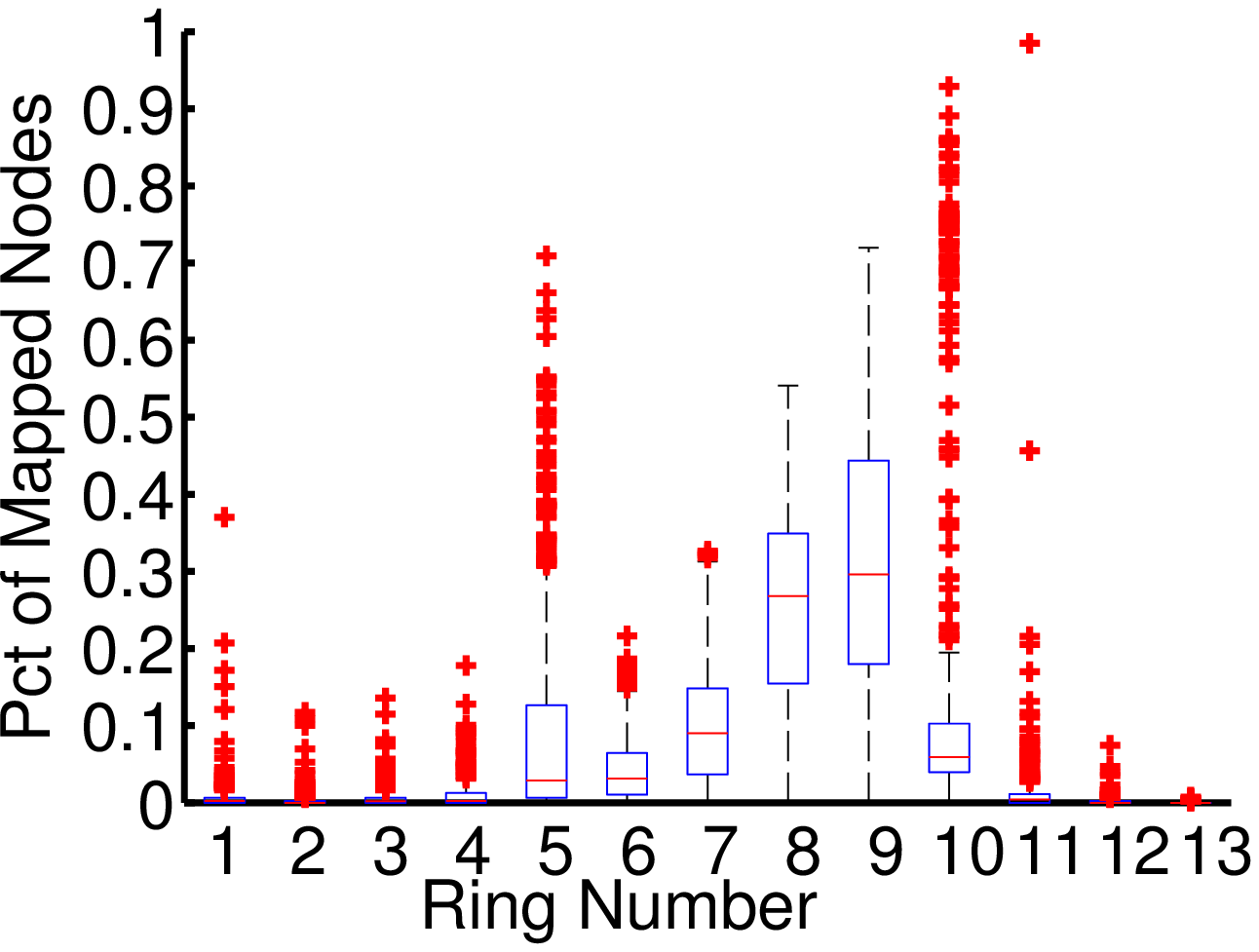}
         }
     \caption{The percent of mapped nodes into different rings, assuming that we obtain the complete delay matrix.
     The $i$-th ring
contains neighbors whose delays to a node $P$ lie in the interval
$\left( {\alpha {s^{i - 1}},\alpha {s^i}} \right]$, with $i>0$,
$\alpha$ a constant, $s$ a multiplicative increase factor
($\alpha=1$, $s=2$ ms as configured by Wong et
al.~\cite{DBLP:conf/sigcomm/WongSS05}). Besides, since our
objective is to determine the distribution of nodes mapped into
the concentric ring, we do not limit the maximum capacity of each
ring.}
     \label{fig:RTTDist}
\end{figure}

\subsection{Our Design}

Based on the design principles in Sec \ref{limitations}, we design
a novel DNNS method named \textit{HybridNN} (Hybrid Nearest
Neighbor Search). We present an overview of HybridNN. To sample
enough candidate neighbors from the proximity region of the
current node, each node must first maintain a neighbor set that
contains enough neighbors within each proximity region. Then using
the neighbor set, we select candidate neighbors using the sampling
conditions of the simple DNNS method, in order to cover the
neighbors closer to the target with high probability. Next, we
determine the candidate neighbor closest to the target, using
delay estimations and direct probes, in order to obtain a better
tradeoff between sampling bandwidth and accuracy. Finally, using
the currently nearest candidate neighbor to the target, we
determine whether to terminate the DNNS query. As shown in Fig
\ref{fig:HybridNN_Illustration}, HybridNN is composed of five
components:

 \noindent   \textbf{Neighbor Maintenance}: This component maintains the
neighbor set for DNNS queries. Since nodes are mapped into the
rings at the middle portion of the concentric ring, which implies
that neighbors mapped into the head portion and tail portion of
the concentric ring are difficult to be sampled using the uniform
sampling based approach. As a result, we need to increase the
sampling probability of such neighbors, in order to fulfill the
sampling conditions for DNNS queries. To that end, we
over-sampling neighbors in the head portions and tail portions of
the concentric rings, besides we uniformly sampling neighbors
located in the middle portions of delays and.

 \noindent    \textbf{Selecting Candidate Neighbor}: This component selects
candidate neighbors to satisfy the sampling conditions of the
simple DNNS method. When a node $P$ receives a DNNS query, node
$P$ determines its delay towards the target $T$, then selects
neighbors from its diversity-optimized neighbor sets (Sec
\ref{neighborMain}) by covering possible closer neighbors towards
the target $T$ (Sec \ref{choose}). Furthermore, we prune those
neighbors that could mislead the DNNS query into poor local
minima.

 \noindent     \textbf{Coordinate Maintenance}: This component updates the
coordinate of the target in order to estimate delays to targets
from candidate neighbors, since the target machine may not have
the coordinate for delay estimation (Sec \ref{coordMaintain}).
Additionally, each service machine maintains a network coordinate
used for delay estimations.

 \noindent     \textbf{Determining Closest Neighbor}: This component determines
the neighbor nearest to the target (Sec \ref{find}). Each node
computes the candidate neighbor closest to the target using delay
estimations and direct probes, in order to balance between the
measurement costs and measurement accuracy.

 \noindent     \textbf{Termination Test}: This component determines to continue
or stop a DNNS query (Sec \ref{terminate}). Recall that in
previous section we set the delay reduction threshold $\beta$ to
be 1 on order to reduce the number of samples and obtain better
approximation ratios to the optimal results. Therefore, HybridNN
conservatively terminate the DNNS query only when all candidate
neighbors having larger delays than the current node.

Finally, HybridNN uses an extensible delay measurement interface.
For instance, by default HybridNN simply use the system-built-in
Ping command to obtain pairwise RTT measurements. When there exist
an on-demand OWD probe service such as Reverse Traceroute
\cite{DBLP:conf/nsdi/Katz-BassettMASSWAK10}, HybridNN configures a
RPC interface to request the pairwise OWD results.

\begin{figure}[tp]
   \leavevmode \centering \setlength{\epsfxsize}{0.90\hsize}
 \epsffile{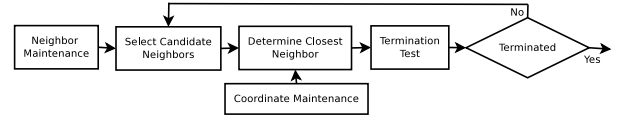}
  \caption{The flow chart of four search steps at a service node for a DNNS query.}
  \label{fig:HybridNN_Illustration}
\end{figure}

\co{
    \item We propose to terminate the DNNS process conservatively, by choosing low delay reductions during the DNNS queries to take into account the clustering structure in the delay space.

}

\subsection{Neighbor Maintenance}
\label{neighborMain}

In order to facilitate the neighbor sampling for DNNS forwarding,
each service node maintains neighbors that are sampled from
different regions in the delay space. We introduce the neighbor
discovery and update in this section.

\subsubsection{Organize Neighbors Into Rings for Proximity
Selection}

\co{To limit the storage overhead, the number of rings is limited
to be a constant $i^*$ (20 by default), and the number of
neighbors in each ring is bounded to be a constant $\Delta$.
Consequently, all rings $i>i^*$ are collapsed into a single
outermost ring spanning the interval $\left( {\alpha {s^{{i^*}}},
+ \infty } \right]$.}

Since the proximity region for neighbor sampling in the simple
DNNS method is a closed ball, we choose the concentric ring to
organize neighbors for each node. For instance, if we need to
locate all neighbors that are at most $d_2$ ms away, we select all
neighbors from those rings whose ring numbers are at most
$\left\lceil {{{\log }_2}{d_2}} \right\rceil$.

An important parameter for the concentric ring is its ring size
$\Delta$, which determines the maximum number of neighbors per
ring. Since we need to sample enough neighbors using the
concentric ring to guarantee to locate a neighbor closer to the
target with a high probability, we analytically determine the
choice of $\Delta$ as follows. First, the total number of samples
$3{\left( {\frac{{{\rho ^2}}}{\beta }} \right)^\alpha }$ is within
the interval $\left[ {3\gamma _g^2,3\gamma _g^4} \right]$, since
we set the delay reduction threshold $\beta$ to 1. Therefore, if
we set the number of neighbors $\Delta$ at each ring to be at
least $O(\gamma _g^2)$, we can ensure that with a high
probability, we can find a neighbor that is closer to the target
than the current node $P$. Furthermore, since $\gamma _g$ is low
on average from previous sections, we can set the number of
neighbors $\Delta$ to be a modest integer (8 by default).

\co{ In order to sample enough neighbors at the proximity region
for each current node that receives a DNNS query, we analyze the
number of neighbors $\Delta$ at each ring, which should fulfill
the number of samples required to find neighbors closer to the
target.

Recall that the samples lie in the set of nodes whose delays to
current node $P$ are not larger than ${\rho {d_{PT}}}$.

Since we set the delay reduction threshold $\beta$ to be 1, the
total number of samples $3{\left( {\frac{{{\rho ^2}}}{\beta }}
\right)^\alpha }$ is within the interval $\left[ {3\gamma
_g^2,3\gamma _g^4} \right]$. }


Furthermore, to adapt to the dynamics of delays, we use a moving
median as a latency filter for extracting stable delay
measurements to each neighbor \cite{DBLP:conf/nsdi/LedlieGS07},
which allows to have up-to-date delay estimates resilient to the
measurement noises.

\subsubsection{Biased Sampling based Neighbor Discovery}
\label{newNodesFinding}

Based on the distribution of neighbors for each ring in the
previous section, we have seen that we need to over-sample
neighbors mapped into the head portion and the tail portion of the
concentric rings. To that end, we adopt both uniform sampling and
over-sampling approaches.

\textbf{Uniform sampling}. We reuse the gossip process in
Meridian. Briefly, each node $P$ periodically starts the gossip
process by uniformly selecting a neighbor $Q$ from $P$'s
concentric ring as
 communication partner, and sends a gossip request message to
node $Q$ containing randomly sampled neighbors, one neighbor per
non-empty ring. When $Q$ receives the gossip request, $Q$ will
send a gossip ACK to $P$ immediately; besides, $Q$ iteratively
sends gossip requests towards the sampled neighbors in the gossip
request message of $P$.

Finally, if we use the RTT metric, then node $P$ inserts $Q$ into
the corresponding ring according to the round trip delays measured
as the period between the gossip request and the gossip ACK.
Alternatively, if node $P$ is able to measure the one-way delay
from $P$ to $Q$, then node $Q$ is inserted into the corresponding
ring according to the one-way delay from $P$ to $Q$.

\textbf{Over-sampling}. Our goal is to sample enough neighbors
from those mapped neighbors lying in the head and tail portions of
the concentric rings. For this purpose, we use $K$ closest
neighbor search and $K$ farthest neighbor search. The returned
nodes are directly stored into the concentric ring, as the delay
values between the current service node to the returned nodes are
obtained during the $K$ closest neighbor search and $K$ farthest
neighbor search processes.

\begin{itemize}
    \item \textit{$K$ closest neighbor search}. Each node $P$ periodically finds nearby nodes by issuing $K$
closest neighbor search with itself as target. Here $K$ is a
system parameter. Firstly, node $P$ randomly selects a neighbor
$Q$ from its concentric ring, and sends to $Q$ a $K$ nearby
neighbor search message.  Then node $Q$ starts a $K$ closest
neighbor search process. After the $K$ closest neighbor search
process is completed, found nearby nodes and the corresponding
delays to $P$ are returned to node $P$, and $P$ saves these
returned nearby nodes into its concentric ring.
    \item \textit{$K$ farthest neighbor search}. Similar as the $K$  closest neighbor search process, each node $P$
periodically issues $K$ farthest neighbor search. Later, the $K$
farthest neighbor search results include found distant neighbors
and the corresponding delay values to node $P$. $P$ stores the
returned distant neighbors into its concentric ring by the
corresponding delay values.
\end{itemize}
Due to space limits, the details for $K$ closest neighbor search
and $K$ farthest neighbor search are omitted here, which can be
found in the full technical report \cite{HybridNNReport}.

\subsubsection{Replacing Suboptimal Neighbors Without Probes}
\label{updateNeighbors}

In order to bound the memory overhead of the concentric ring, we
need to manage the size of the concentric rings when some rings
reach their maximum capacity $\Delta$. To reduce CPU costs due to
frequent ring managements, we lower the frequency of ring
managements: we first set up another tolerance threshold $\Delta
_t$ for each ring; then we begin the ring management when some
rings having at least $\Delta+\Delta _t$ neighbors; during the
ring management, we remove $\Delta _t$ neighbors from those rings
that have at least $\Delta+\Delta _t$ neighbors.

When we need to remove $\Delta _t$ neighbors from some rings, we
follow the removing philosophy of Meridian: preserve those that
maximize the diversity of neighbors in a ring using the maximal
hypervolume polytope algorithm
(\cite{DBLP:conf/sigcomm/WongSS05}). This is because the higher
diversity in the neighbor set translates to better chances of
locating a nearby nodes for any target. However, the maximal
hypervolume polytope algorithm requires all-pair delay
measurements of nodes in a ring, which needs $O\left( {{\Delta
^2}} \right)$ probes. In order to avoid such measurements, we turn
to adopt network coordinates for delay predictions.

\co{

To maximize the neighbor diversity per ring, as nodes in one
identical ring are of similar distances to the center of the ring,
we need to select nodes from different locations. As a result, we
utilize the greedy maximal hypervolume polytope algorithm
(\cite{DBLP:conf/sigcomm/WongSS05}) to remove neighbors that are
near to other neighbors.}

For delay predictions, we use the revised Vivaldi algorithm
\cite{DBLP:conf/sigcomm/DabekCKM04} that is robust to TIVs
\cite{DBLP:conf/imc/WangZN07}. We denote the revised Vivaldi
\cite{DBLP:conf/imc/WangZN07} as \textit{TIV-Vivaldi($x_i$, $e_i$,
$d_{ij}$, $x_j$, $e_j$)}, where the input $x_i$, $x_j$ denote the
coordinate of node $i$ and $j$, respectively; the input $e_i$,
$e_j$ denote the averaged error of node $i$'s and $j$'s
coordinates, respectively. The output of \textit{TIV-Vivaldi} are
the updated coordinate $x_i$ and coordinate error $e_i$ of node
$i$.

Each service node passively maintains a coordinate, and estimates
delays using coordinate distances. Besides, for estimating delays
with neighbors in the concentric ring, each service node also
stores its neighbors' coordinates.

Since delay varies, each node updates its own and cached
coordinates periodically. Rather than introduce additional delay
probes, we update coordinates by reusing the delay measurements to
other service nodes during the biased sampling procedure.
Therefore, we significantly reduce the maintenance costs compared
to Meridian. First, each node receiving the gossip message
piggybacks its coordinate to the sender along with the
acknowledged gossip message. After receiving the coordinate from
the gossip receiver node, the gossip sender node stores the new
coordinate of the gossip receiver node, and updates its own
coordinate by triggering \textit{TIV-Vivaldi} using the delays
obtained during the gossiping process.

\subsection{Select Candidate Neighbors} \label{choose}

Assume that node $P$ receives a DNNS query to the target $T$.
Based on the sampling conditions of the simple DNNS method, node
$P$ needs to select $3{\left( {\frac{{{\rho ^2}}}{\beta }}
\right)^\alpha }$ neighbors whose delays to node $P$ are in the
delay range $[0, \rho d_{PT}]$. Since each ring contains $O(\gamma
_g^2)$ neighbors, we simply select all neighbors of rings numbered
in the range $\left[ {1,\left\lceil {{{\log }_2}\left( {\rho
{d_{PT}}} \right)} \right\rceil } \right]$ as candidate neighbors.

Furthermore, we also prune several neighbors that mislead the DNNS
process. First, candidate neighbors that contain too few non-empty
rings are more likely to provide no hints on continuing the DNNS
queries, thus the DNNS queries can be trapped into local minima,
due to the neighbors' sparse diversity of the delay space.
Therefore, we remove all neighbors with fewer than $\tau$
non-empty rings ($\tau = 4$ by default). Second, all neighbors
that have received the identical DNNS query should be removed in
order to avoid the search loops. Therefore, let the\emph{
forwarding path } of a DNNS query be the sequence of nodes
forwarding the query. we remove any node on the forwarding path.

\subsection{Coordinate Maintenance for Targets}
\label{coordMaintain}

In order to reduce the delay measurement costs, we predict delays
from service nodes to the target, since each service node has
computed its network coordinate during the neighborhood management
process (Sec~\ref{updateNeighbors}). As a result, reusing the
coordinates for predicting delays can reduce the measurement
costs.

Unfortunately, we may not know the coordinate of the target, as
the target can be any machine on the Internet. Therefore, we
propose to compute the coordinate for the target on-the-fly based
on the \textit{TIV-Vivaldi}.

First, when node $P$ receives the DNNS query for a target $T$,
node $P$ will initialize the network coordinate $x_T$ for target
$T$ if $T$'s coordinate is not stored in the DNNS query message.
To that end, node $P$ asks a fixed number of neighbors (at most
10) to directly probe the target $T$. Then, node $P$ updates
target $T$'s coordinate by \textit{TIV-Vivaldi} using the
coordinates and delay measurements from these neighbors to target
$T$, which updates $T$'s coordinate $x_T$ and coordinate error
$e_T$ as the output of \textit{TIV-Vivaldi}. Finally, node $P$
stores target $T$'s coordinate into the DNNS query and forwards to
the next-hop node for recursive search. This completes the
coordinate initialization for the target $T$.

Second, after initializing $T$'s coordinate, each node $Q$ that
forwards the DNNS query will update target $T$'s coordinate for
better convergence of target $T$'s coordinate. To that end, each
node $Q$ applies \textit{TIV-Vivaldi} to update target $T$'s
coordinate $x_T$ and coordinate error $x_T$, using node $Q$'s
coordinate and delay $d_{QT}$ the target $T$.

\co{ using Algorithm~\ref{alg:InitTargetCoord}

\begin{algorithm}[t]
\caption{Maintaining target's network coordinate.}
\label{alg:InitTargetCoord}
\begin{algorithmic}
{\small \STATE \textit{InitTargetCoord($P$, $T$, $M$)} \STATE
\COMMENT{Input: current node $P$, target $T$, DNNS query $M$, the
constant $L$ (15 by default) as the number of sampled neighbors to
initialize the target's coordinate.} \STATE \COMMENT{Output: the
initialized coordinate of target $T$: $x_T$}

\IF {$M$.init == False}

 \STATE ${\Omega _P} \gets$ $L$ sampled neighbors from $P$'s concentric
 ring;

 \FOR{$i \in {\Omega _P}$}

 \STATE $d_{iT}\gets directProbe(i,T)$; \algorithmiccomment{ICMP probes}

 \STATE $\left[ {{x_T},{e_T}} \right] \gets$ TIV-Vivaldi$(x_T, e_T,d_{iT}, x_i,e_i)$; \algorithmiccomment{update target's coordinate}

 \ENDFOR

 \STATE $M$.init $\gets$ True;

 \ELSE

 \STATE $d_{PT}\gets directProbe(P,T)$;

 \STATE $\left[ {{x_T},{e_T}} \right] \gets$ TIV-Vivaldi$(x_T,e_T,d_{PT},x_P,e_P)$;

\ENDIF

 \STATE Return $\left[ {{x_T},{e_T}} \right]$;

}
\end{algorithmic}
\end{algorithm}
}

\subsection{Determine Closest Neighbor} \label{find}

After we assign a network coordinate to the target in Sec
\ref{coordMaintain}, we can use the network coordinate distances
to approximate the real-world delay and reduce the measurement
costs. Nevertheless, since the coordinate distances are only
approximations, closest neighbors selected according to the
network coordinates may be inconsistent with the real ones.

Therefore, we locate closest neighbors to the target $T$ from the
candidate neighbors found in Sec \ref{choose}, by combining the
delay predictions with a small number of direct probes.

First, based on the coordinate distances from candidate neighbors
to target $T$, we find top-$m$ nearest neighbors ${S_c}$ to the
target $T$ from the candidate neighbors.

Second, since coordinate distances may be erroneous, we also
choose those candidate neighbors $S_e$ whose coordinates are not
reliable. Since each TIV-Vivaldi coordinate $x_i$ is accompanied
by a coordinate error metric $e_i$ \cite{DBLP:conf/imc/WangZN07},
we choose unreliable neighbors whose coordinate errors exceed a
threshold. We found that setting the threshold to be 0.7 can
significantly reduce the negative impact due to the coordinate
inaccuracies.

Third, to adapt to coordinate errors caused by TIV, since high
coordinate distance errors indicate violations of triangle
inequality \cite{DBLP:conf/imc/WangZN07}, we simply include all
candidate neighbors $S_t$ whose coordinate distance and real delay
towards the current node $P$ differs by more than 50 ms, which has
good tradeoff between accuracy and bandwidth costs.

Finally, using the union of selected candidate neighbors $S_* =
{S_c} \cup {S_e} \cup {S_t} $, the current node $P$ asks neighbors
in $S_*$ to probe the delays to target $T$, from which node $P$
determines the closest neighbor.  Ties are broken by choosing the
neighbor with most accurate coordinate.

\subsection{Termination Test} \label{terminate}

Recall from Sec \ref{limitations}, HybridNN set the delay
reduction threshold $\beta$ to be 1, in order to reduce the number
of selected neighbors and obtain better approximation ratios for
the found nearest neighbors. Therefore, when the closest neighbor
selected from Sec \ref{find} has a larger delay to the target than
that of the current node $P$, node $P$ terminates the DNNS query.
Then node $P$ sends the currently closest node to the host that
issues the DNNS query.

\section{Extensions to HybridNN}

HybridNN can be readily extended to search more than just one
nearest node. Here we will just give two examples namely, $K$
closest neighbor search and $K$ farthest neighbor search, which
are both utilized to oversample neighbors in the network delay
space in order to increase the diversity for neighborhood
management.

\subsection{$K$ Distributed Nearest Neighbor Search} \label{KNN}

The $K$ Distributed Nearest Neighbor Search (KDN$^2$S) aims to
locate the $K$ nearest neighbors to a target $T$, where $K$ is a
system parameter. To store the found nearest neighbors, we append
a new field $M$.$\Omega$ that caches nearest neighbors to the DNNS
query message $M$.

A naive KDN$^2$S solution is based on the \textit{finding and
removing} approach:  first we find one closest neighbor towards
the target based on the HybridNN algorithm, then we delete the
found nearest neighbor from the system, and we restart the
HybridNN algorithm from the same query node until we locate $K$
nearest servers to the target. Nevertheless, deleting the closest
neighbors from the system is not practical for a large-scale
system due to the broadcasting communication overhead, and
repeated DNNS processes increase the query overhead for the
service nodes on the DNNS forwarding paths.

On the other hand, if we assume that the concentric ring of each
node does not append new neighbors, the network coordinate of each
node keeps unchanged and the network delays keep stable during the
period of a KDN$^2$S query, we find that there exists
\textit{temporal correlation} in the forwarding paths of
consecutive DNNS queries starting from the identical node in the
naive KDN$^2$S solution: \textit{if we issue a new DNNS query from
the same starting node immediately after the preceding DNNS query,
then the forwarding path truncated the last-hop node of the new
DNNS process is a subpath of the forwarding path of the preceding
DNNS query, since we can see that the intermediate nodes on these
two forwarding paths are identical in HybridNN.} Our assumption
generally holds after the network coordinates converge and the
concentric rings contain enough neighbors. Furthermore, the
constancy of end to end network delays has been confirmed to be on
the orders of hours by Zhang and Duffield
\cite{DBLP:conf/imw/ZhangD01} as well as the iPlane project
\cite{DBLP:conf/nsdi/MadhyasthaKAKV09,DBLP:conf/osdi/MadhyasthaIPDAKV06}.

Using the temporal correlation of consecutive forwarding paths
from the same starting node, we propose a backtracking based
KDN$^2$S algorithm, as shown in Algorithm~\ref{alg:KNNAlgorithm}.
After we find one nearest neighbor and terminate at a service node
$P_1$ by HybridNN, we resume the KDN$^2$S query from $P_1$, by
backtracking from $P_1$ to its predecessor node $P_2$ on the DNNS
forwarding path, and by recursively finding the nearest neighbor
at $P_2$, until we locate $K$ nearest neighbors. With
backtracking, the KDN$^2$S resumes the query at service nodes that
are close to the target, therefore we can quickly locate new
nearest neighbors with reduced forwarding overhead compared to the
naive KDN$^2$S solution.

Fig \ref{fig:KDNS} gives an example of KDN$^2$S using
Algorithm~\ref{alg:KNNAlgorithm}. Suppose an end host $A$ needs
two nearest neighbors to the target $T$. Node $A$ sends a KDN$^2$S
request to a service node $B$. Then $B$ starts the KDN$^2$S by
forwarding a  KDN$^2$S query $M$ to a neighbor $P_2$ closer to
$T$. Similarly, $P_2$ forwards the query $M$ to $P_1$. Now node
$P_1$ finds that it cannot find a neighbor closer to the target
$T$ than itself, therefore, $P_1$ is the first nearest neighbor to
the target. Then $P_1$ appends its address into $M$.$\Omega$ as a
found nearest neighbor.  Next $P_1$ triggers the KDN$^2$S
backtracking step by forwarding $M$ to $P_1$'s predecessor $P_2$
on the KDN$^2$S forwarding path. On receiving $M$, $P_2$ excludes
$P_1$ from the choice of candidate neighbors, and finds a new
neighbor $P_3$ closer to the target $T$ than $P_2$. Then $P_2$
forwards $M$ to $P_3$. $P_3$ decides that it is the closest node
to $T$ among its neighbors. Therefore, $P_3$ appends itself to
$M$.$\Omega$ as a new nearest neighbor. Finally, $P_3$ sends the
found nearest neighbors in $M$.$\Omega$, i.e., $P_1$ and $P_3$, to
the end host $A$, which completes the KDN$^2$S.

\begin{figure}[tp]
   \leavevmode \centering \setlength{\epsfxsize}{0.8\hsize}
 \epsffile{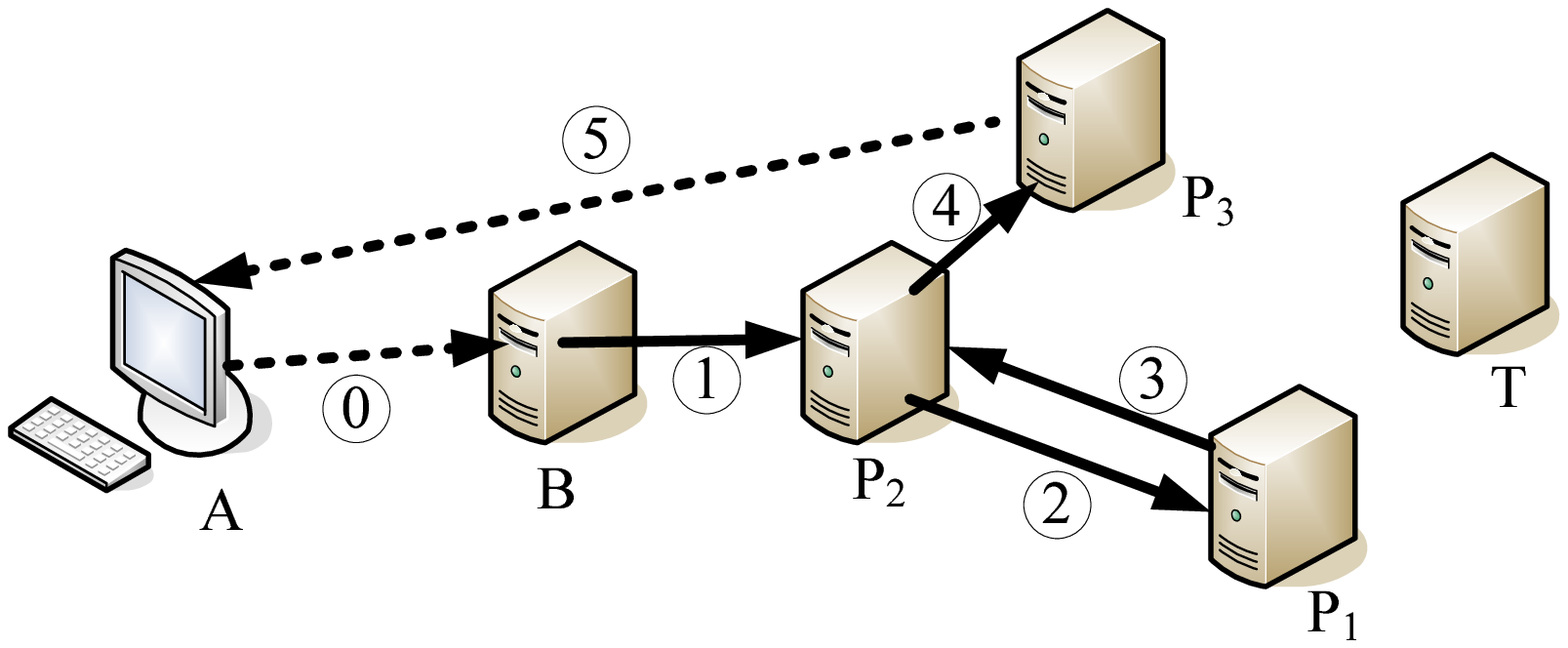}
  \caption{KDN$^2$S.}
  \label{fig:KDNS}
\end{figure}

\begin{algorithm}[t]
\begin{algorithmic}[1]
{\small \STATE \textit{KDN$^2$S($H$, $T$, $K$, $M$)} \STATE
\COMMENT{Input: current node $H$, the target $T$, required number
of closest neighbors $K$, query message $M$} \STATE
\COMMENT{Output: nearest neighbors to $T$}

\IF{$\left| {M.\Omega } \right| = = K$}

\STATE Return ${M.\Omega }$; \algorithmiccomment{enough closest
neighbors}

\ENDIF

\STATE $S \gets$ chooseCandidates($P$, $T$, $M$);

\STATE $S \gets S - M.\Omega$ \algorithmiccomment{remove found
nearest neighbors to avoid search loops}

\STATE $x_T \gets$ InitTargetCoord($P$, $T$);

\STATE $[u_1, S_c, D_T] \gets $ NearestDetector($P$, $S$, $x_T$,
$M$);

 \STATE $[\phi_1, d_{{\phi_1}T}, P_1] \gets$ TerminateTest($P$, $u_1$,$S_c$, $D_T$,
 $M$); \algorithmiccomment{find one closest neighbor, and terminate at node $P_1$}

\STATE $M.\Omega  \leftarrow M.\Omega  \cup \left\{ {{\phi _1}}
\right\}$; \algorithmiccomment{cache ${{\phi _1}}$ into the query
message}

\STATE  Select the predecessor node $P_2$ of node $P_1$ on the
forwarding path $M$.Path; \algorithmiccomment{find the predecessor
for backtracking}

\STATE \textit{KDN$^2$S($P_2$, $T$, $K$, $M$)};
\algorithmiccomment{recursive search}

 }
\end{algorithmic}
\caption{The pseudo-code of KDN$^2$S.} \label{alg:KNNAlgorithm}
\end{algorithm}

\co{\WHILE{node $P_1$ can find a neighbor $P_c$ satisfying
${d_{{P_c}T}} \leq {d_{{P_1}T}}$, ${P_c} \in S$}

\STATE $M.\Omega \gets M.\Omega \cup {P_c}$; \algorithmiccomment{a
new closest neighbor is found}

\ENDWHILE}

%
%

\subsection{$K$ Distributed Farthest Neighbor Search} \label{KFN}

Similar as the KDN$^2$S, $K$ Distributed Farthest Neighbor Search
(KDFNS) is also based on the backtracking idea. First, we locate
one farthest neighbor and terminate at a service node $P$, then we
backtrack from node $P$ to its predecessor node on the forwarding
path to recursively locate the rest $K-1$ farthest neighbors.

To locate one farthest neighbor, we recursively forward the
KDN$^2$S query to a service node $P_1$ that is at least $\left( {1
+ {\beta _{\mathit{farthest}}}} \right)$ (${\beta
_{\mathit{farthest}}}$ is 1.2 by default) times farther to the
target $T$ than the current service node $P$. In other words, we
need to locate a node that is not covered by the ball ${B_T}\left(
{\left( {1 + {\beta _{farthest}}} \right){d_{PT}}} \right)$. Since
${B_T}\left( {\left( {1 + {\beta _{farthest}}} \right){d_{PT}}}
\right) \subseteq {B_P}\left( {\rho \left( {1 + {\beta
_{farthest}}} \right){d_{PT}}} \right)$ by the sandwich lemma in
Lemma~\ref{sandwich}, $P_1$ needs to be at least $\rho ( {1 +
{\beta _{farthest}}} ){d_{PT}}$ from node $P$.

Accordingly, in each search step, we try to find such node $P_1$
from the concentric ring of the current service node $P$, whose
delay value to $P$ is larger or equal the $\rho \left( {1 + {\beta
_{farthest}}} \right){d_{PT}}$. If there exists a such node $P_1$,
then node $P_1$ recursively runs the KDFNS as node $P$. Otherwise,
if we can not locate such node $P_1$, the search is terminated,
and the currently farthest node to the target is cached as a
farthest neighbor to the target. Afterwards, we select the rest
$K-1$ distant neighbors by the backtracking process similar as
that in $K$ closest neighbor search.

Algorithm~\ref{alg:KFNAlgorithm} shows the complete KDFNS process.
First, we choose candidate neighbors satisfying the delay
constraint to the current service node $P$. Then we find the
farthest neighbor to the target (\textit{FarthestDetector}())
combining the delay predictions with direct probes in order to
reduce the measurement overhead. Specifically, we choose $m$
farthest neighbors from the candidate neighbors; besides, we also
add neighbors with uncertain coordinates and erroneous predictions
similar as Sec~\ref{find}. Next, we determine one farthest
neighbor recursively (FarthestTerminateTest). Finally, from the
terminating node $P_1$, we backtrack to the predecessor node of
$P_1$ on the forwarding path, and recursively run the KDFNS until
we locate enough farthest nodes to the target.

\begin{algorithm}[t]
\begin{algorithmic}[1]
{\small \STATE \textit{KDFNS($H$, $T$, $K$, $M$)} \STATE
\COMMENT{Input: current node $H$, the target $T$, required number
of farthest neighbors $K$, query message $M$} \STATE
\COMMENT{Output: farthest neighbors to $T$}

\IF{$\left| {M.\Omega } \right| = = K$}

\STATE Return ${M.\Omega }$; \algorithmiccomment{complete the
KDFNS}

\ENDIF

\STATE $S \gets$ chooseFarthestCandidates($P$, $T$, $M$);
\algorithmiccomment{choose neighbors whose delay values to $P$ is
larger than or equal to $\rho ( {1 + {\beta _{farthest}}}
){d_{PT}}$}

\STATE $S \gets S - M.\Omega$ \algorithmiccomment{remove found
farthest neighbors to avoid search loops}

\STATE $x_{T} \gets$ InitTargetCoord($P$, $T$);

\STATE $[u_1, S_c, D_T] \gets $ FarthestDetector($P$, $S$, $x_T$,
$M$); \algorithmiccomment{select the farthest neighbor to $T$ from
$S$}

 \STATE $[\phi_1, d_{{\phi_1}T}, P_1] \gets$ FarthestTerminateTest($P$, $u_1$,$S_c$, $D_T$,
 $M$); \algorithmiccomment{find one farthest neighbor, and terminate at node $P_1$}

\STATE $M.\Omega  \leftarrow M.\Omega  \cup \left\{ {{\phi _1}}
\right\}$; \algorithmiccomment{cache ${{\phi _1}}$ into the query
message}

\STATE  Select the predecessor node $P_2$ of node $P_1$ on the
forwarding path $M$.Path; \algorithmiccomment{find the predecessor
for backtracking}

\STATE \textit{KDFNS($P_2$, $T$, $K$, $M$)};
\algorithmiccomment{recursive search}

 }
\end{algorithmic}
\caption{The pseudo-code of KDFNS.} \label{alg:KFNAlgorithm}
\end{algorithm}

\section{Simulation}
\label{simulationExp}

In this section, we report the results of simulation experiments
based on the real-world data sets in Sec~\ref{datset}.

\subsection{Experimental Setup}
\label{simSetup}

We compare HybirdNN with several DNNS algorithms.
(1)\textbf{Vivaldi}. We compute the coordinate of each node based
on the Vivaldi algorithm \cite{DBLP:conf/nsdi/LedlieGS07}, and
find the nearest service nodes for each requesting node using
shortest coordinate distances. The coordinate dimension for
Vivaldi is 5. (2) \textbf{CoordNN}. To quantify the usefulness of
direct probes of HybirdNN, we present a DNNS algorithm CoordNN,
which is identical with HybridNN except that it uses only and no
direct probes when determining the best next-hop neighbors. (3)
\textbf{DirectDN2S}. To evaluate HybridNN, we present a DNNS
algorithm DirectDN2S, which is identical with HybridNN except that
it only utilizes direct probes for finding next-hop best neighbors
without pruning neighbors based on coordinate distances as
HybridNN. (4) \textbf{Meridian} \cite{DBLP:conf/sigcomm/WongSS05}.
Meridian recursively forwards the DNNS queries to a node that is
$\beta$ times closer to the target than the current node, and
returns the found nearest neighbor when no such node is selected.
We configure the parameters of Meridian algorithm identical with
the original configuration by Wong et
al.~\cite{DBLP:conf/sigcomm/WongSS05}, with the delay reduction
threshold $\beta$ as 0.5, the upper bound on the size of each ring
as 10, and the number of rings in the concentric ring is 20.

For HybridNN, the default configuration is summarized in Table
\ref{tab:parameterValues}. CoordNN and DirectDN2S share identical
parameters with HybridNN. We also evaluated the sensitivity of
parameters for HybridNN, which is reasonably robust against the
parameter choices. The detailed sensitivity results of system
parameters for HybridNN can be found in the technique report
published online \cite{HybridNNReport}.

\begin{table}\footnotesize
 \centering
\caption{Parameter values of HybridNN for
simulation.\label{tab:parameterValues}} {
\begin{tabular}{|l| l| l|} \hline
Parameter & Meaning & Value \\
\hline
 $\Delta$ & maximal size of the ring & 8\\\hline
$\Delta+\Delta _t$ & threshold of the ring size for ring updates &
10\\\hline ${\beta}$ & nearest search threshold & 1\\\hline
 $\rho$ & inframetric parameter & 3\\\hline $\left| x
\right|$ & coordinate dimension  & 5\\\hline $K$ & size of sampled
neighbors for neighbor discovery & 10\\\hline $m$ & number of
neighbors for direct probes  & 4\\\hline $\tau$ & number of
non-empty rings & 4\\\hline
\end{tabular}}
\end{table}%


We have developed a discrete-time simulator for  DNNS. The
simulator randomly chooses a set of nodes as service nodes (by
default 500) that can receive DNNS queries. Other nodes in the
system are clients that can issue DNNS queries to these service
nodes. For Host479, 200 nodes are the service nodes. The DNNS
queries are repeated 10,000 times. For each DNNS query, we
uniformly select one client as the target machine, and a random
service node receiving the query. Besides, the simulation is
repeated 5 times by shuffling the set of service nodes to avoid
biases in choosing service nodes. For HybridNN, CoordNN,
DirectDN2S and Meridian, the inter-gossip events for neighborhood
discovery are generated by an exponential distribution with
expected value of 1 second. The inter-ring management events are
generated by an exponential distribution with expected value of 2
seconds. For HybridNN, DirectDN2S and CoordNN, the time interval
between two oversampling events of $K$ closest neighbor search and
$K$ farthest neighbor search are generated by an exponential
distribution with expected value of 60 seconds. The inter-DNNS
event generation follows an exponential distribution with expected
value of 60 seconds. For Vivaldi, the coordinate of each node is
updated for 1000 rounds, by uniformly selecting a service node as
the counterpart during each round.

The performance metrics for each DNNS query include: (1)
\textbf{Absolute Error}: defined as the absolute difference
between the estimated nearest neighbor $j$ and the real nearest
neighbor $i$ to the target $T$, i.e., $ {{d_{jT}} - {d_{iT}}}$.
(2) \textbf{Relative Error}:  defined as the ratio of the absolute
error for the estimated nearest neighbor $j$ to the delay between
the real nearest neighbor $i$ and the target $T$, i.e., $\frac{{
{{d_{jT}} - {d_{iT}}}}}{{{d_{iT}}}}$. The absolute error
quantifies the increased delay values of the estimated nearest
neighbors, while the relative error measures the multiplicative
ratios to the optimal delay values for the estimated neighbors.
Therefore, large relative errors do not necessarily correspond to
high absolute errors. (3) \textbf{Search Hop}: defined as the
number of service nodes on the forwarding path minus one.
Therefore, if node $A$ forwards a DNNS query to node $B$ and node
$B$ returns the nearest neighbor to the query host, the search hop
for the DNNS query is one.

\subsection{Comparison}


\textbf{Absolute Error}. Fig~\ref{fig:AbsoErr} shows the absolute
errors of the different algorithms. DirectDN2S achieves lowest
absolute errors except for the Host479 data sets. HybridNN is
close to DirectDN2S in terms of reducing absolute errors, however,
HybridNN is the most accurate on Host479 data sets. Next, CoordNN
is worse than both DirectDN2S and HybridNN. The accuracy of
DirectDN2S and HybridNN compared to CoordNN indicates that
utilizing direct probes greatly reduces the inaccuracy of the
estimation, while using coordinate distances alone can lead to a
bad local minima.

The inaccuracy of DirectDN2S compared to HybridNN on the Host479
data set is rather counter-intuitive. The inaccuracy of DirectDN2S
may be caused by the asymmetry in the delay data sets that
misleads the greedy search into a local minima, since DirectDN2S
is more accurate than HybridNN on the other three data sets that
are all symmetric for pairwise delays. On the other hand, HybridNN
does not always choose the neighbor closest to the target as the
forwarding node, since HybridNN also incorporates the approximated
delay predictions when choosing neighbors, which can help HybridNN
bypass the bad local minimum caused by the asymmetry in the delay
values.

Furthermore, Meridian shows greater absolute errors compared to
other algorithms including Vivaldi, which implies that the
coordinate distances are at least effective if used it in the
centralized approach. We are aware that the superiority of Vivaldi
over Meridian in most cases are consistent with the experiments
independently performed by Choffnes and Bustamante
\cite{Choffnes:2010:PTE:1764873.1764880}. The main reasons for the
less accuracy of Meridian are the local minima caused by the TIV
and clustering in the delay space. On the other hand, Vivaldi can
adapt to TIV using adaptive coordinate movements.

\begin{figure}[tp]
     \centering
             \subfigure[DNS1143.]
        {
          \setlength{\epsfxsize}{.44\hsize}
          \epsffile{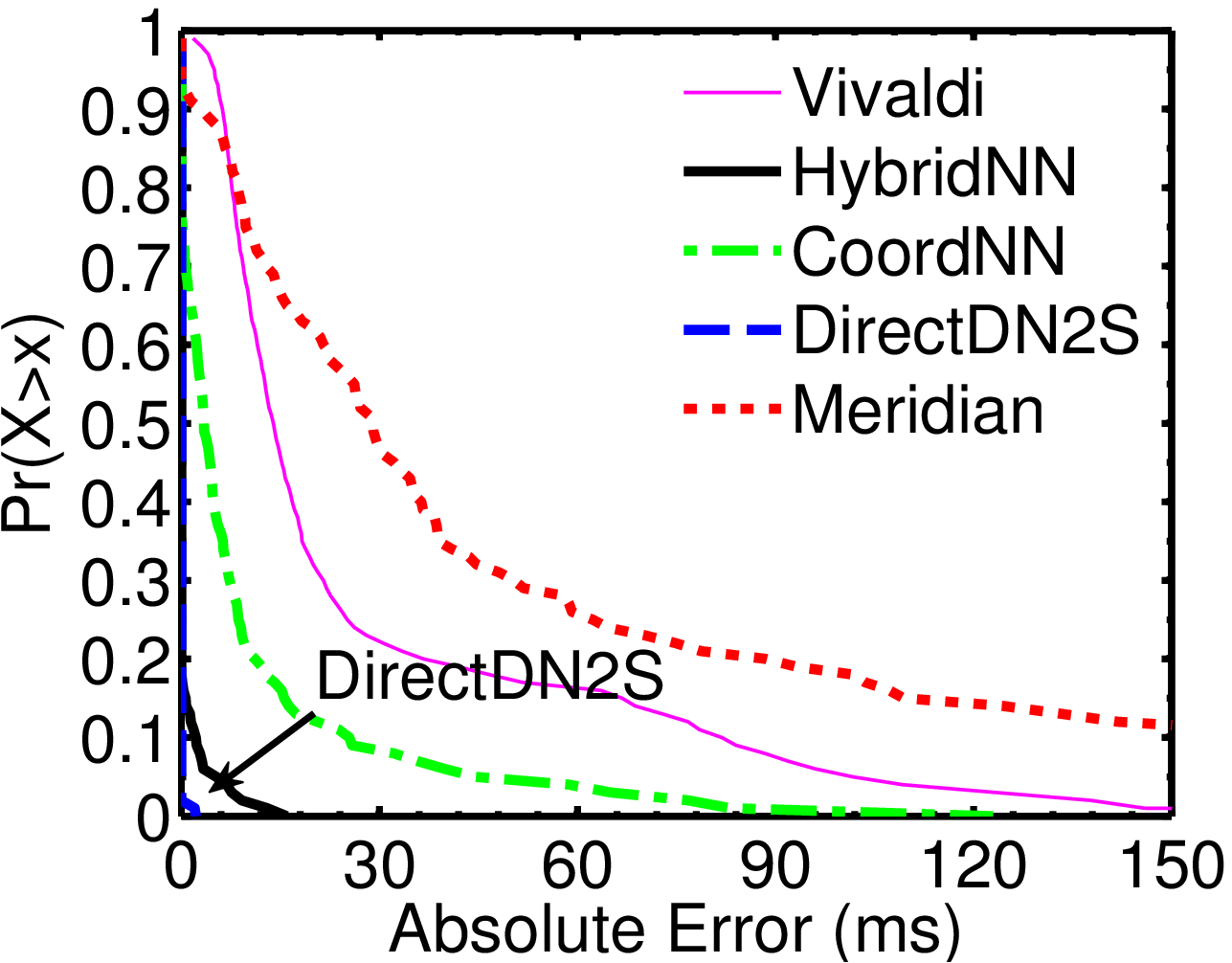}
         }
               \subfigure[DNS2500.]
        {
          \setlength{\epsfxsize}{.44\hsize}
          \epsffile{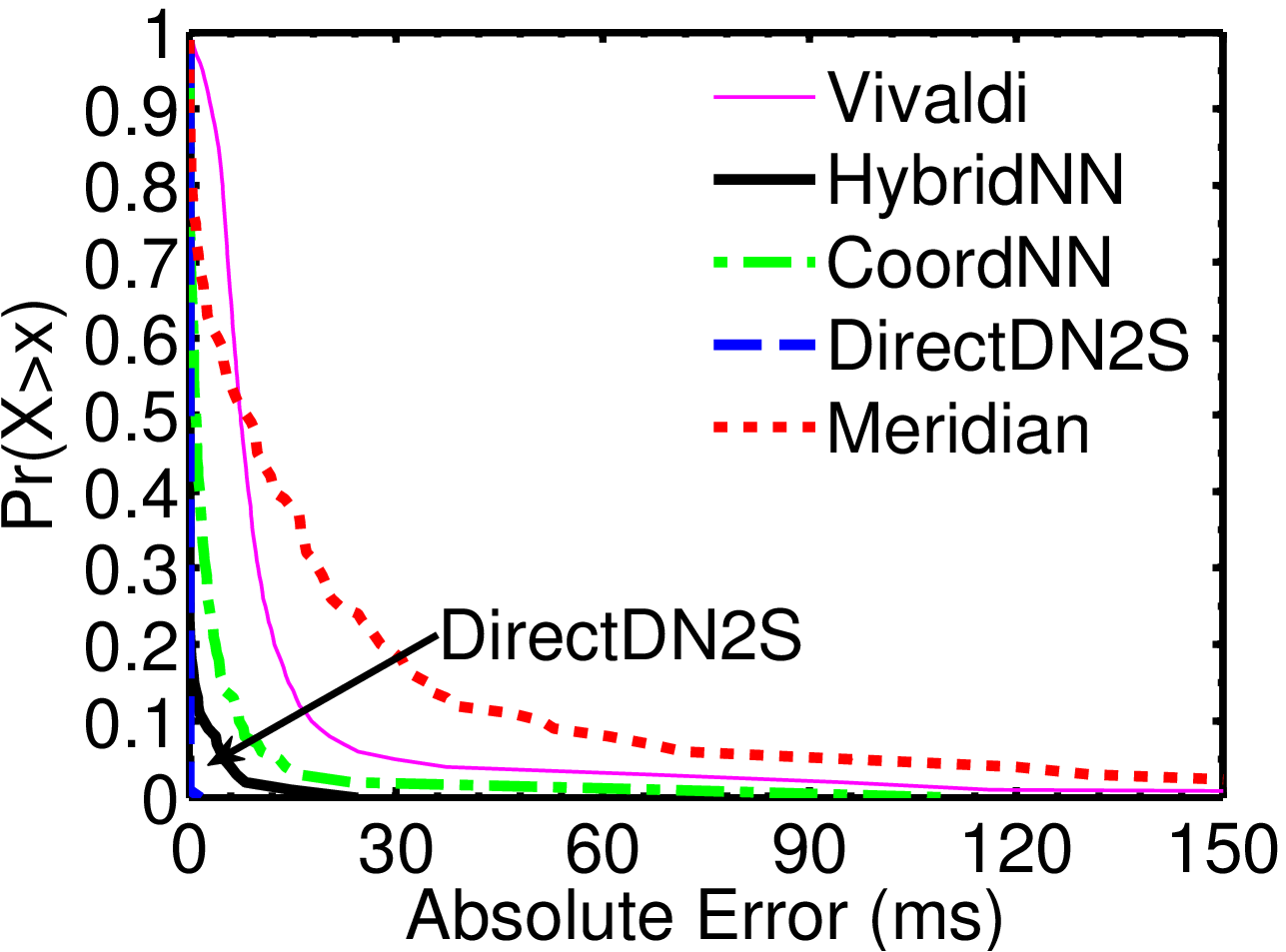}
         }
               \subfigure[DNS3997.]
        {
          \setlength{\epsfxsize}{.44\hsize}
          \epsffile{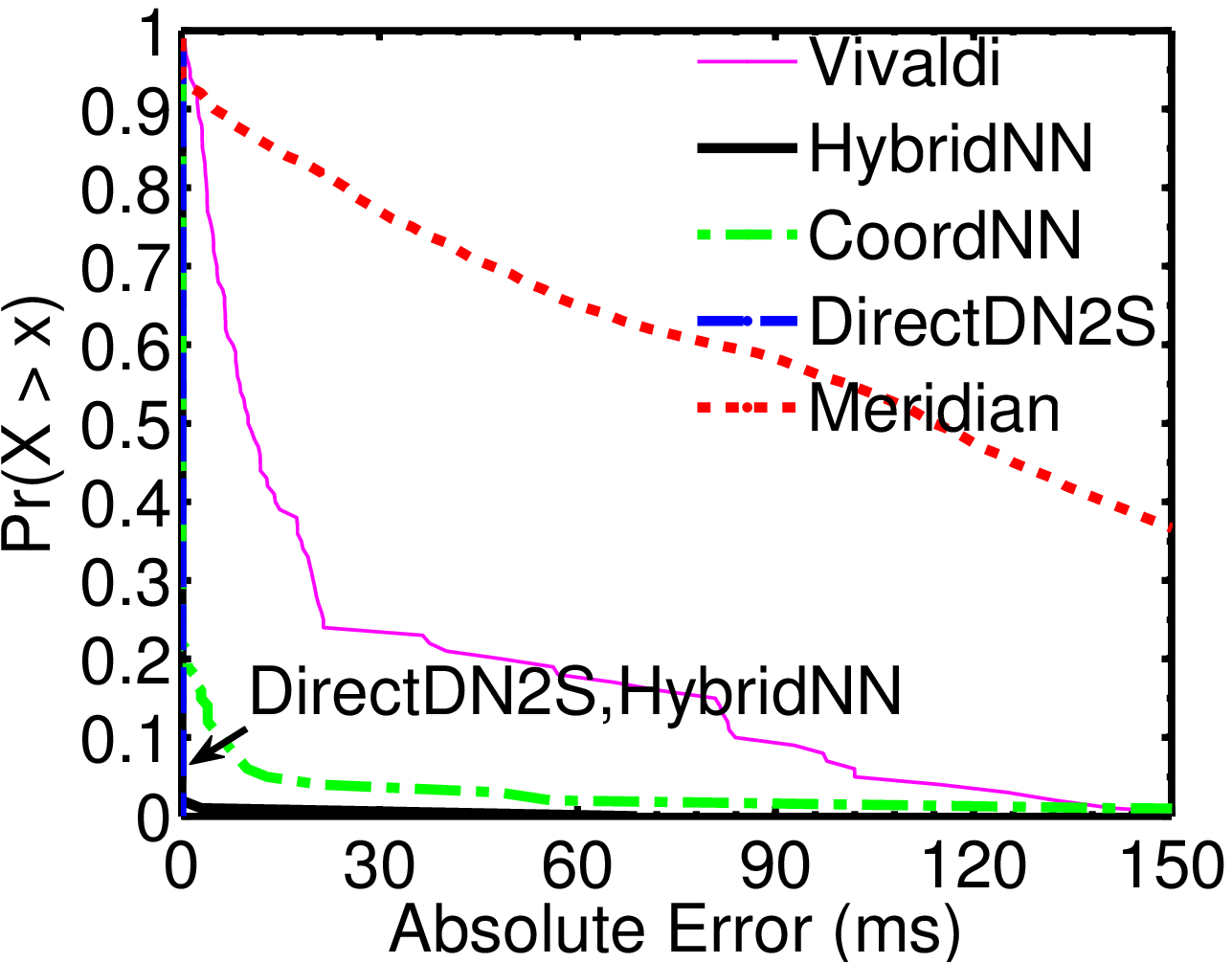}
         }
               \subfigure[Host479.]
        {
          \setlength{\epsfxsize}{.44\hsize}
          \epsffile{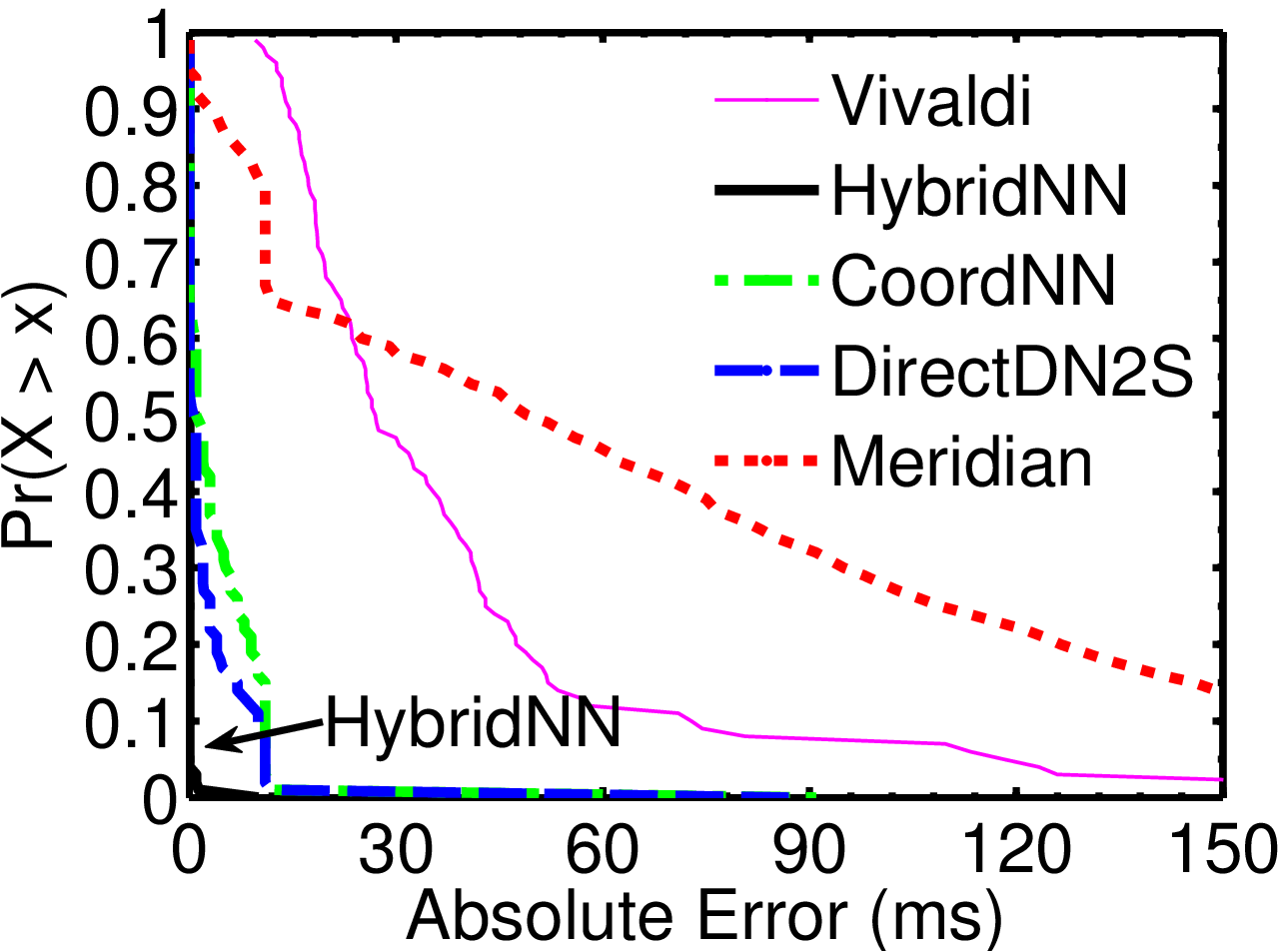}
         }
     \caption{The CCDFs of absolute errors.}
     \label{fig:AbsoErr}
\end{figure}

\textbf{Relative Error}. Fig~\ref{fig:relErrDist} shows the
relative errors of DNNS algorithms. The results are consistent
with those of the absolute errors. DirectDN2S achieves near-zero
relative errors for most DNNS queries on all data sets except
Host479. HybridNN and DirectDN2S have similar accuracy, while
HybridNN is more accurate than DirectDN2S on Host479. Furthermore,
CoordNN is less accurate than HybridNN, while Meridian and Vivaldi
are less accurate than DirectDN2S, HybridNN and CoordNN.

\begin{figure}[tp]
     \centering
            \subfigure[DNS1143.]
        {
          \setlength{\epsfxsize}{.44\hsize}
          \epsffile{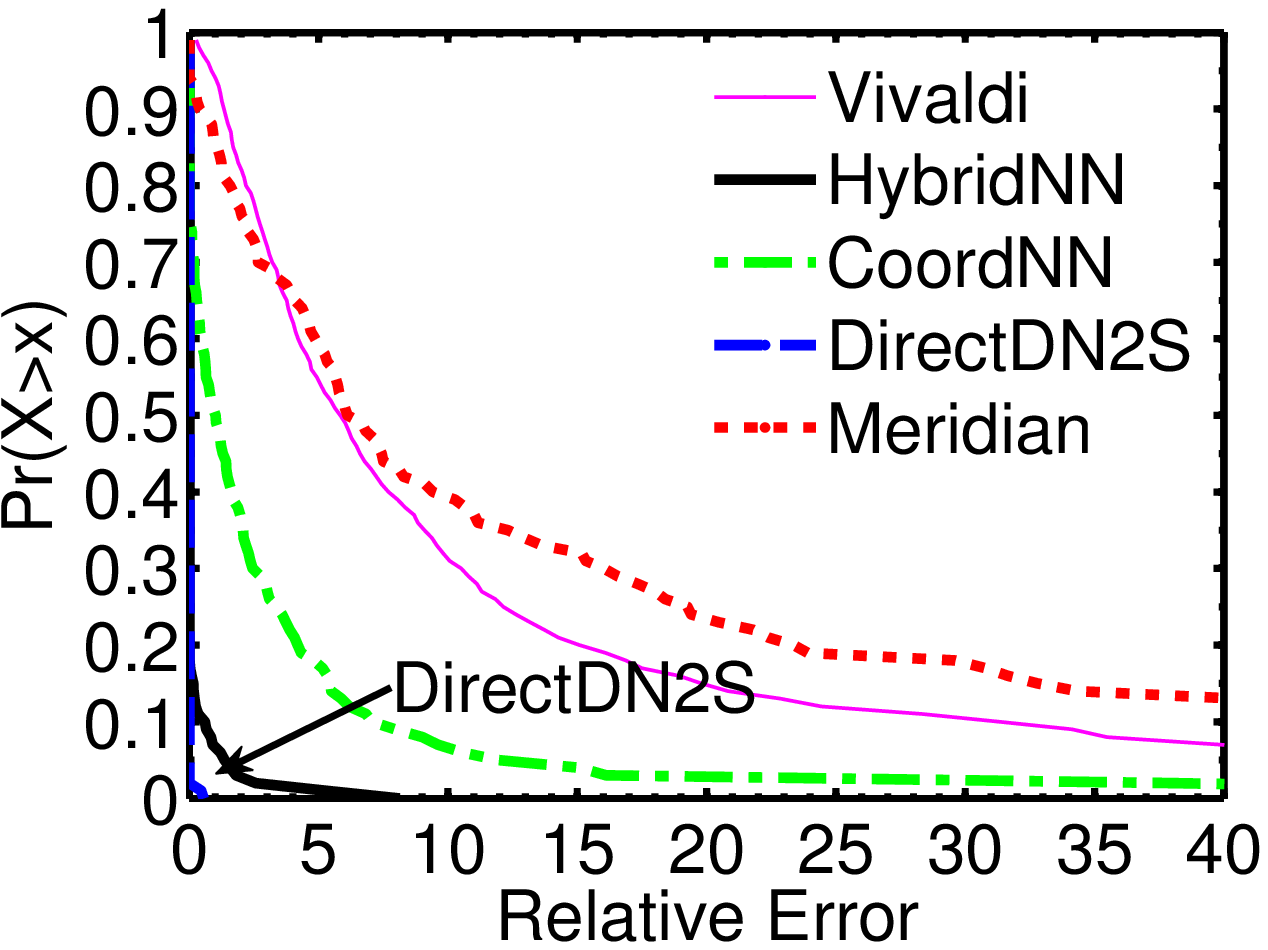}
         }
               \subfigure[DNS2500.]
        {
          \setlength{\epsfxsize}{.44\hsize}
          \epsffile{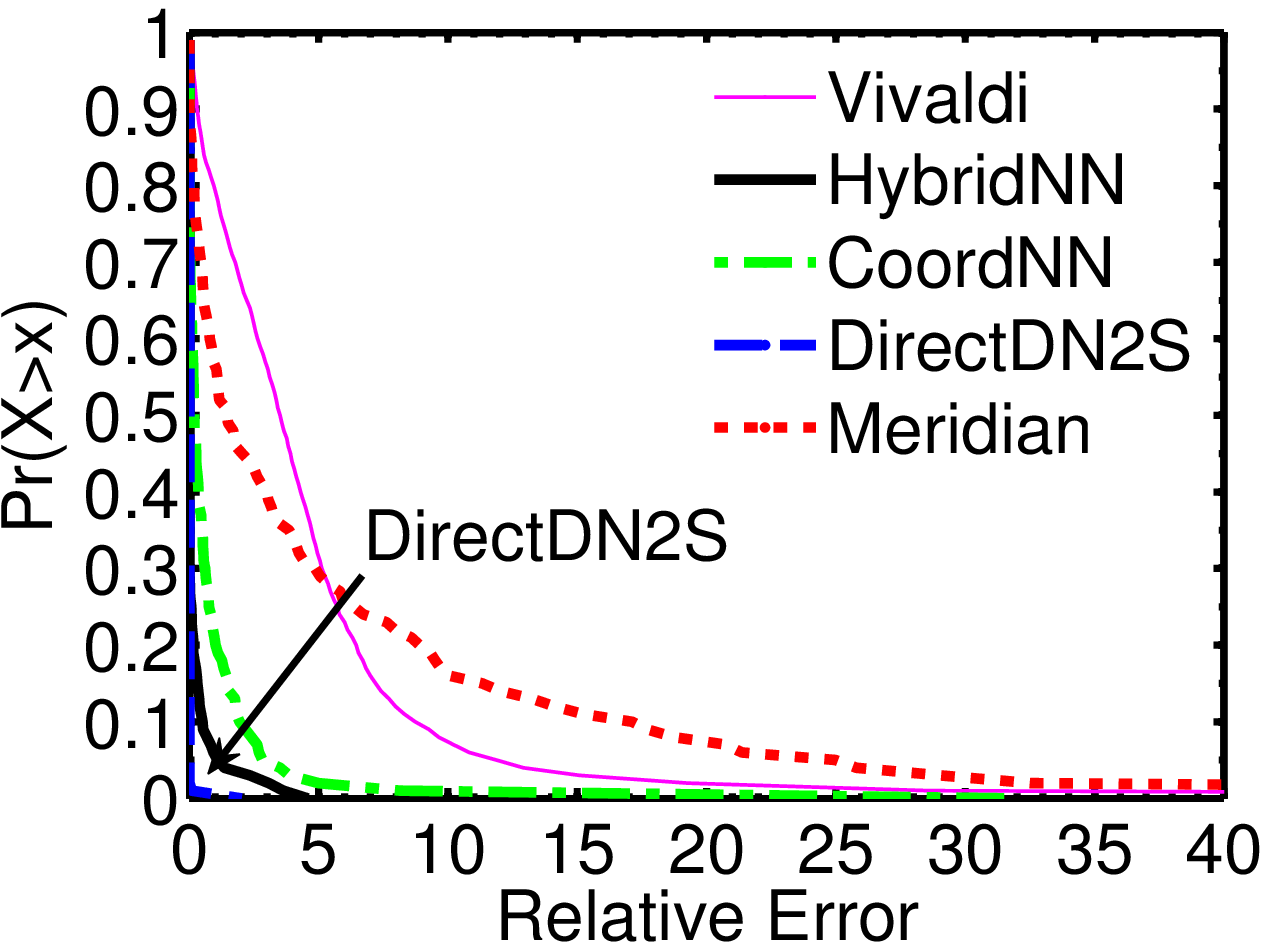}
         }
          \subfigure[DNS3997.]
        {
          \setlength{\epsfxsize}{.44\hsize}
          \epsffile{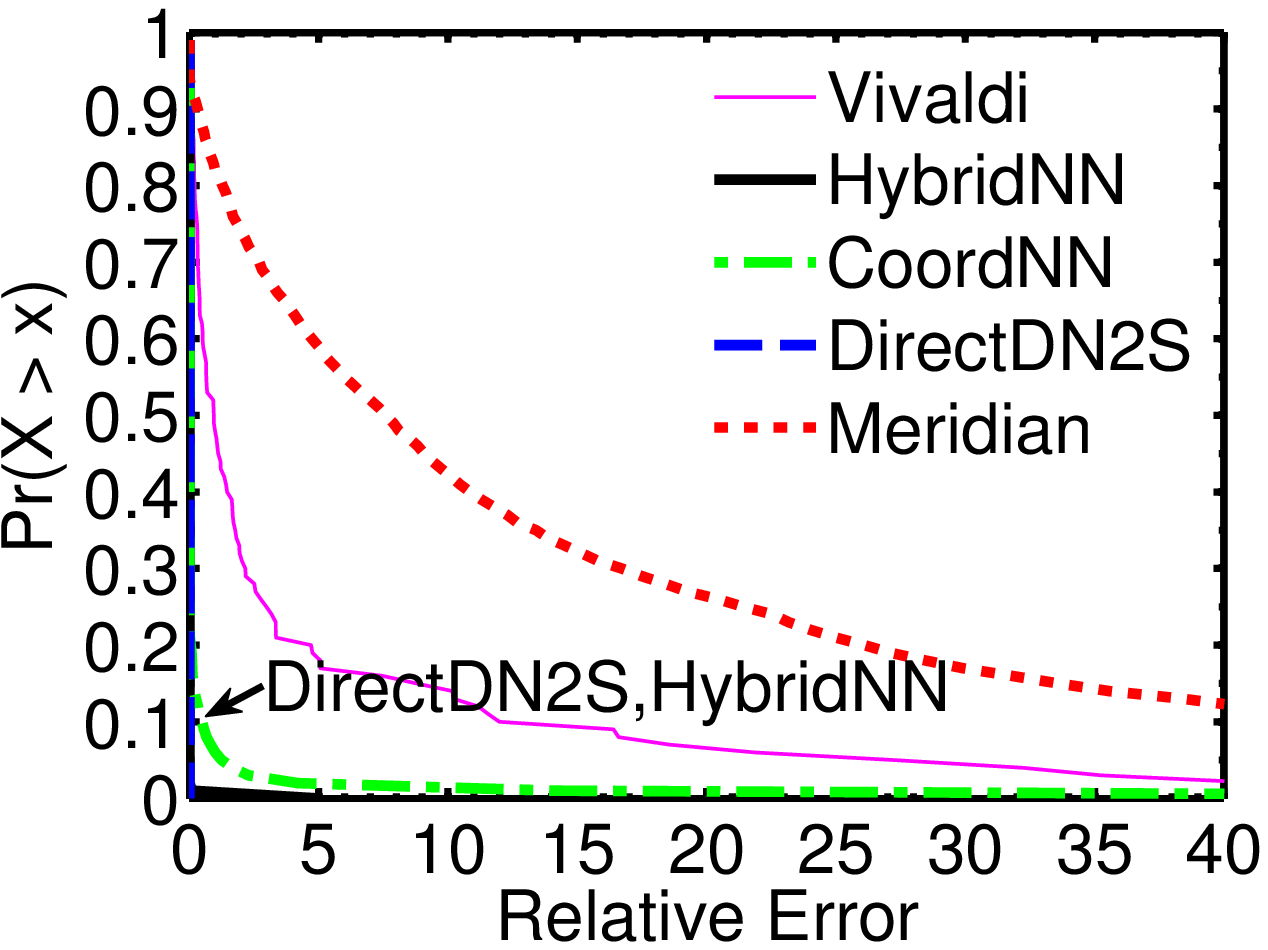}
         }
               \subfigure[Host479.]
        {
          \setlength{\epsfxsize}{.44\hsize}
          \epsffile{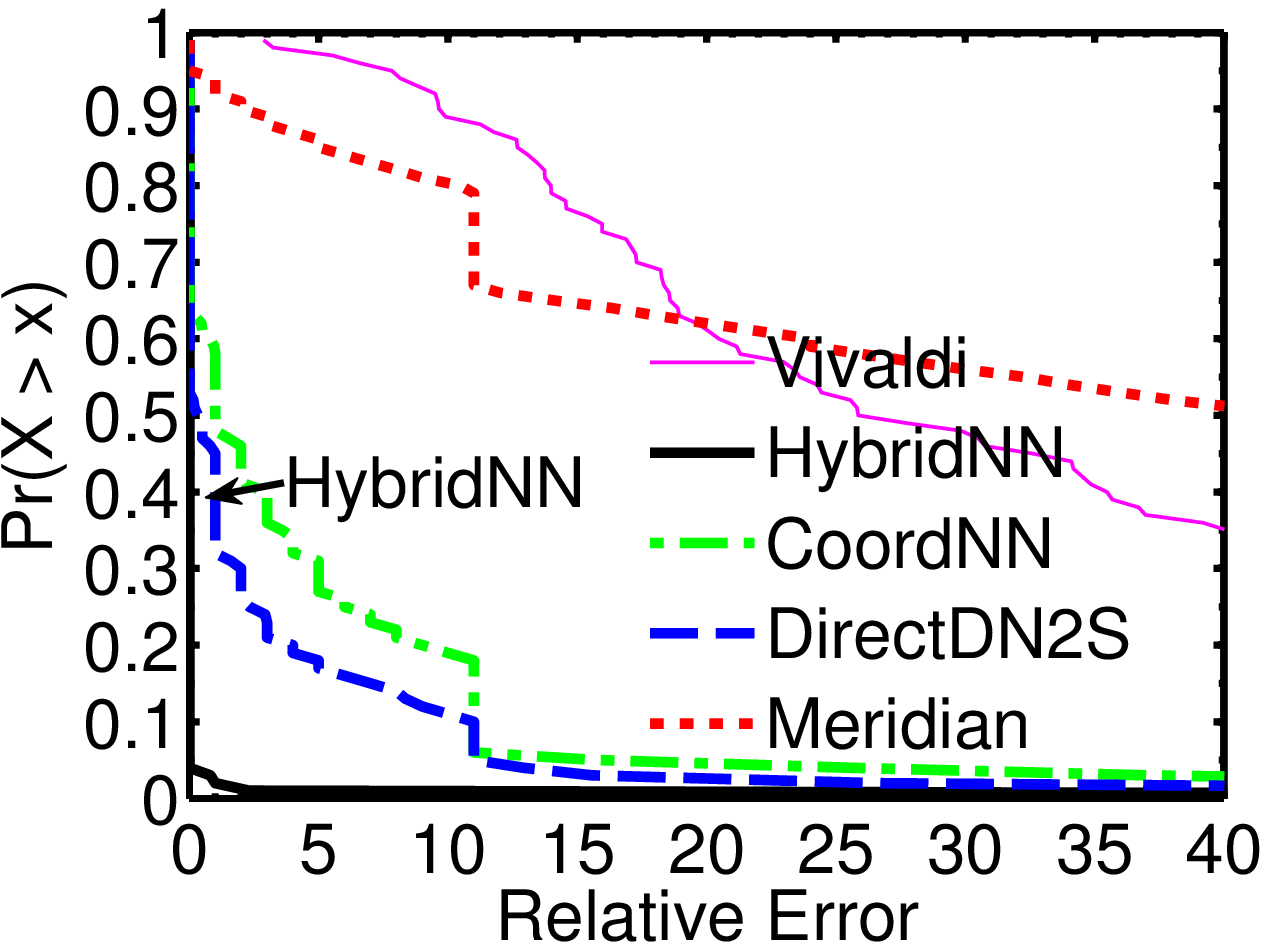}
         }
     \caption{The CCDFs of relative errors.}
     \label{fig:relErrDist}
\end{figure}

\textbf{Search hops}. Next, we quantify the distributions of the
number of search hops for DNNS algorithms, as shown in Fig
\ref{fig:searchHopDist}. Recall that the search hops are equal to
the lengths of DNNS forwarding paths minus one.

We can see that the search hops of most DNNS queries are rather
modest for all DNNS algorithms. Meridian in about 80\% of the
cases has 2 search hops. While HybridNN and DirectDN2S in over
80\% of the cases have no more than 3.

Moreover, almost all searches for Meridian, HybridNN, DirectDN2S
are below 6 search hops. On the other hand, CoordNN has longer
search hops than Meridian, HybridNN and the DirectDN2S; and a
fraction of search hops even exceed 10 on all data sets.

\co{
\subsection{Discussions}

Next, we discuss the relation between the accuracy of the
simulated DNNS algorithms with the distributions of the search
hops (Fig \ref{fig:searchHopDist}).

First, observing that the search hops for 80\% of DNNS queries in
Meridian are just 2, implies that Meridian is prone to terminate
before finding the ground-truth nearest neighbors to the targets.
Therefore, Meridian incurs higher absolute errors and relative
errors than HybridNN, DirectDN2S and CoordNN, which all have
longer forwarding paths than Meridian.

Second, a modest number of search hops suffices for accurate DNNS
queries, as the search hops of DirectDN2S and HybridNN are mostly
below 4.

Third, CoordNN is less accurate than HybridNN and DirectDN2S,
although CoordNN has longer forwarding paths than HybridNN and
DirectDN2S. This is mainly because CoordNN can be misguided into
the local minimum owing to the inaccuracy of delay predictions by
network coordinates.

Fourth, HybridNN and DirectDN2S have similar accuracy where
HybridNN and DirectDN2S also have similar distributions of search
hops. Nevertheless, HybridNN is more accurate than DirectDN2S on
the Host479 data set. Meanwhile, we can see that HybridNN has
fewer 2-hop search paths than DirectDN2S, which may cause
DirectDN2S to be trapped into local minimum like those in
Meridian. }

\begin{figure}[tp]
     \centering
           \subfigure[DNS1143.]
        {
          \setlength{\epsfxsize}{.44\hsize}
          \epsffile{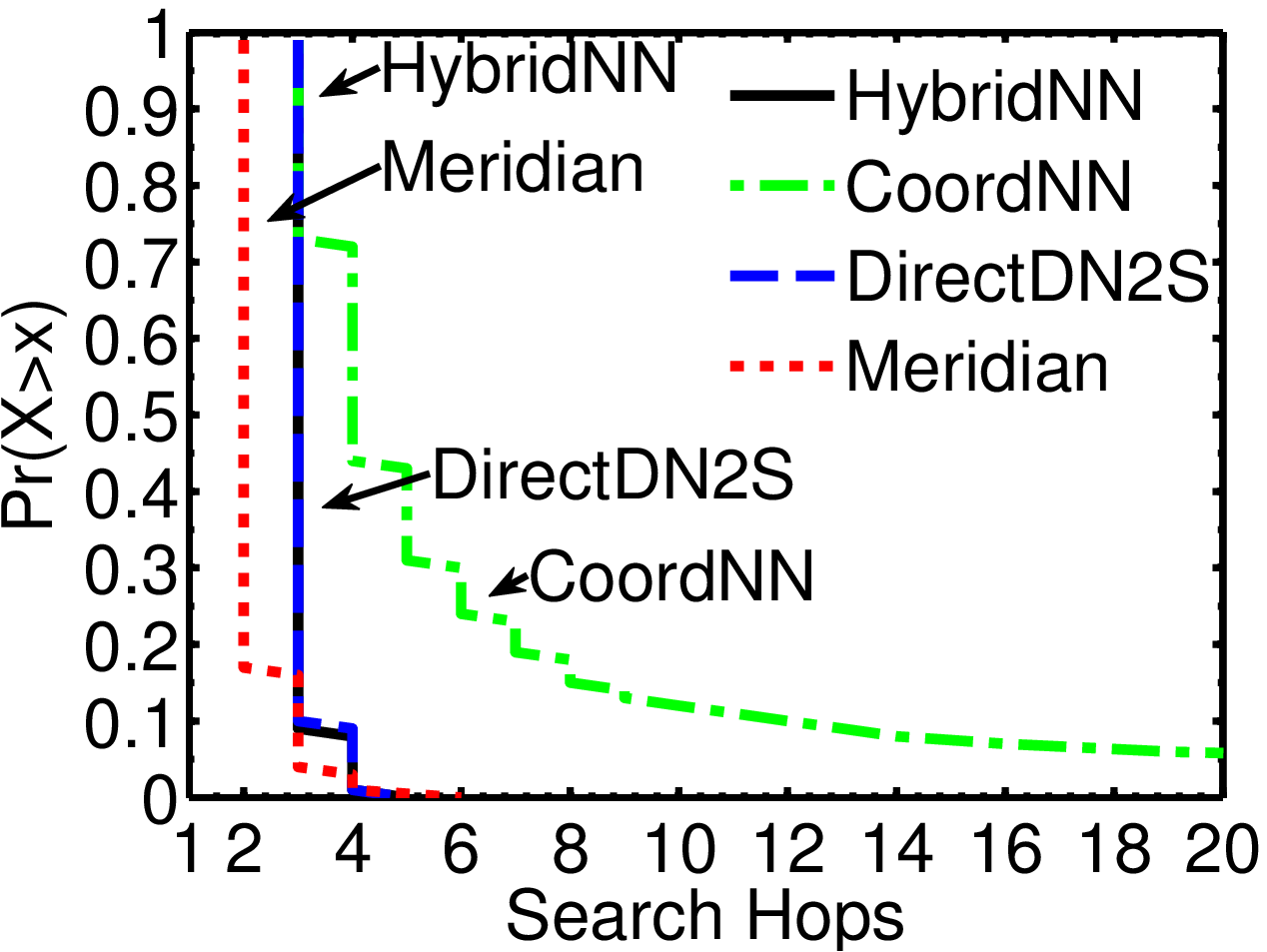}
         }
               \subfigure[DNS2500.]
        {
          \setlength{\epsfxsize}{.44\hsize}
          \epsffile{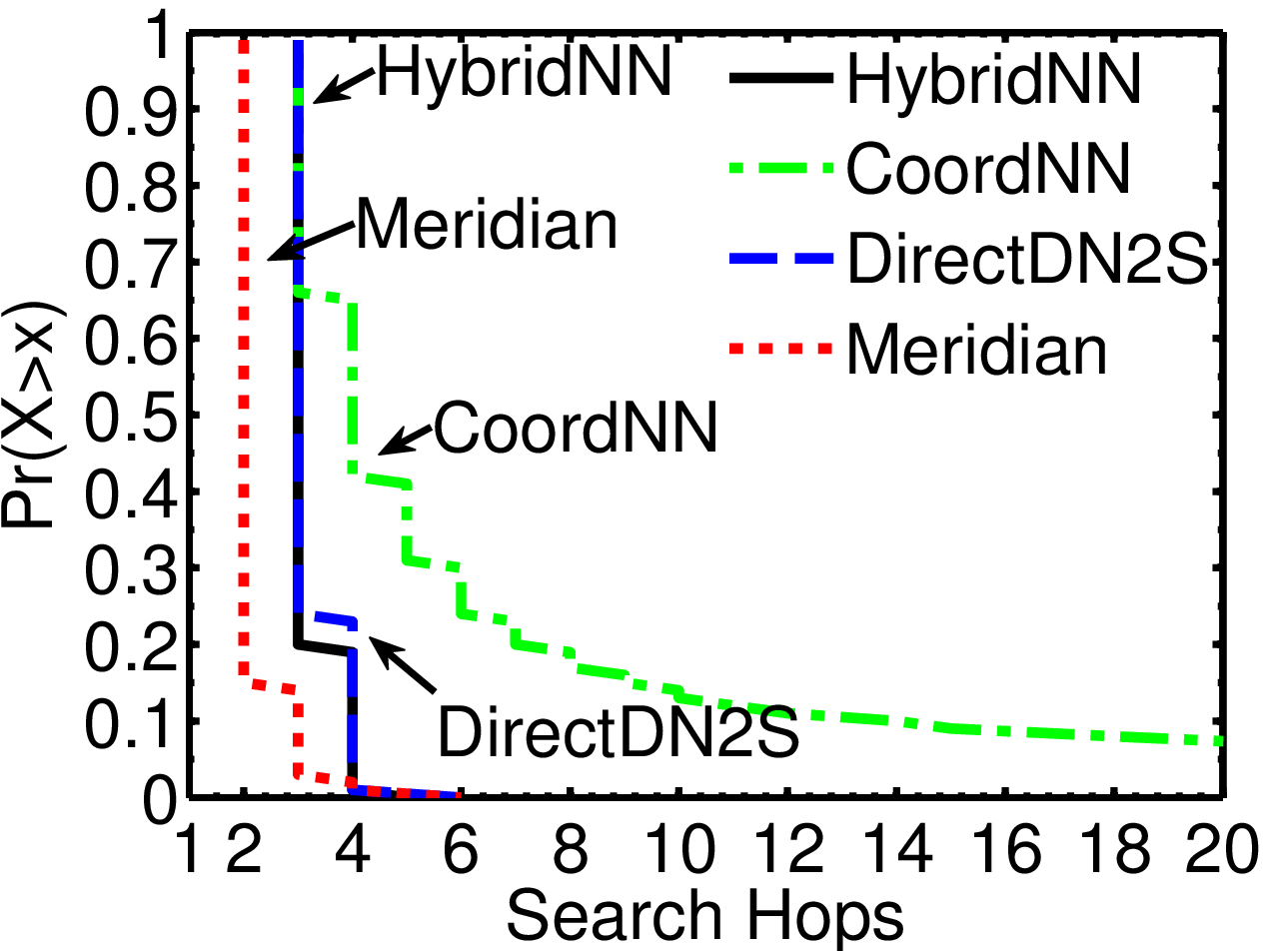}
         }
          \subfigure[DNS3997.]
        {
          \setlength{\epsfxsize}{.44\hsize}
          \epsffile{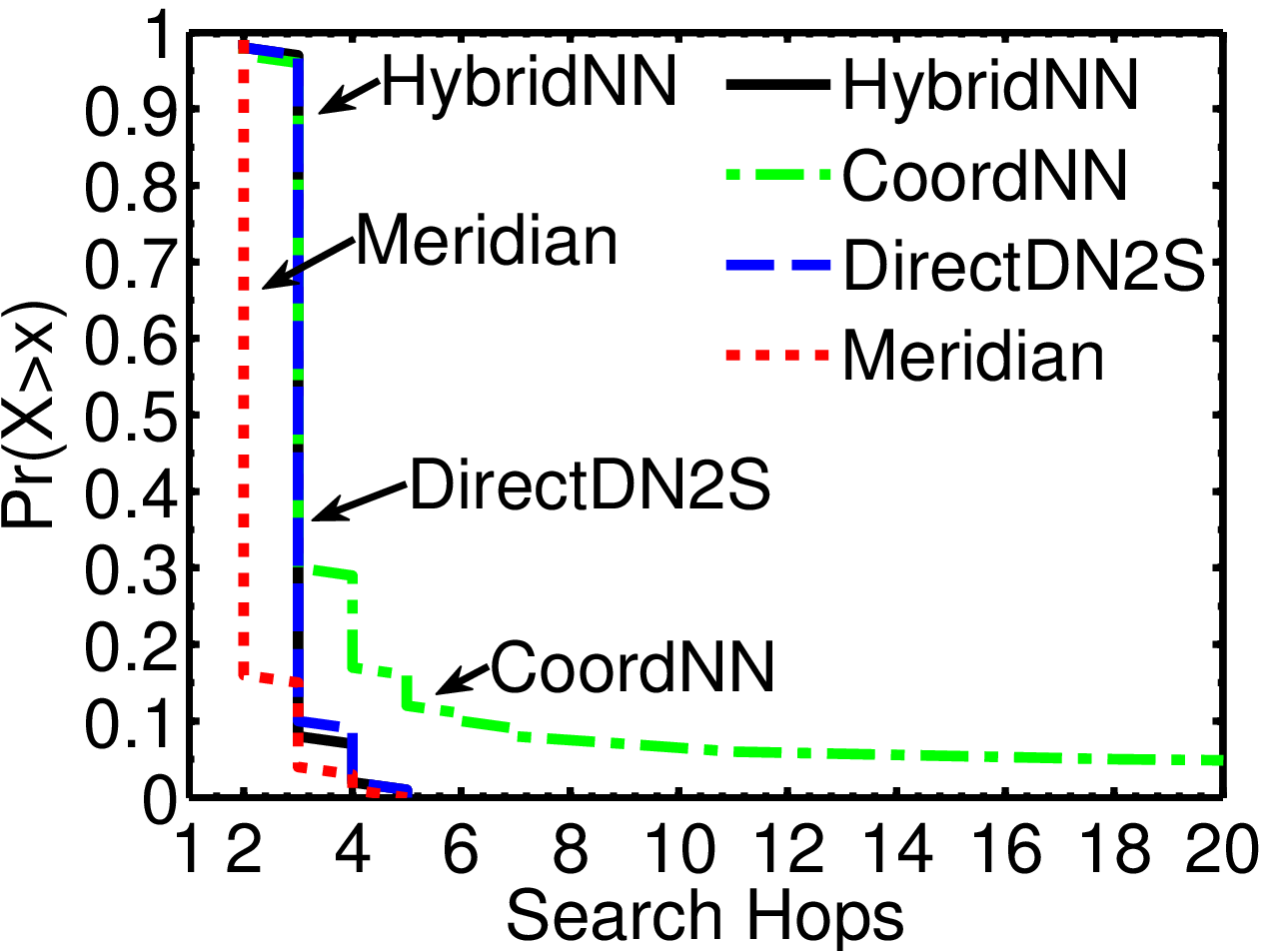}
         }
               \subfigure[Host479.]
        {
          \setlength{\epsfxsize}{.44\hsize}
          \epsffile{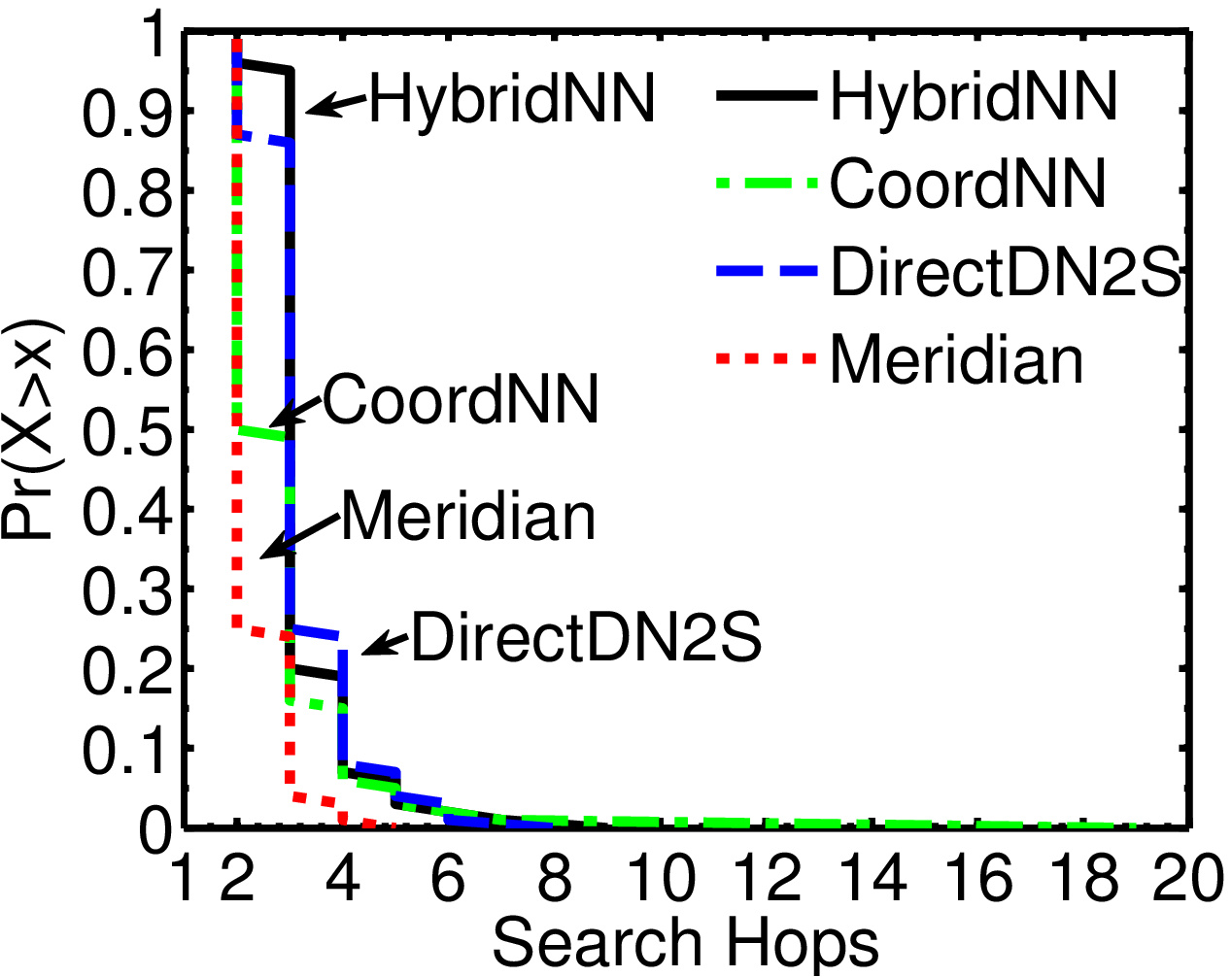}
         }
     \caption{The CCDFs of search hops.}
     \label{fig:searchHopDist}
\end{figure}

\subsection{Sensitivity of Parameters}

In this section, we evaluate the robustness of HybirdNN to the
 system size as well as the choices of system parameters.

\subsubsection{System Size $N$}

To evaluate the size of service machines on the performance of
HybridNN, we evaluate the performance of HybridNN by increasing
the size of service machines. We select target machines randomly
from all nodes, including the clients and the service machines, as
the size of clients shrinks when increasing the percentage of
service machines.  Fig.~\ref{fig:NumberOfServiceNodes} shows the
performance of HybridNN with increasing the percentage of service
nodes. HybridNN achieves similar accuracy when the size of service
nodes increase compared to clients. Therefore, HybridNN is quite
robust to the different scales of systems. On the other hand, the
query loads of HybridNN increase slowly, for example, HybridNN
nearly double the loads when the percentage of service nodes
reaches 1.

\begin{figure}[htbp]
     \centering
          \subfigure[DNS1143.]
        {
          \setlength{\epsfxsize}{.44\hsize}
          \epsffile{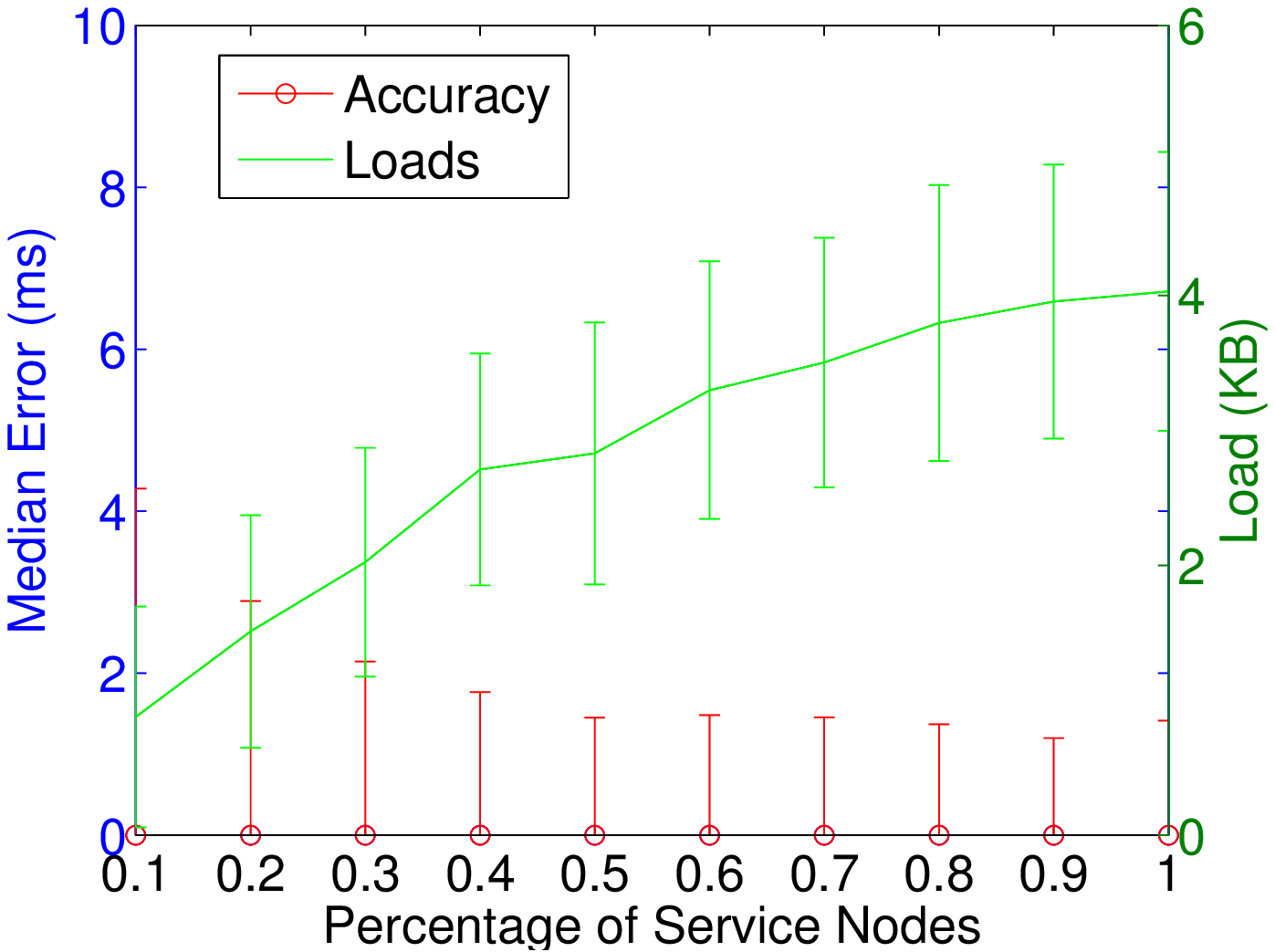}
          \label{fig:B}
         }
               \subfigure[DNS2500.]
        {
          \setlength{\epsfxsize}{.44\hsize}
          \epsffile{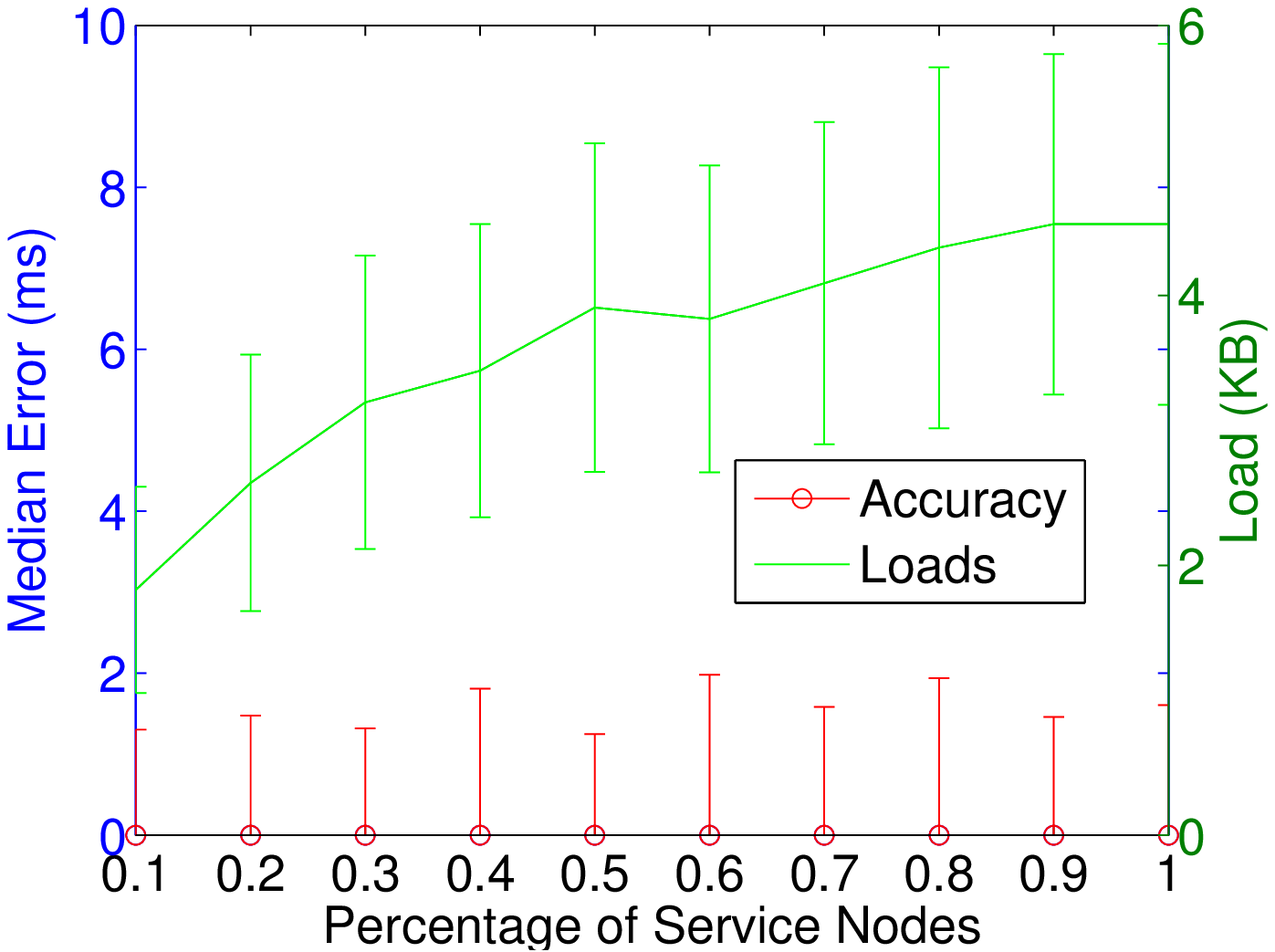}
          \label{fig:B}
         }
                        \subfigure[DNS3997.]
        {
          \setlength{\epsfxsize}{.44\hsize}
          \epsffile{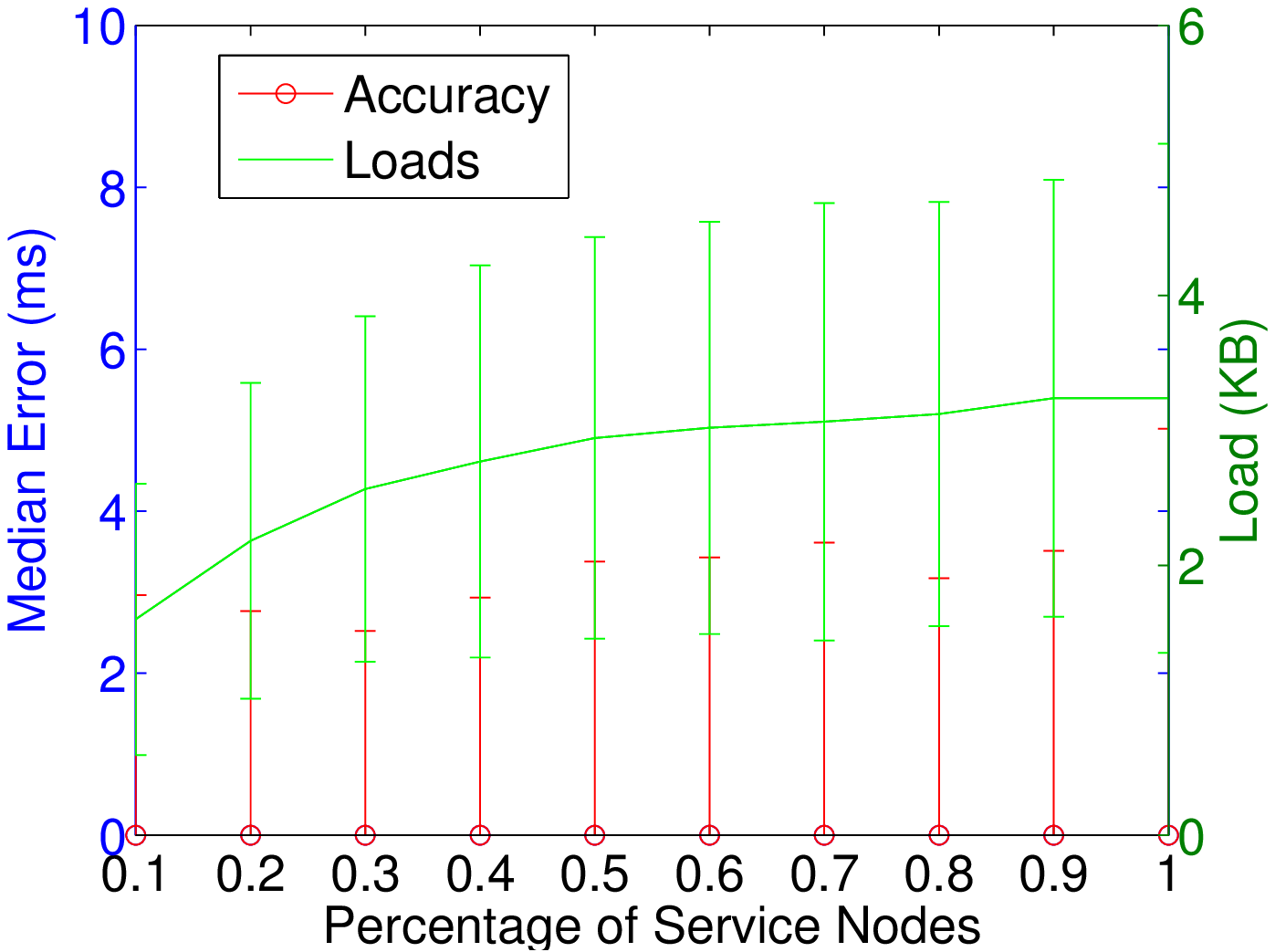}
          \label{fig:B}
         }
               \subfigure[Host479.]
        {
          \setlength{\epsfxsize}{.44\hsize}
          \epsffile{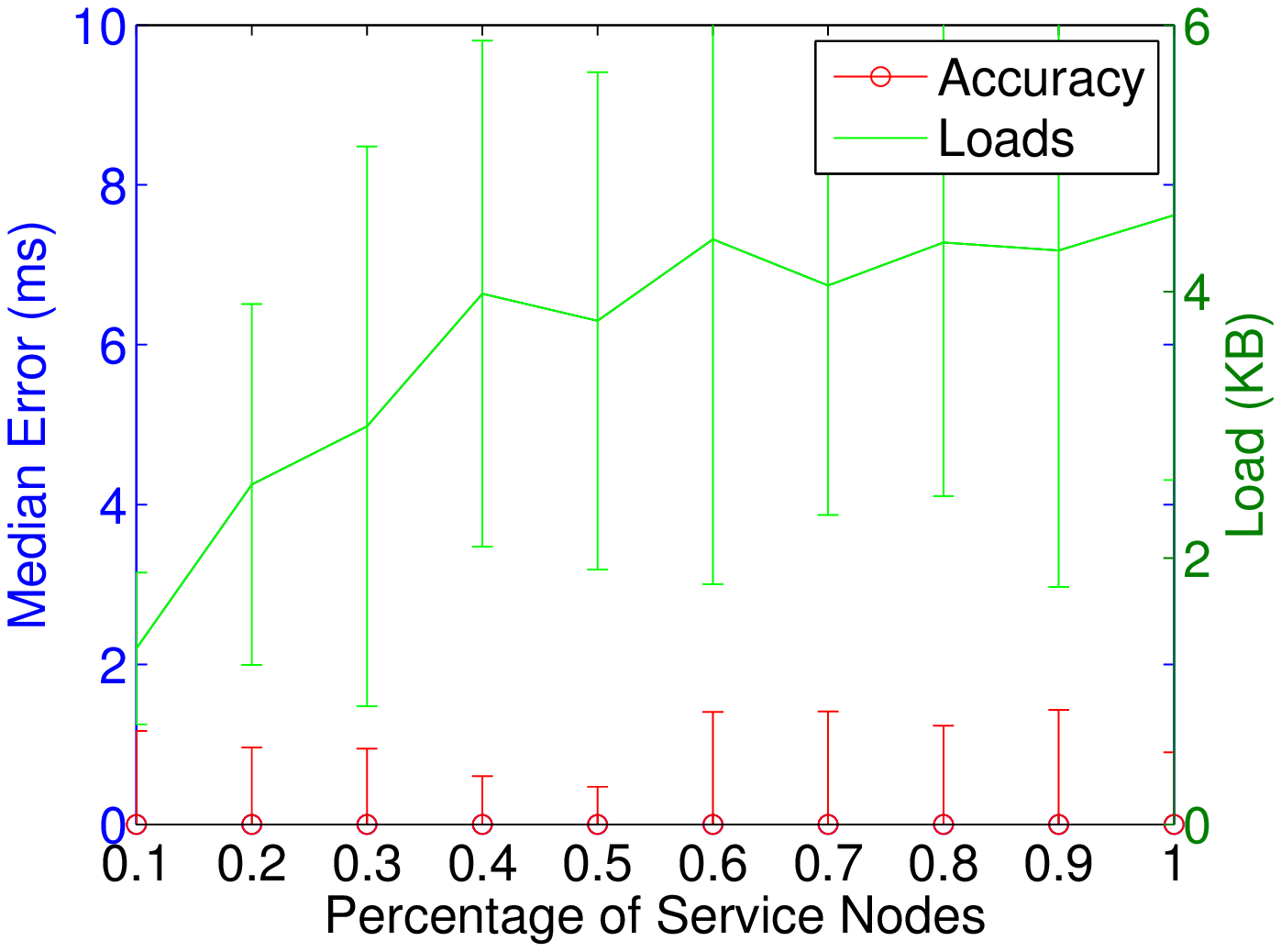}
          \label{fig:B}
         }
     \caption{Size of Service Nodes.}
     \label{fig:NumberOfServiceNodes}
\end{figure}

\subsubsection{Inframetric $\rho$}

Fig.~\ref{fig:rhoSen} shows the accuracy and loads as the
increment of Inframetric parameter $\rho$. The accuracy of
HybridNN is insensitive to choices of $\rho$.  This is because for
most delays, its $\rho$-edge metrics are quite lower. Therefore,
with lower $\rho$ we can cover possible best next-hop neighbors
for DNNS queries. Furthermore, although larger $\rho$ increases
the size of possible next-hop candidate neighbors, the loads of
DNNS queries of HybridNN keep stable for different $\rho$, due to
that we use nearly constant-sized next-hop nodes. Besides, we can
see the standard deviations of errors are quite low for most data
sets.

\co{ Since we add constants to delays of TIVMin, the loads of DNNS
queries are enlarged due to broader ranges of candidate neighbors,
and the presence possibility of closer nodes are reduced, as many
nodes are put into identical rings by the similarity of shifted
delays.}

\begin{figure}[htbp]
     \centering
          \subfigure[DNS1143.]
        {
          \setlength{\epsfxsize}{.44\hsize}
          \epsffile{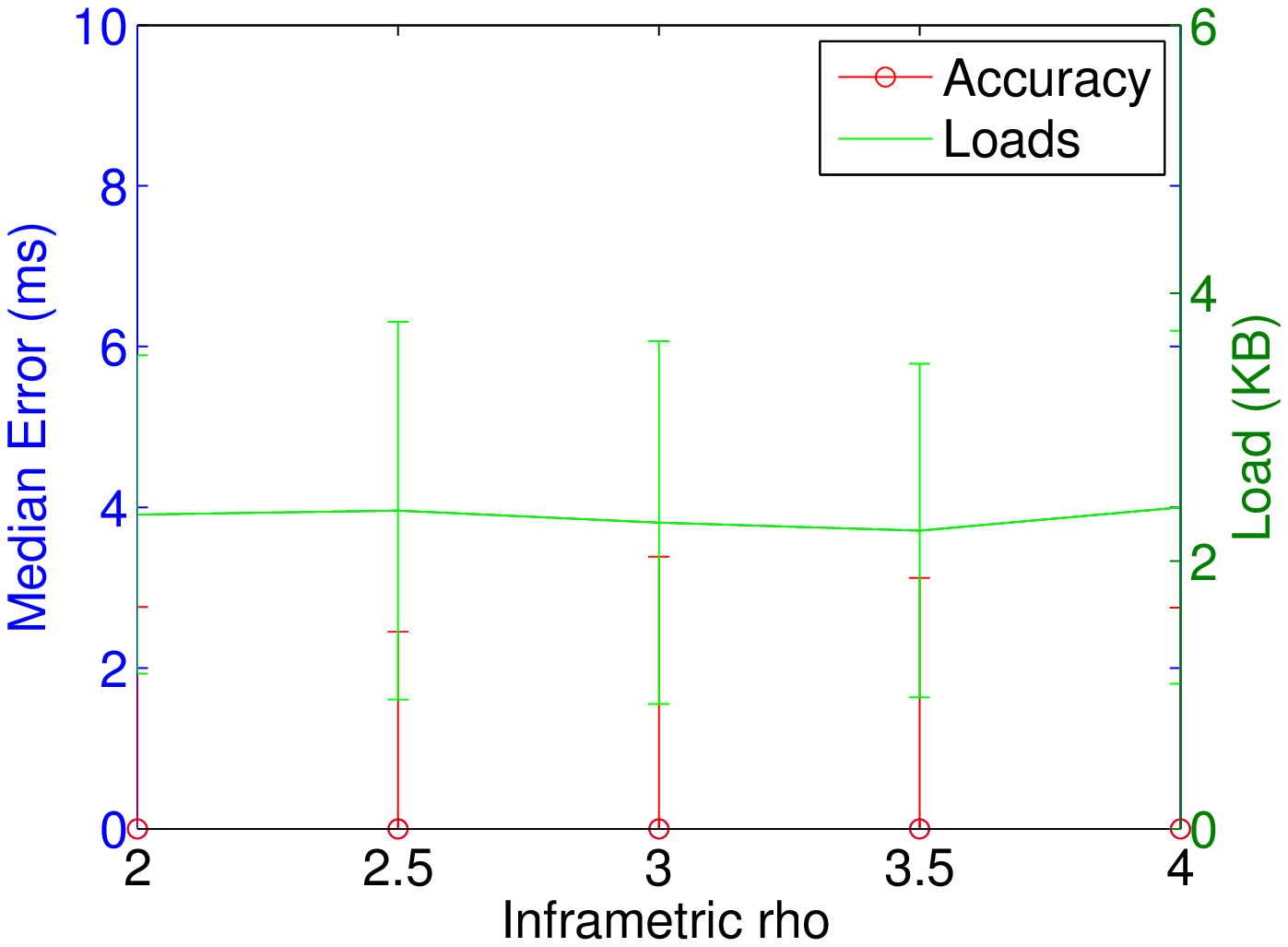}
          \label{fig:B}
         }
      \subfigure[DNS2500.]
        {
          \setlength{\epsfxsize}{.44\hsize}
          \epsffile{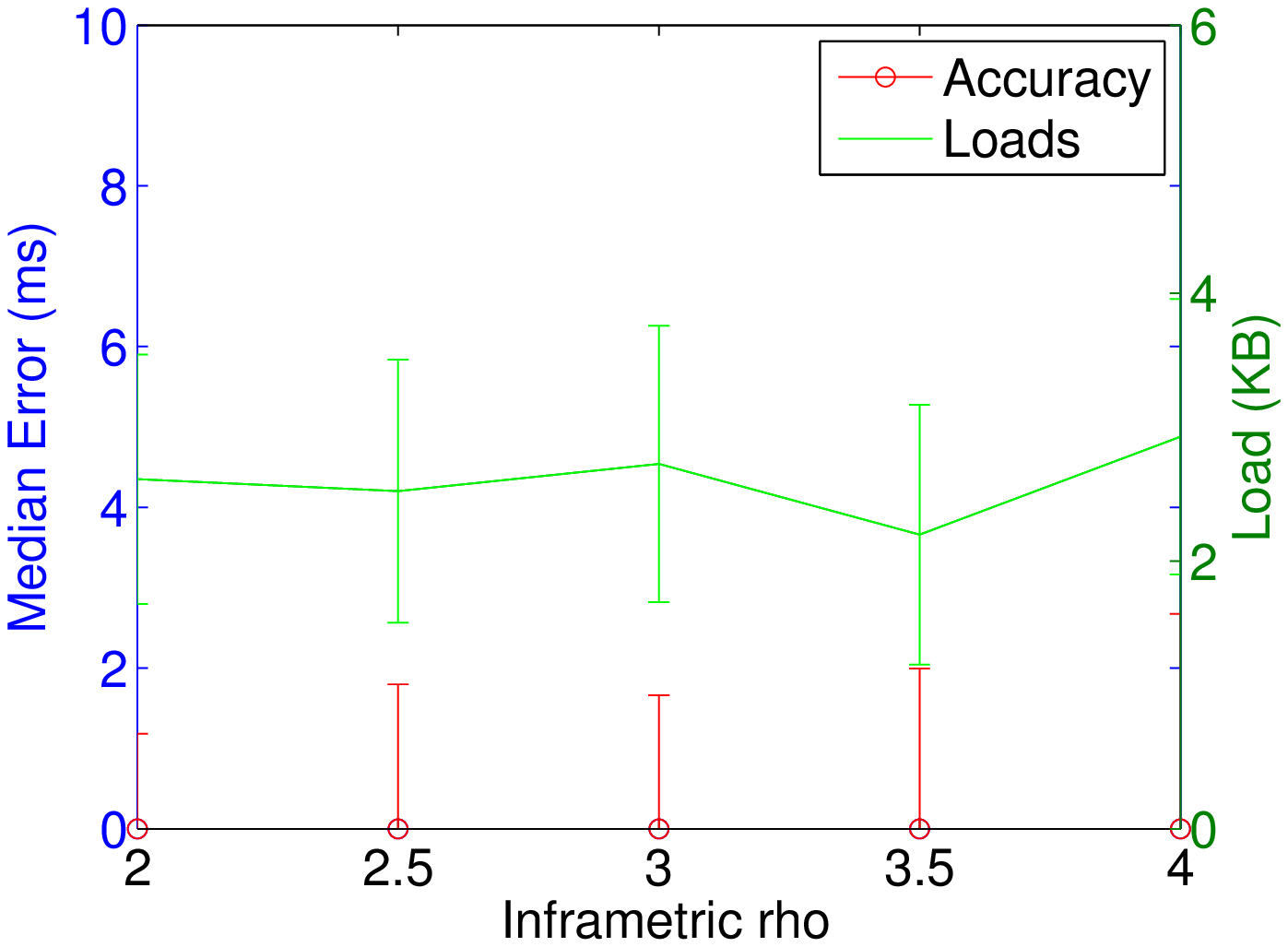}
          \label{fig:B}
         }
               \subfigure[DNS3997.]
        {
          \setlength{\epsfxsize}{.44\hsize}
          \epsffile{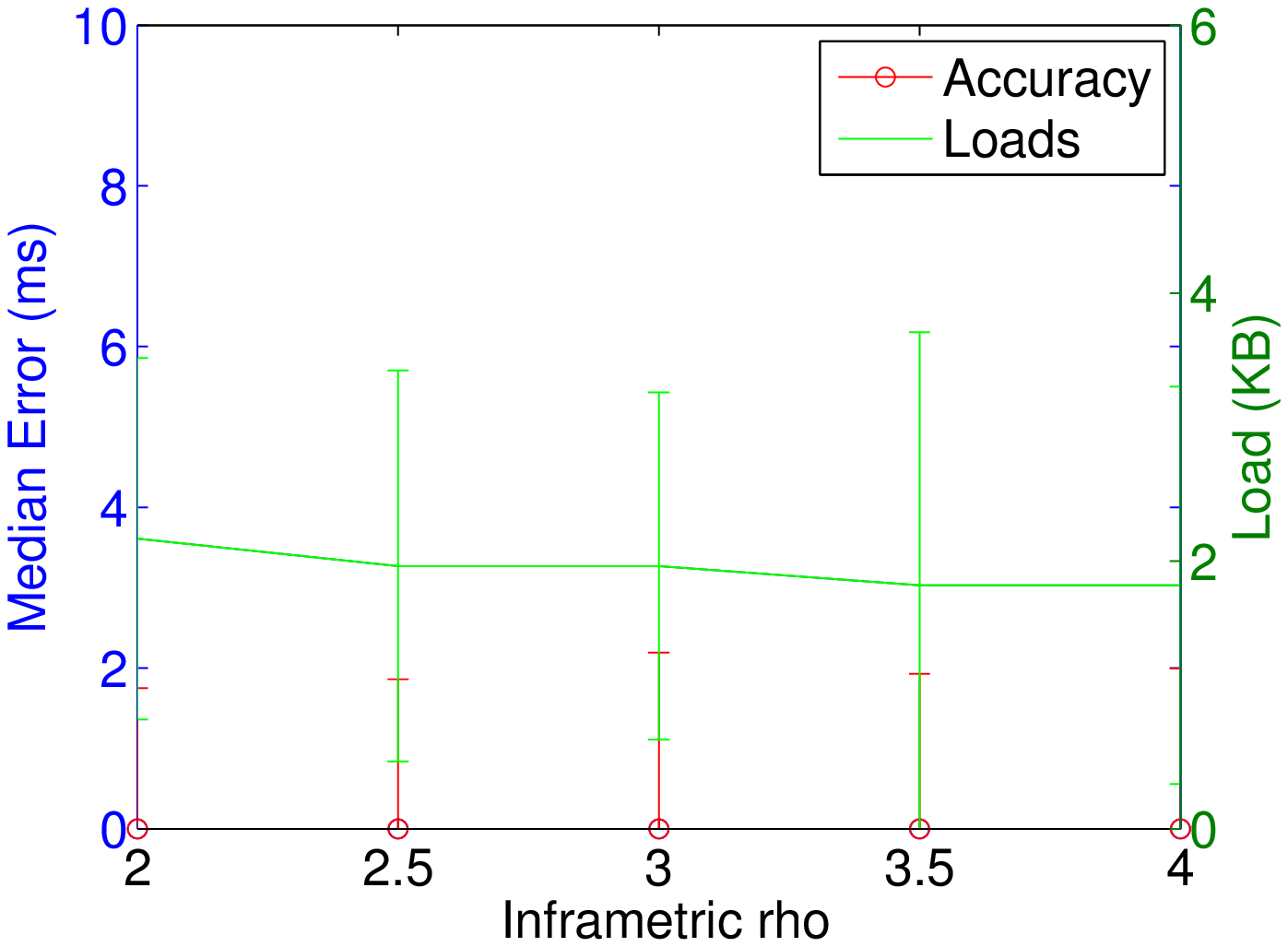}
          \label{fig:B}
         }
               \subfigure[Host479.]
        {
          \setlength{\epsfxsize}{.44\hsize}
          \epsffile{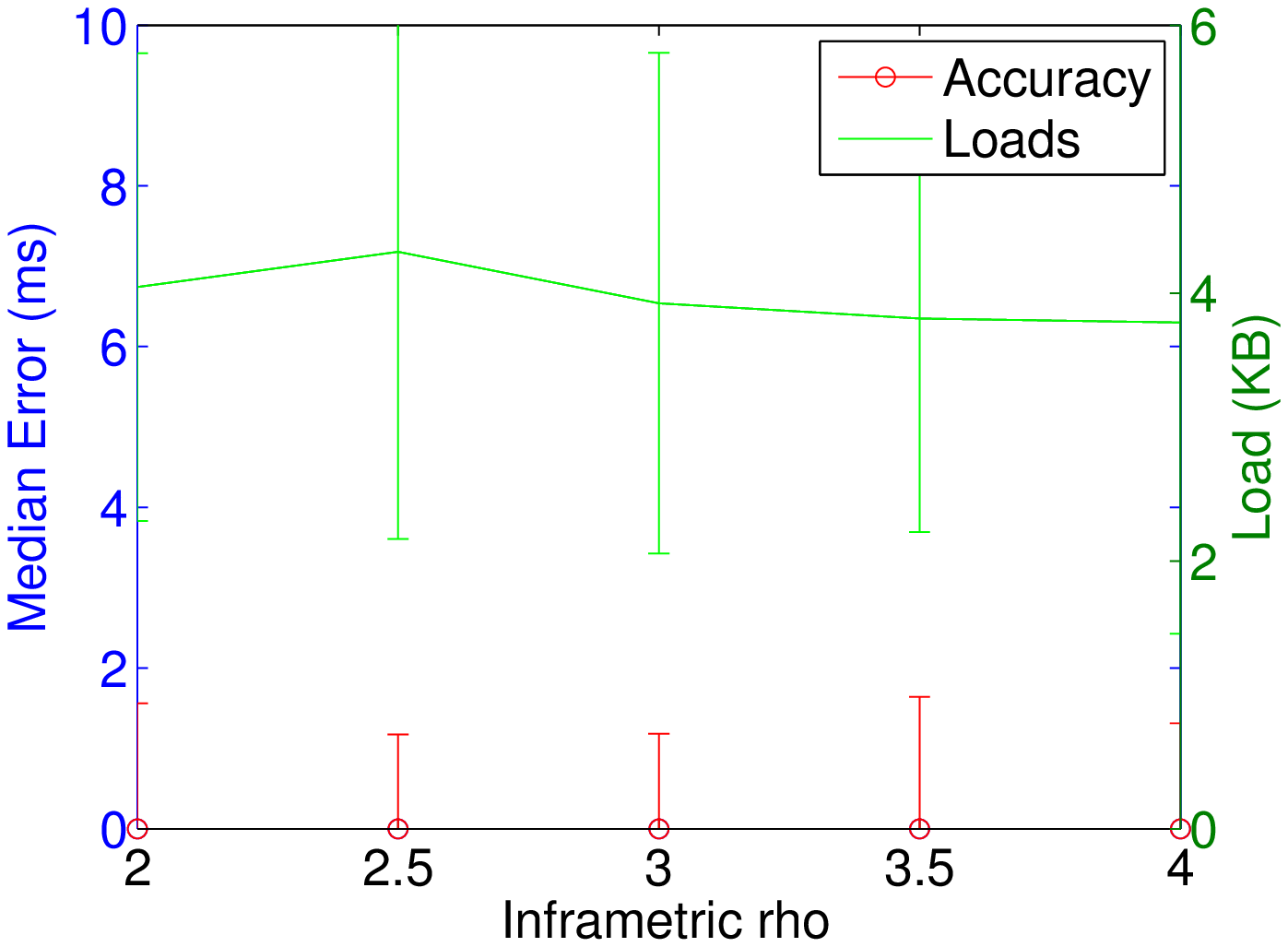}
          \label{fig:B}
         }
     \caption{Inframetric $\rho$.}
     \label{fig:rhoSen}
\end{figure}

\subsubsection{Non-Empty Threshold $\tau$}

Fig.~\ref{fig:nmTh} shows the accuracy and loads as the increment
of Non-empty thresholds for pruning candidate neighbors for
next-hop nodes. As the increment of non-empty thresholds for
pruning candidate neighbors that have too few rings containing
nodes, the standard deviation of HybridNN is reduced before the
threshold reaches 4, then increases after the threshold is over 4,
and the median errors are increased when the non-empty threshold
exceed 8. Besides, the loads are reduced when the non-empty
thresholds increase. Therefore, selecting modest-sized non-empty
thresholds (e.g., 4) can keep accuracy and reduce loads.

\begin{figure}[htbp]
     \centering
        \subfigure[DNS2500.]
        {
          \setlength{\epsfxsize}{.44\hsize}
          \epsffile{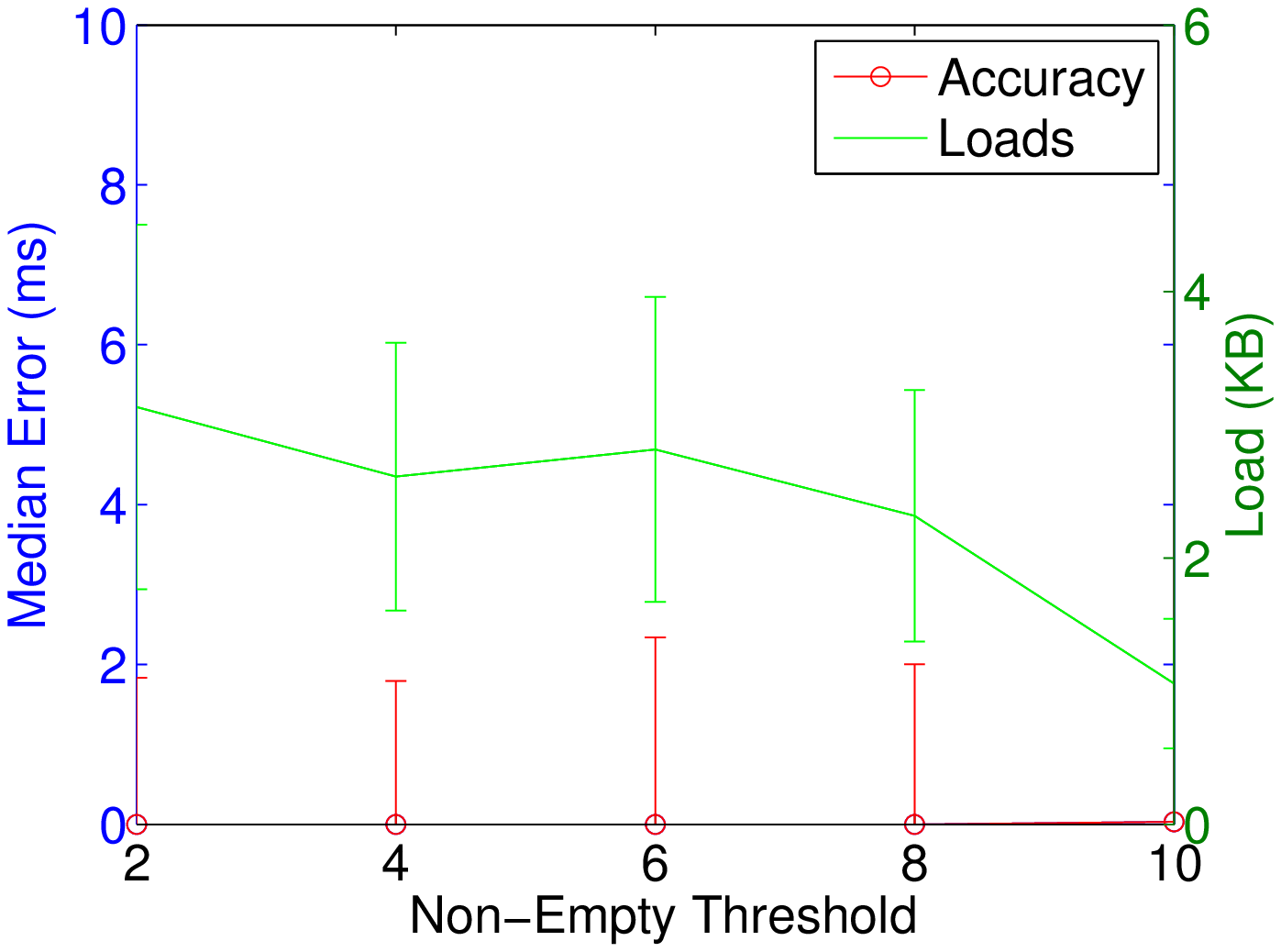}
          \label{fig:B}
         }
               \subfigure[DNS1143.]
        {
          \setlength{\epsfxsize}{.44\hsize}
          \epsffile{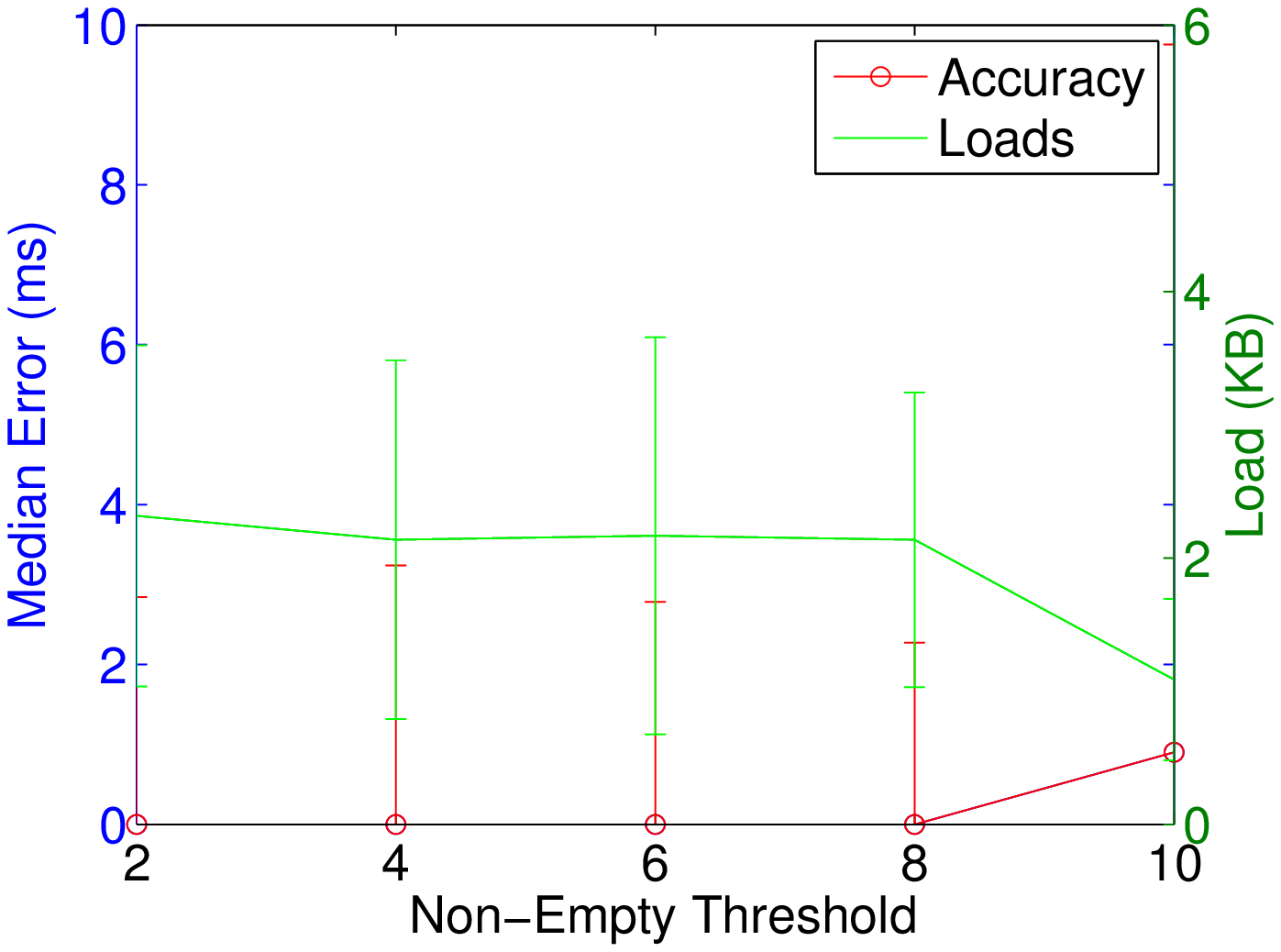}
          \label{fig:B}
         }
          \subfigure[DNS3997.]
        {
          \setlength{\epsfxsize}{.44\hsize}
          \epsffile{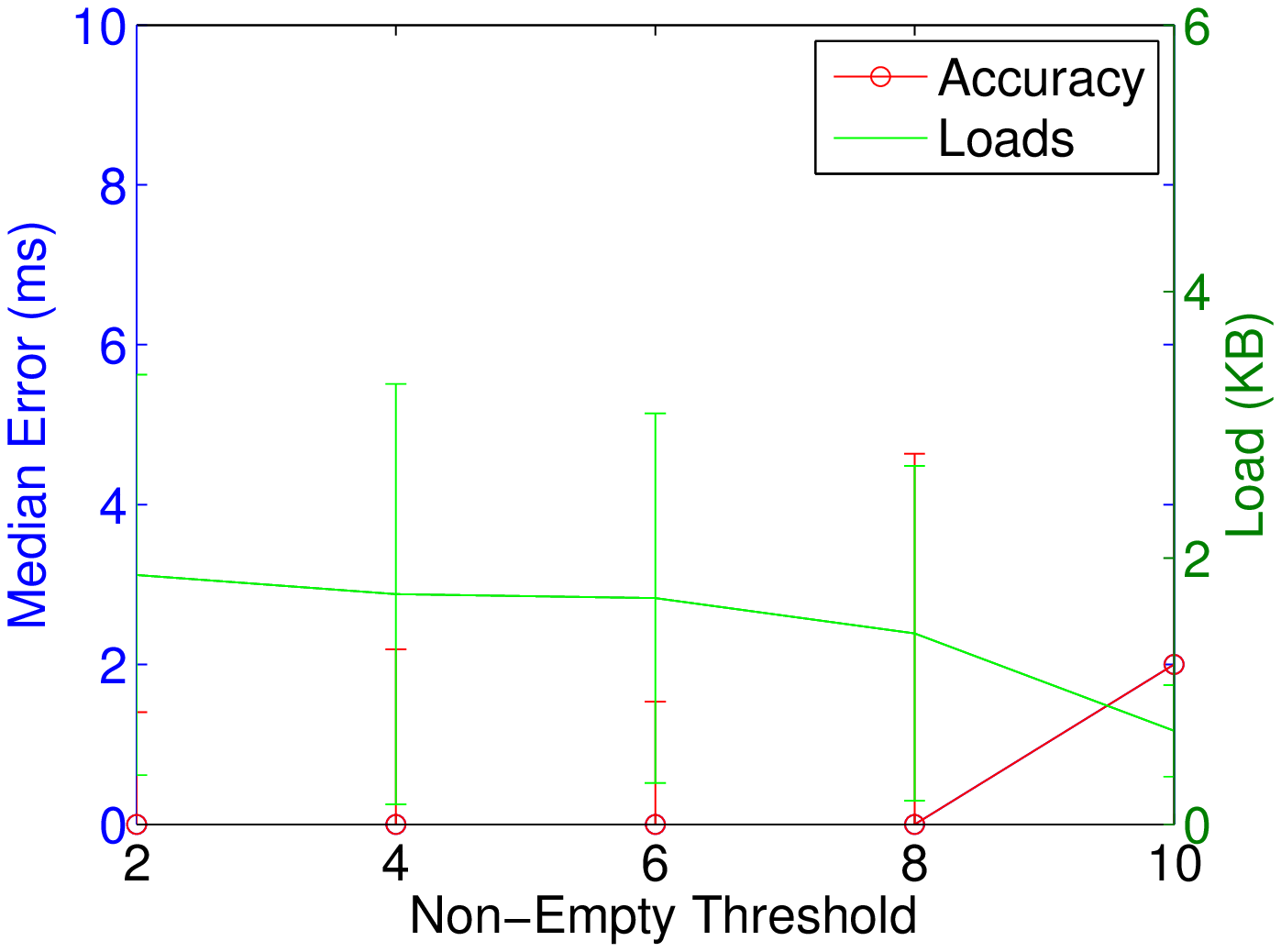}
          \label{fig:B}
         }
               \subfigure[Host479.]
        {
          \setlength{\epsfxsize}{.44\hsize}
          \epsffile{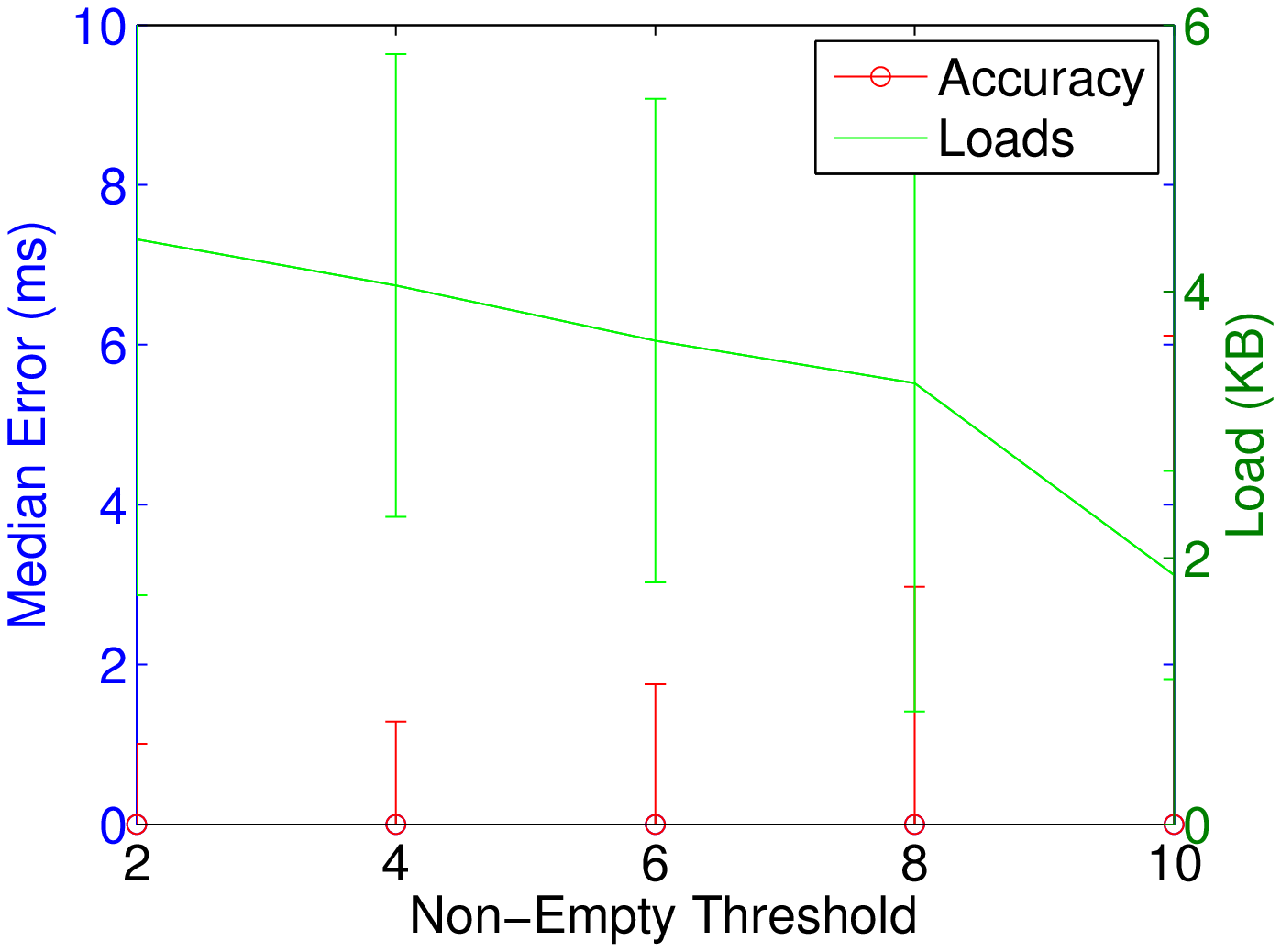}
          \label{fig:B}
         }
     \caption{Non-Empty Threshold.}
     \label{fig:nmTh}
\end{figure}

\co{

\subsubsection{Cutoff threshold}

Fig.~\ref{fig:AB} plots the accuracy and loads as the increment of
cutoff thresholds of the search process. The standard deviations
increase with increasing cutoffs; while the standard deviations
are reduced. But the median errors keep stable at 0. Since we
adaptively permit the continuation of DNNS queries when the delay
of next-hop node to the target is lower. Therefore, the cutoff
threshold dos not

\begin{figure}[htbp]
     \centering
        \subfigure[DNS1143.]
        {
          \setlength{\epsfxsize}{.44\hsize}
          \epsffile{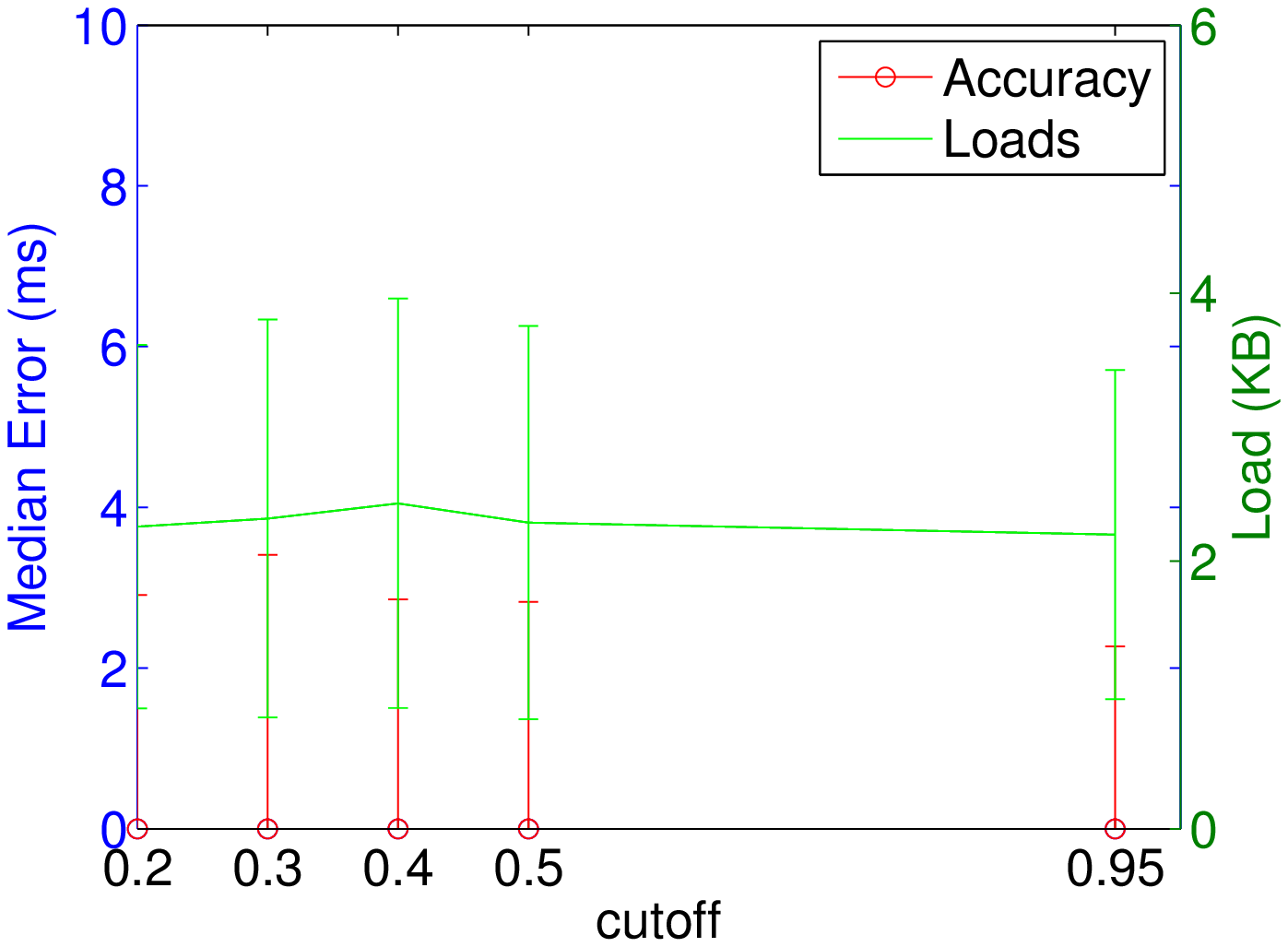}
          \label{fig:B}
         }
    \subfigure[DNS2500.]
        {
          \setlength{\epsfxsize}{.44\hsize}
          \epsffile{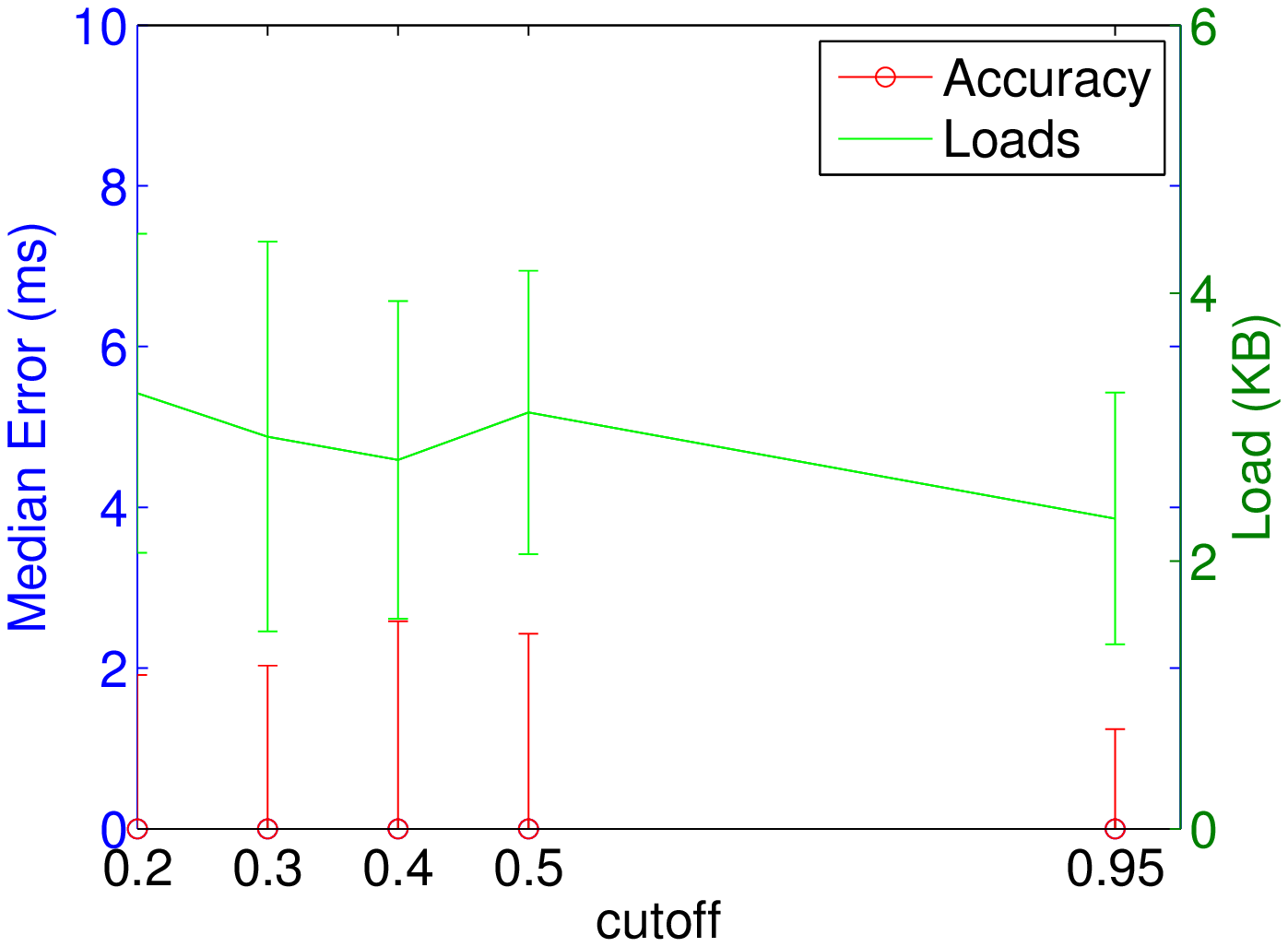}
          \label{fig:B}
         }
               \subfigure[DNS3997.]
        {
          \setlength{\epsfxsize}{.44\hsize}
          \epsffile{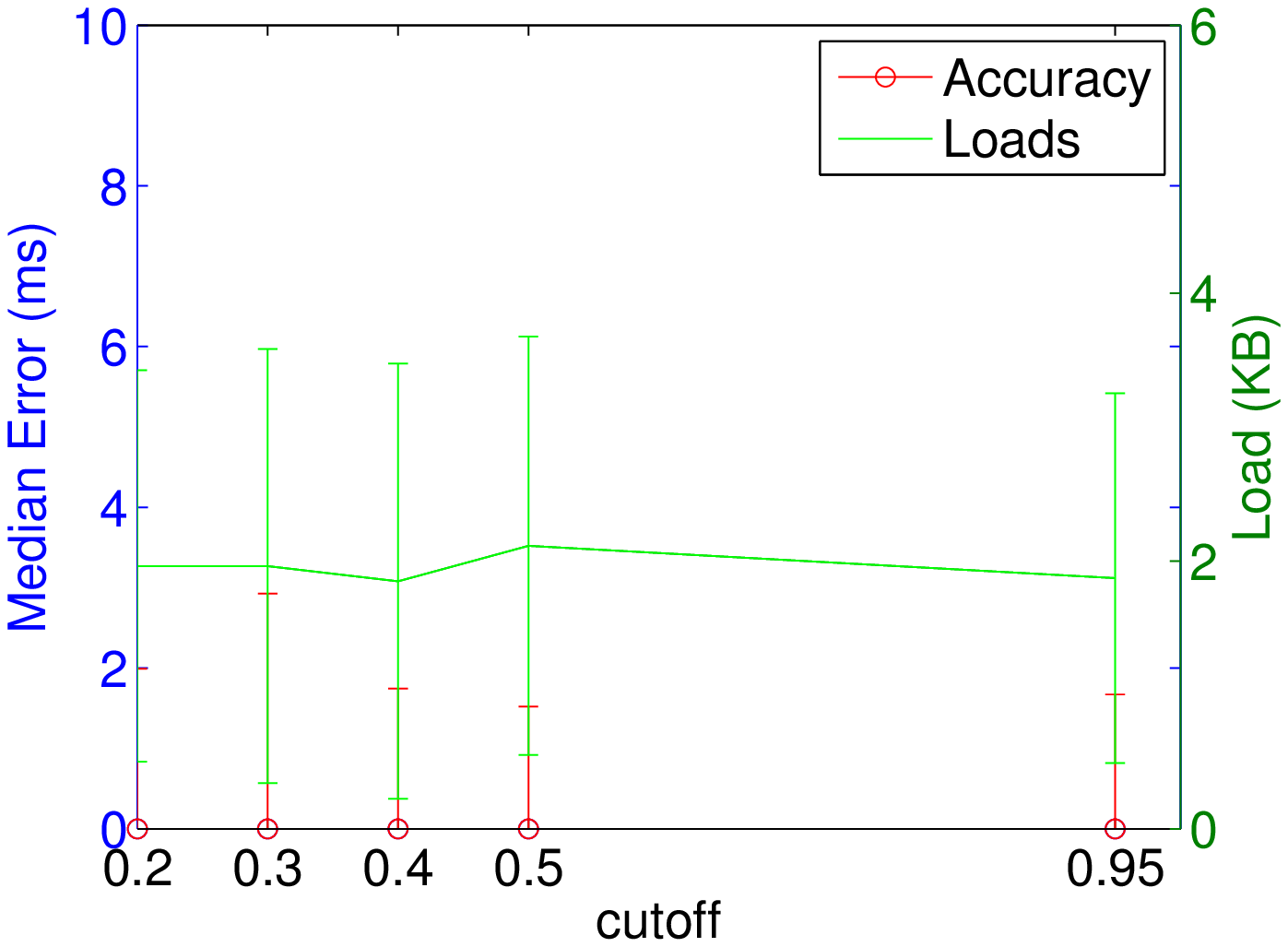}
          \label{fig:B}
         }
               \subfigure[Host479.]
        {
          \setlength{\epsfxsize}{.44\hsize}
          \epsffile{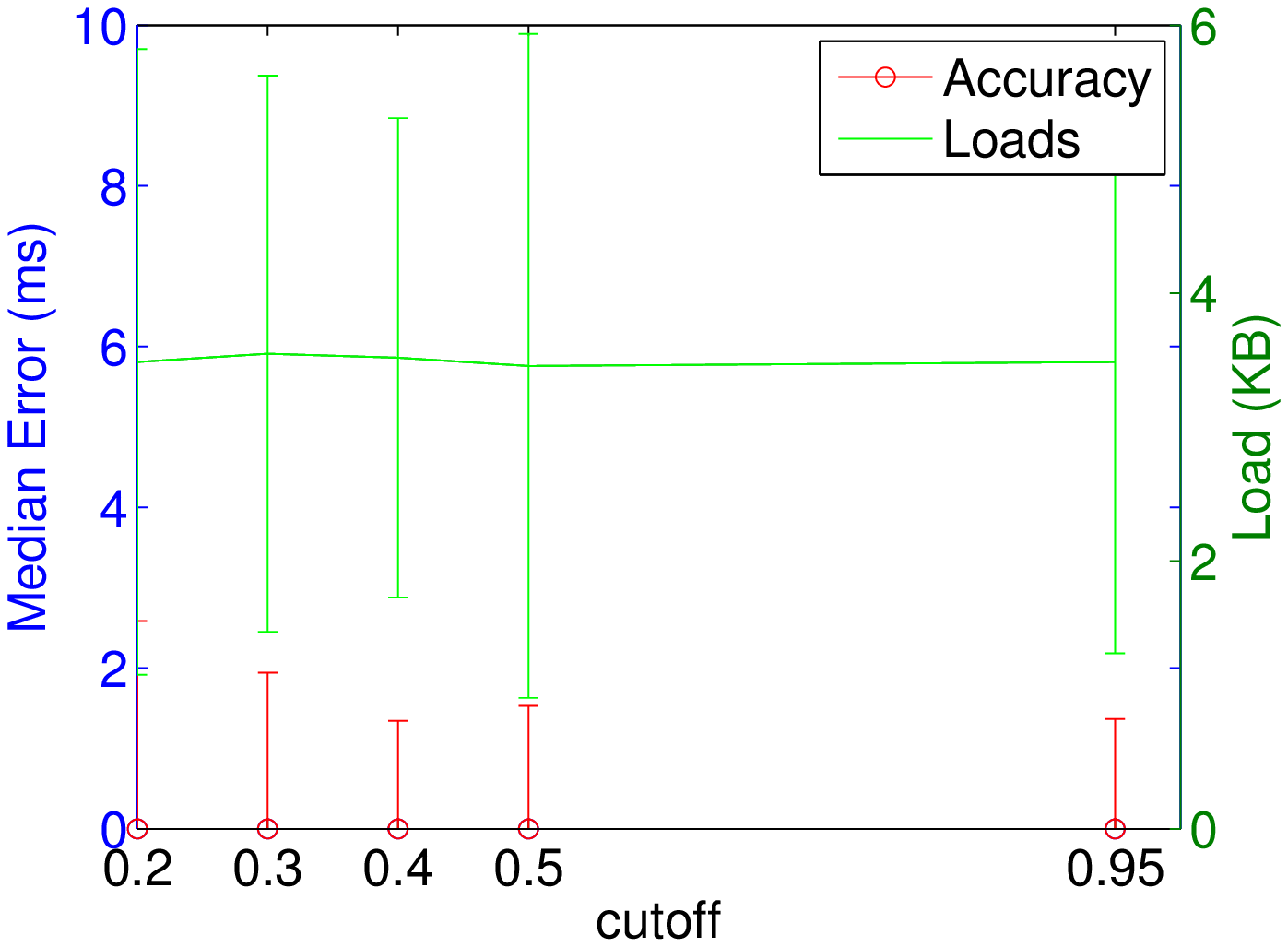}
          \label{fig:B}
         }
     \caption{Search Cutoff.}
     \label{fig:AB}
\end{figure}
}

\subsubsection{Coordinate Dimension $|x|$}

Fig.~\ref{fig:coordDimSen} illustrates the accuracy and loads when
the coordinate dimension changes. HybridNN achieves similar
accuracy and loads as the accuracy of coordinates keeps stably
accurate as the dimension is over 3. Therefore, HybridNN can adapt
to inaccuracy of different dimensions of coordinates without
increasing DNNS query loads efficiently.

\begin{figure}[htbp]
     \centering
           \subfigure[DNS1143.]
        {
          \setlength{\epsfxsize}{.44\hsize}
          \epsffile{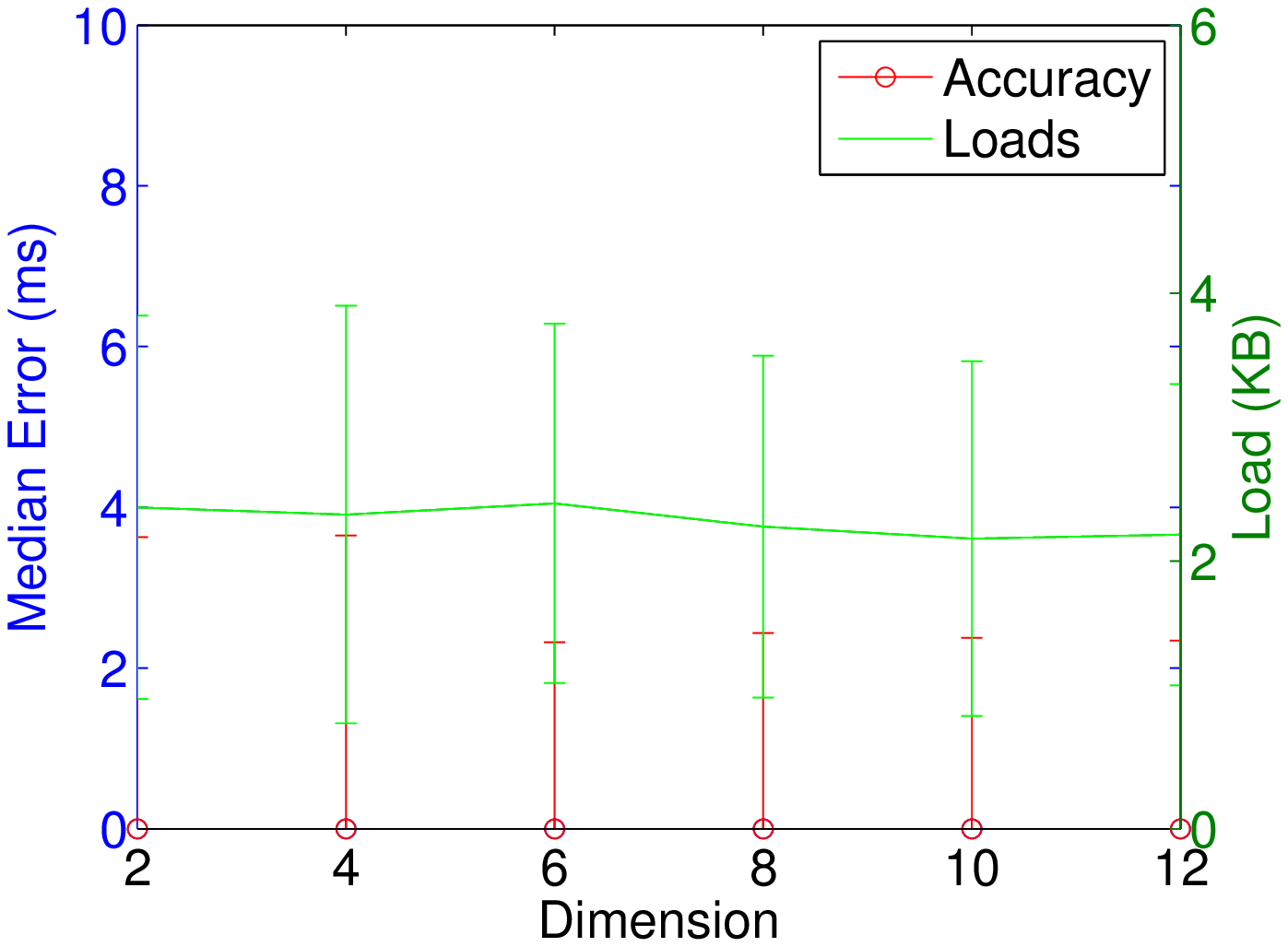}
          \label{fig:B}
         }
               \subfigure[DNS2500.]
        {
          \setlength{\epsfxsize}{.44\hsize}
          \epsffile{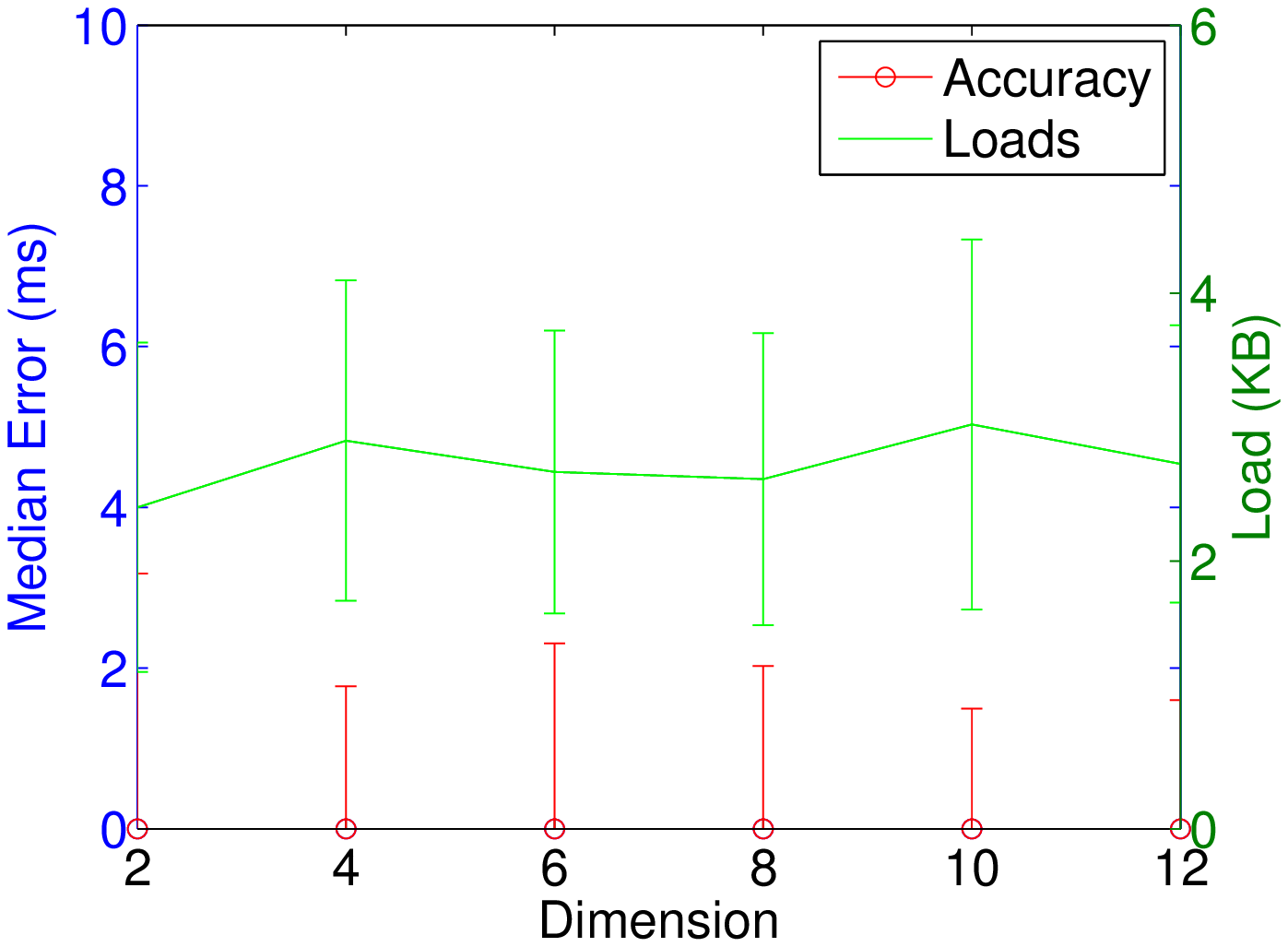}
          \label{fig:B}
         }
               \subfigure[DNS3997.]
        {
          \setlength{\epsfxsize}{.44\hsize}
          \epsffile{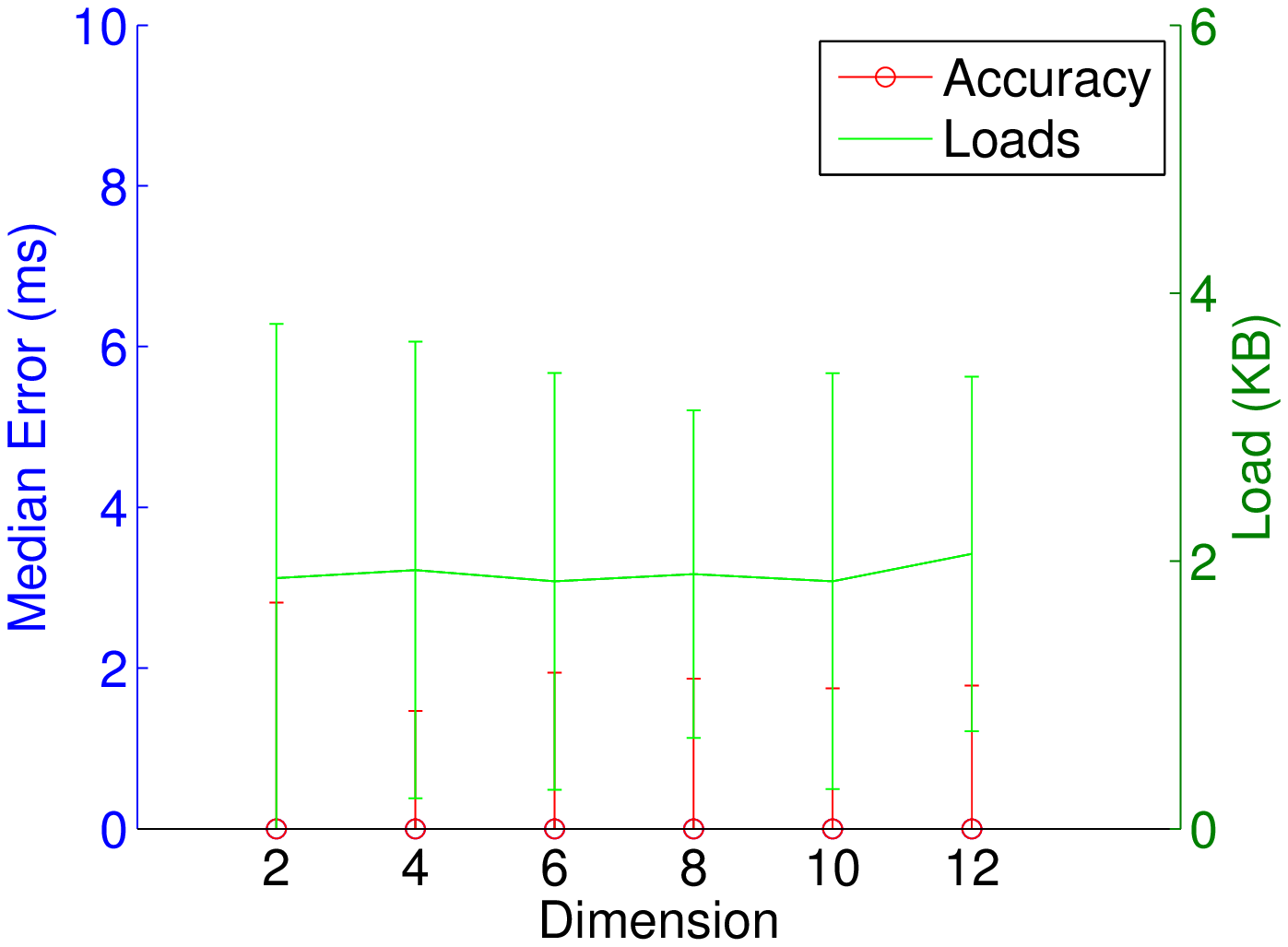}
          \label{fig:B}
         }
               \subfigure[Host479.]
        {
          \setlength{\epsfxsize}{.44\hsize}
          \epsffile{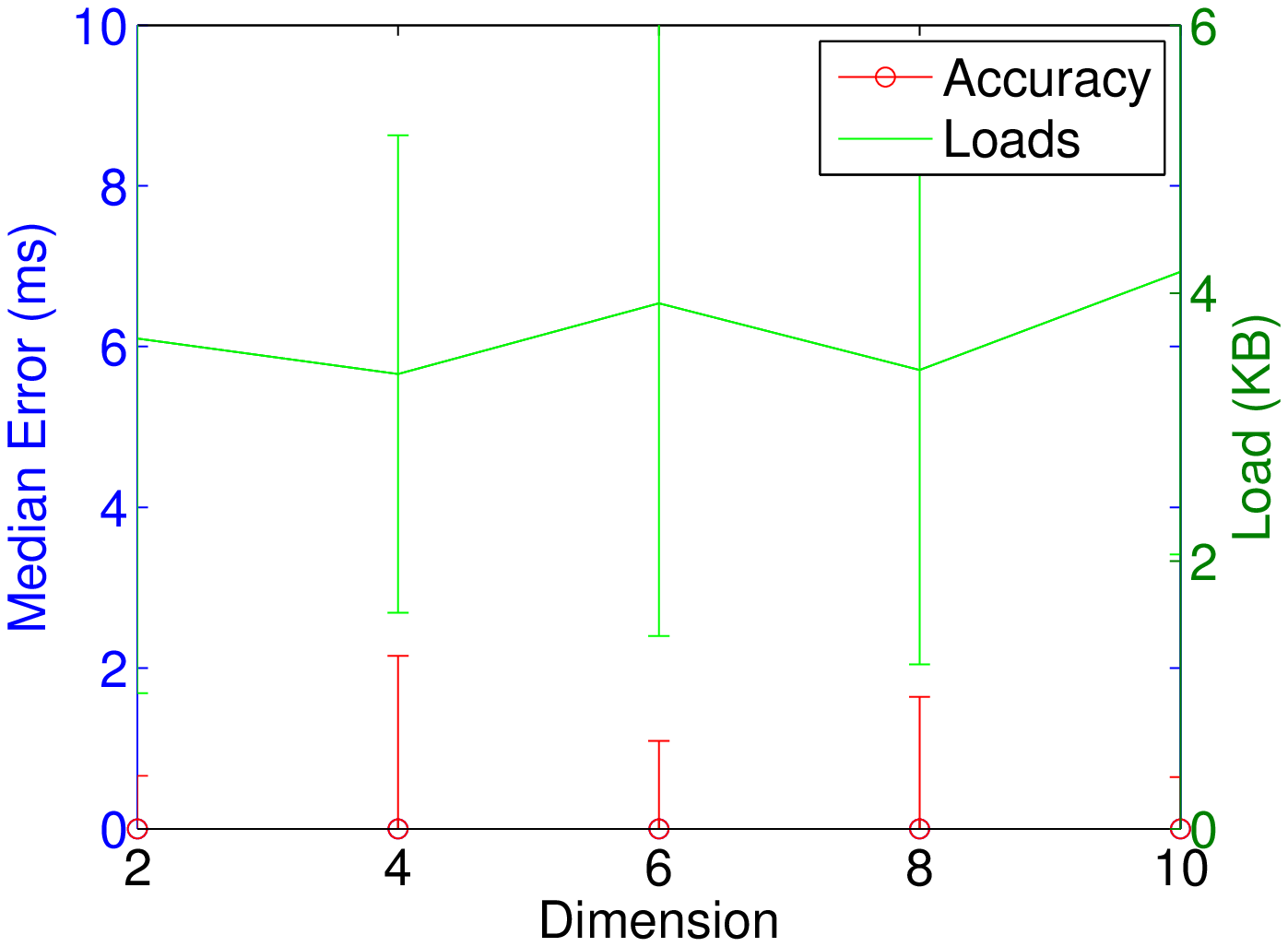}
          \label{fig:B}
         }
     \caption{Coordinate Dimension.}
     \label{fig:coordDimSen}
\end{figure}

\subsubsection{Nodes Per Ring $\Delta$}

Fig.~\ref{fig:NodesPerRing} describes the performance of HybridNN
with increasing upper bounds of nodes per ring. HybridNN achieves
high accuracy event the size of one ring is as small as 5. This is
because HybridNN selects neighbors from broader range $\left[
{0,\rho d} \right]$, where $d$ is the delay from current node to
targets. Besides, the loads of HybridNN grow slowly as the size of
ring increases. As HybridNN utilizes coordinate distances to
select limited number of candidate neighbor.

\begin{figure}[htbp]
     \centering
           \subfigure[DNS1143.]
        {
          \setlength{\epsfxsize}{.44\hsize}
          \epsffile{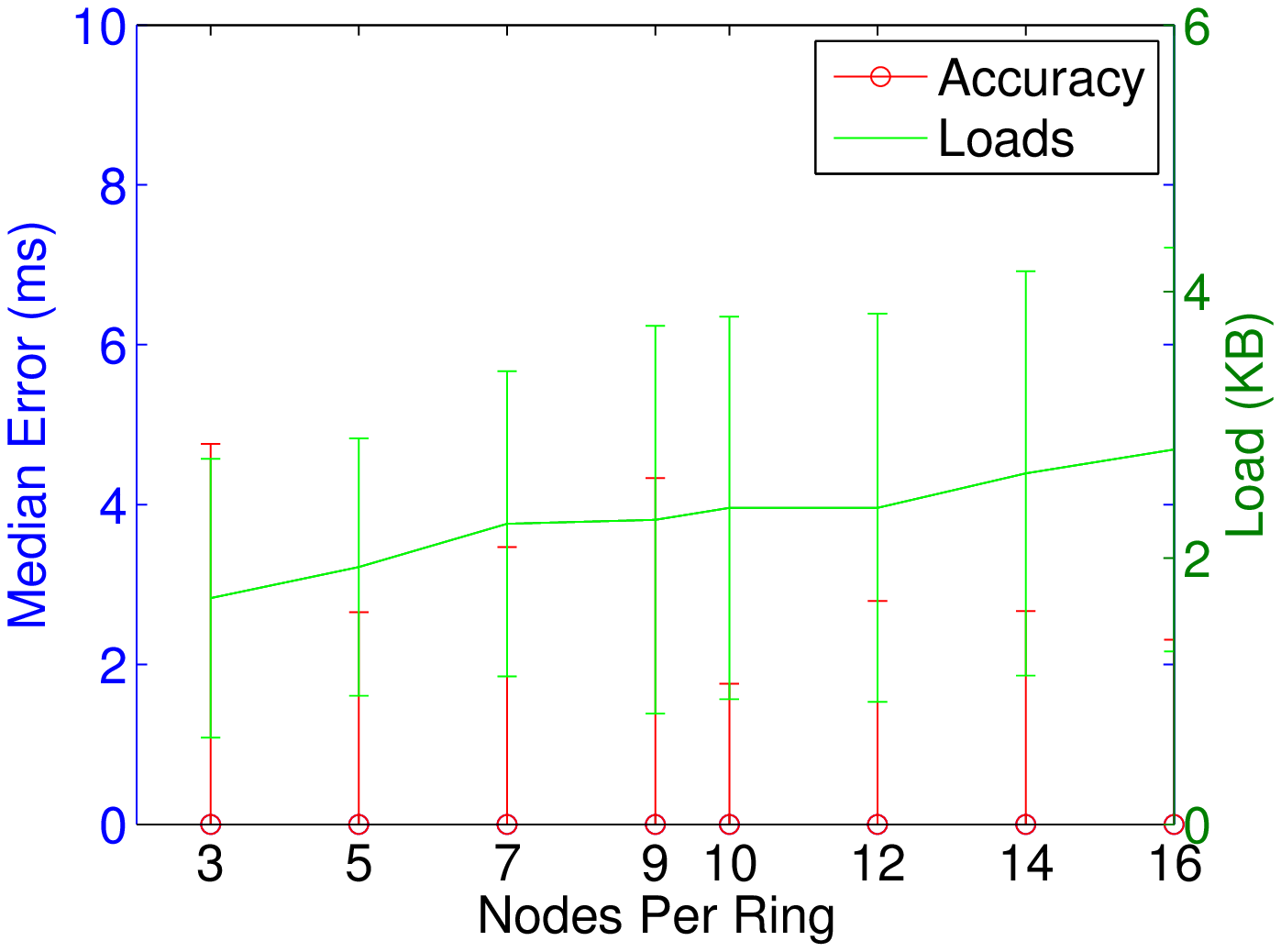}
          \label{fig:B}
         }
               \subfigure[DNS2500.]
        {
          \setlength{\epsfxsize}{.44\hsize}
          \epsffile{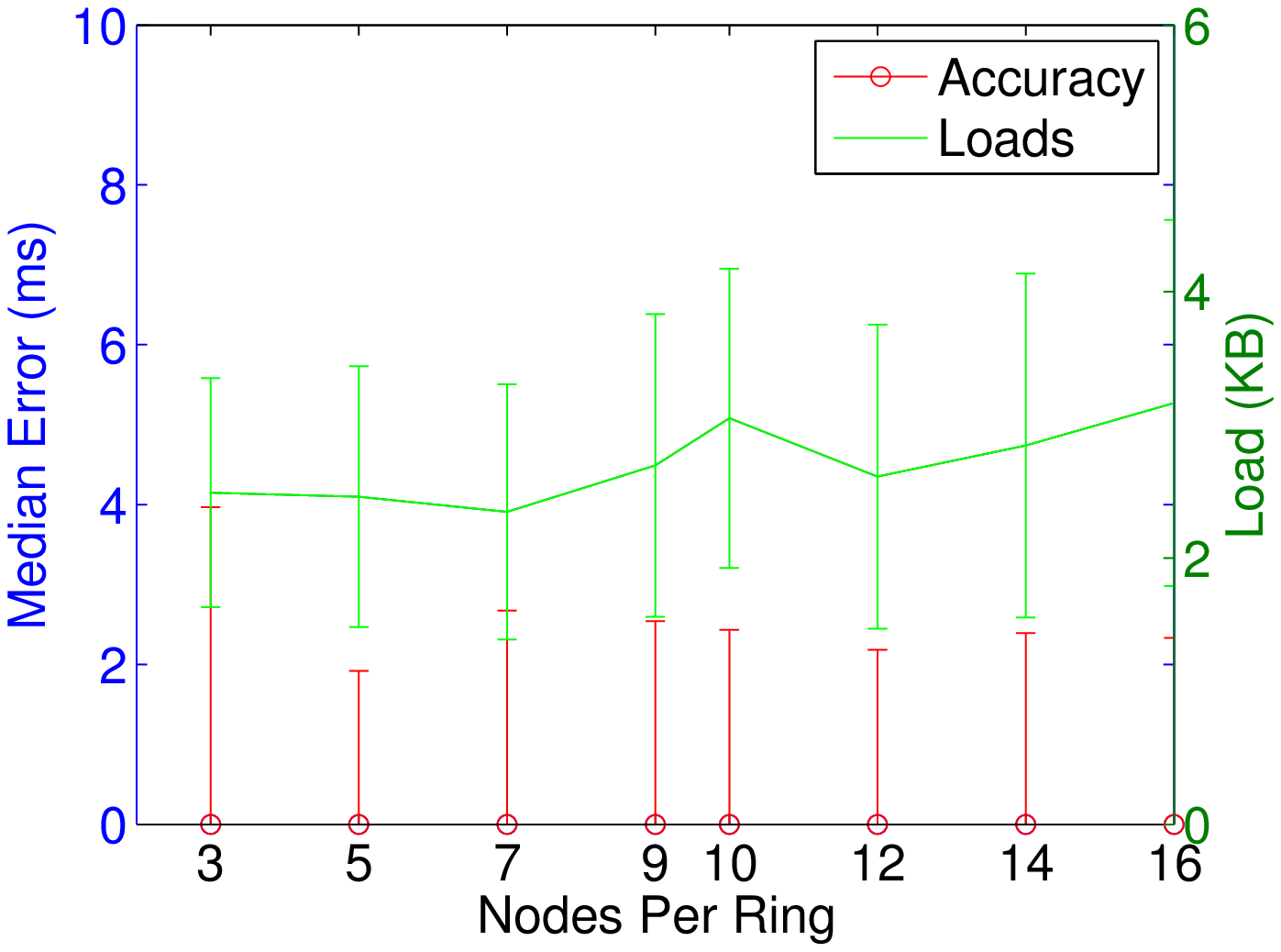}
          \label{fig:B}
         }
               \subfigure[DNS3997.]
        {
          \setlength{\epsfxsize}{.44\hsize}
          \epsffile{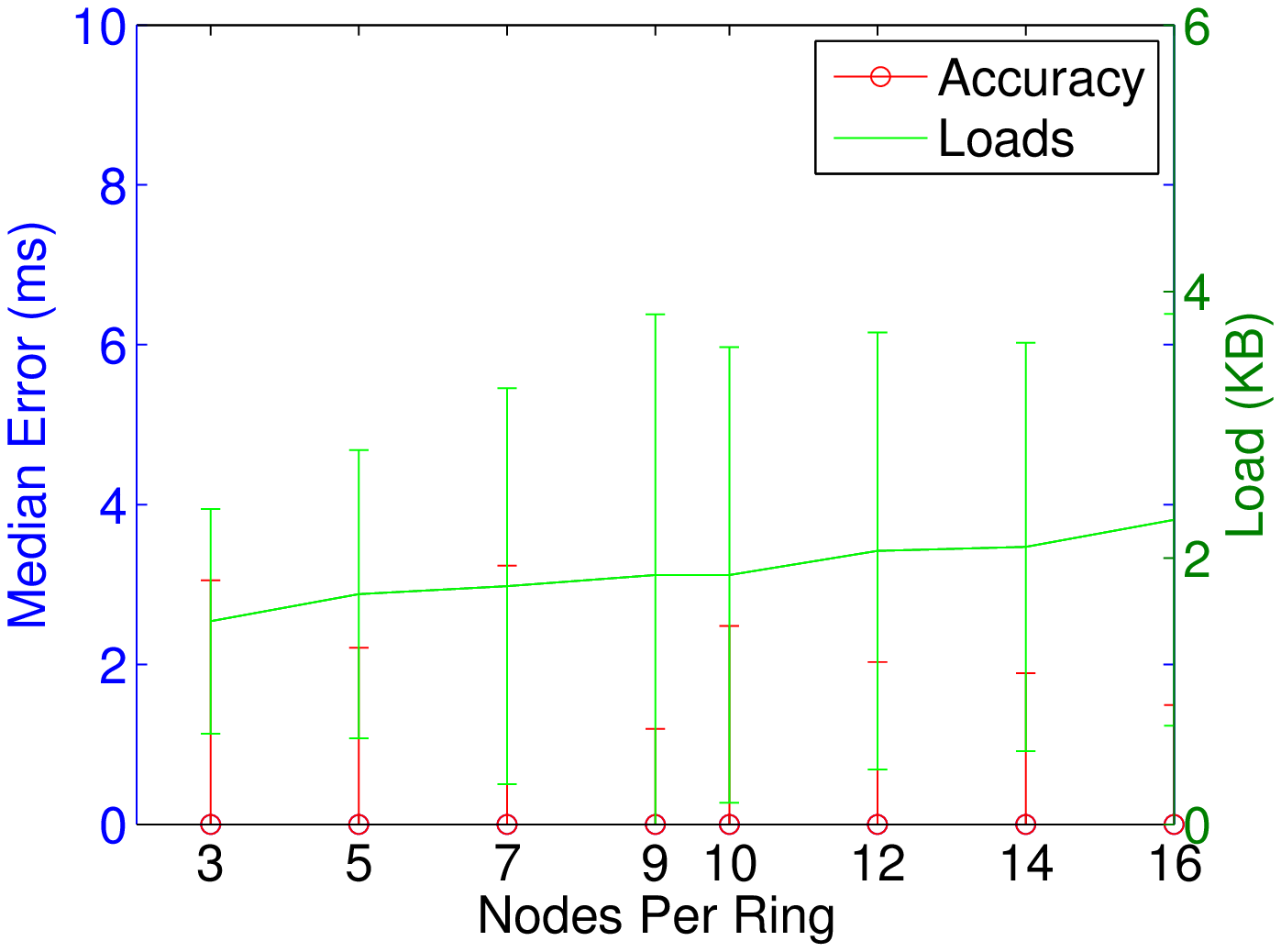}
          \label{fig:B}
         }
               \subfigure[Host479.]
        {
          \setlength{\epsfxsize}{.44\hsize}
          \epsffile{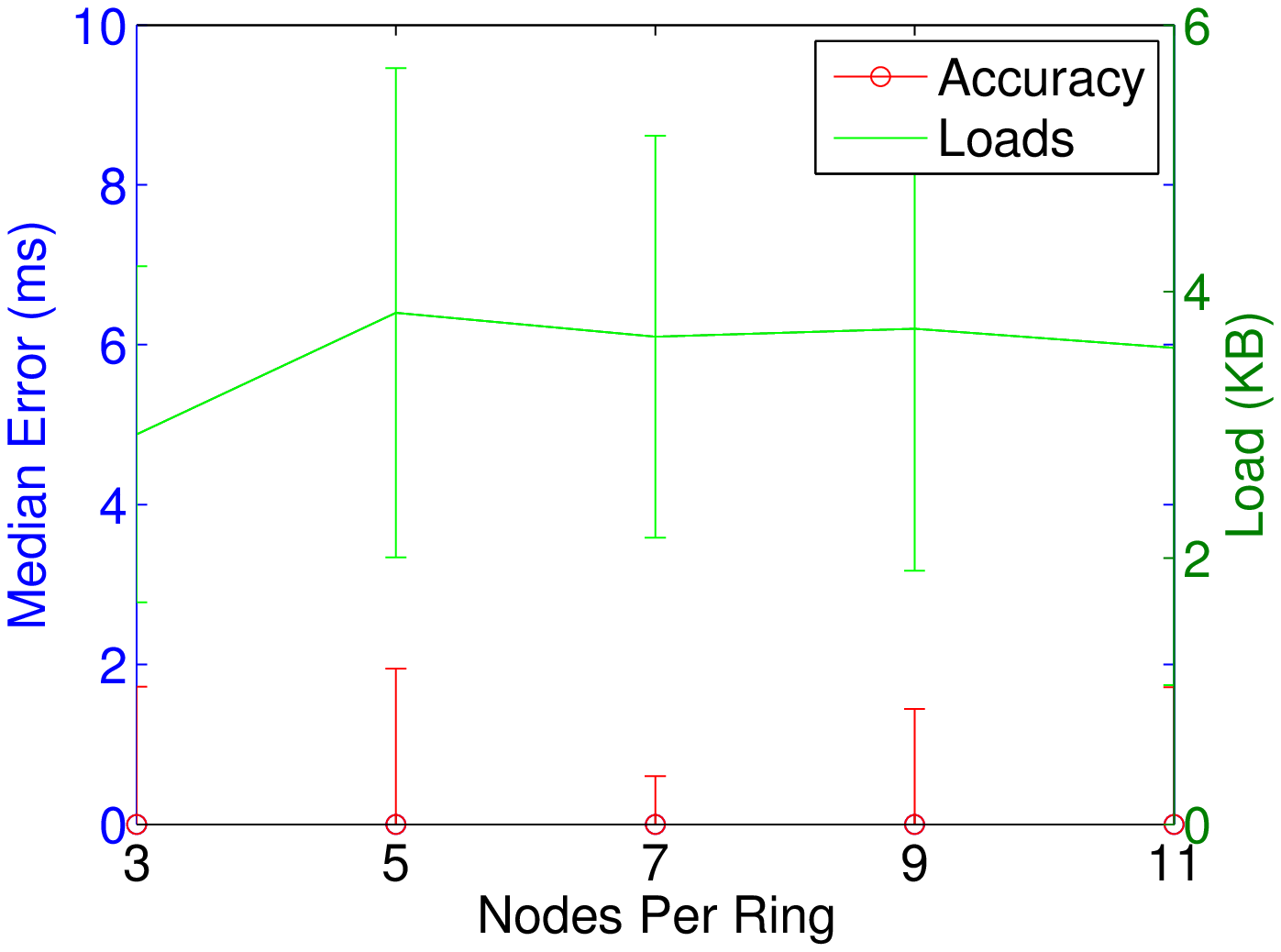}
          \label{fig:B}
         }
     \caption{Nodes Per Ring.}
     \label{fig:NodesPerRing}
\end{figure}

\subsubsection{OverSampled nearest and farthest nodes $K$}

Fig.~\ref{fig:Over-samplednumber} illustrates the performance of
HybridNN as the variation of oversampled number of nearest and
farthest nodes $K$. HybridNN achieves similar accuracy and loads
when the oversampled size $K$ of nearest neighbors and farthest
neighbors. This is because we periodically start the oversampled
process, which can find many nearby or far-away nodes
accumulatively.

\begin{figure}[htbp]
     \centering
         \subfigure[DNS1143.]
        {
          \setlength{\epsfxsize}{.44\hsize}
          \epsffile{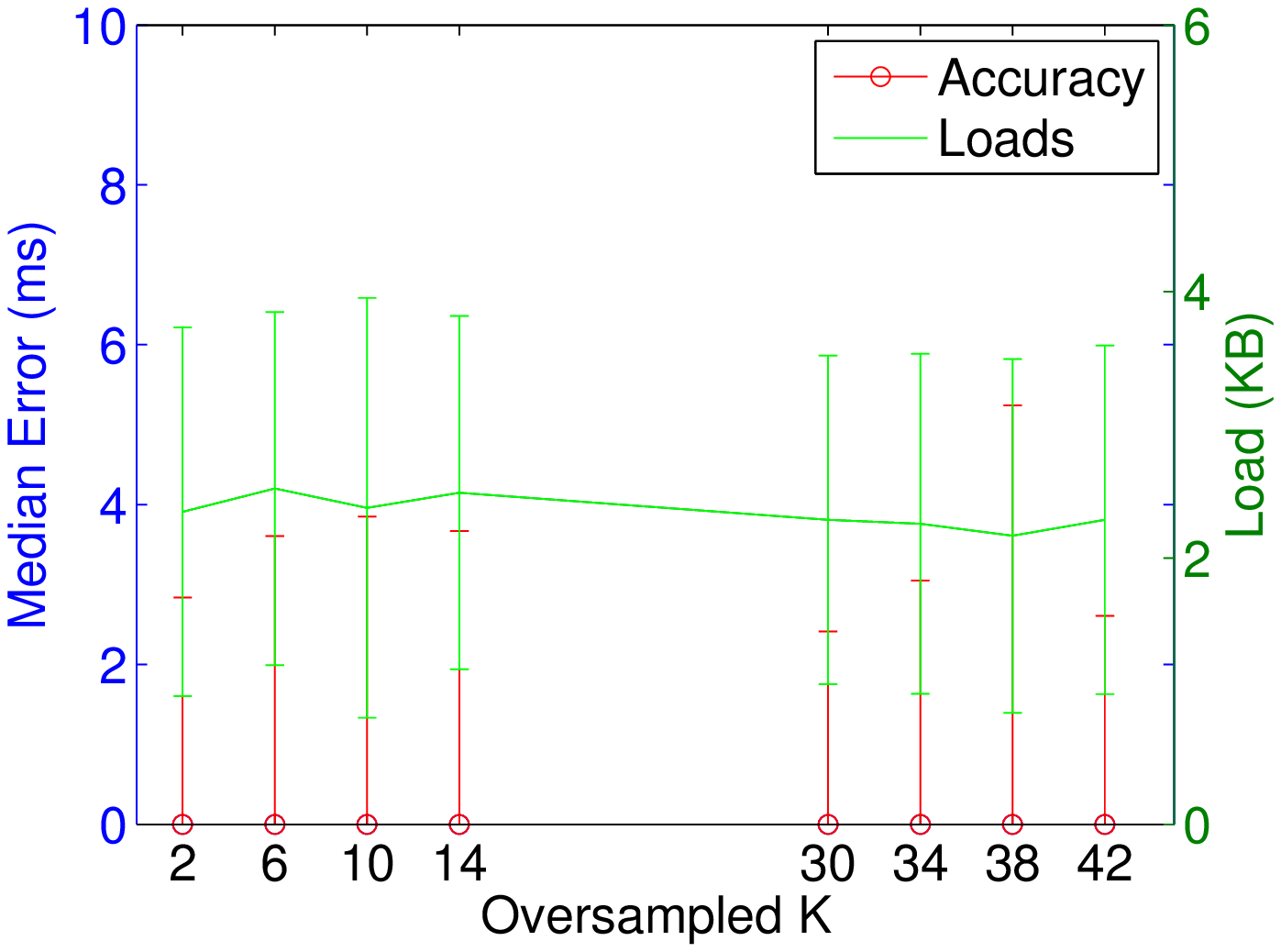}
          \label{fig:B}
         }
               \subfigure[DNS2500.]
        {
          \setlength{\epsfxsize}{.44\hsize}
          \epsffile{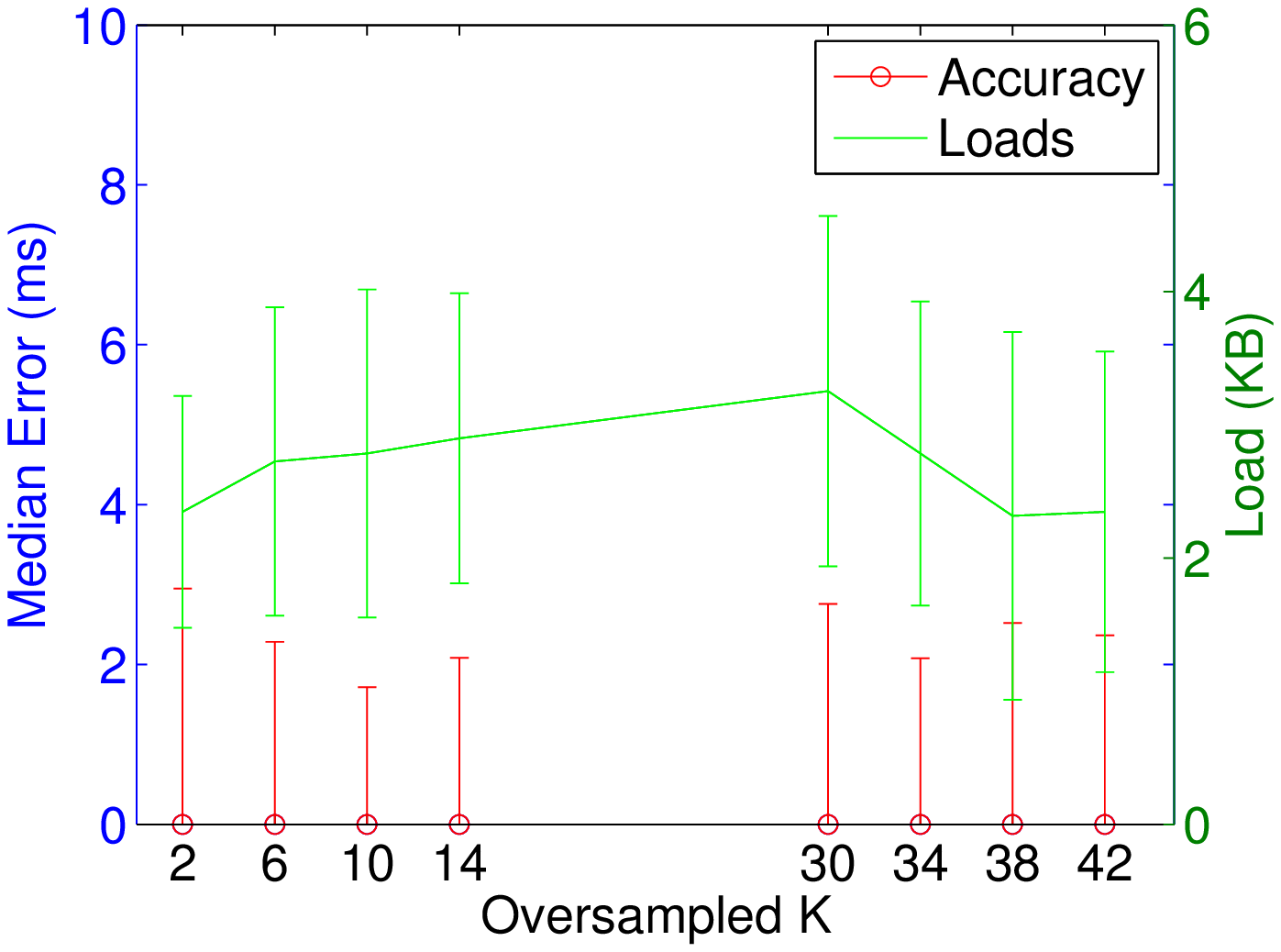}
          \label{fig:B}
         }
               \subfigure[DNS3997.]
        {
          \setlength{\epsfxsize}{.44\hsize}
          \epsffile{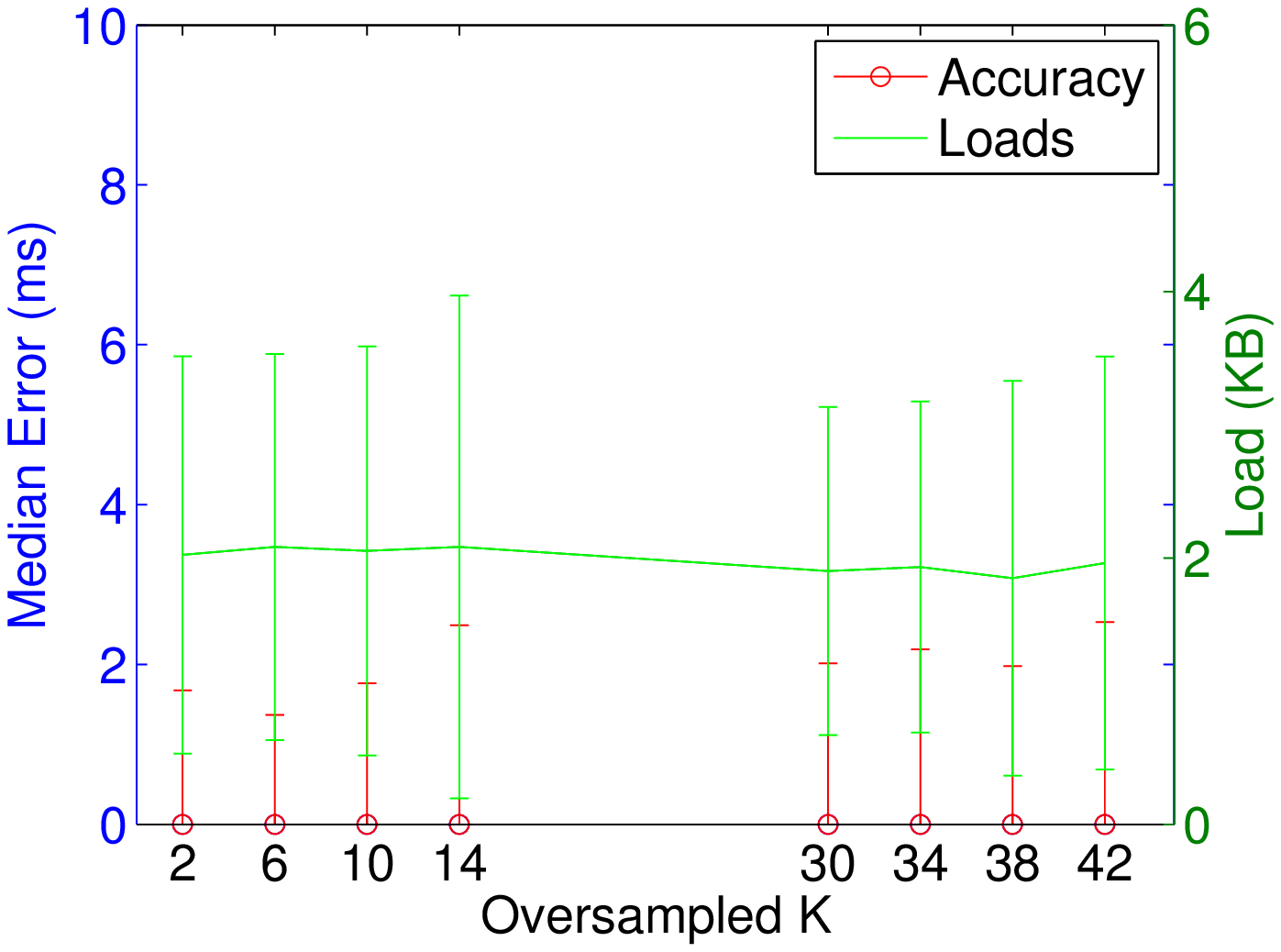}
          \label{fig:B}
         }
               \subfigure[Host479.]
        {
          \setlength{\epsfxsize}{.44\hsize}
          \epsffile{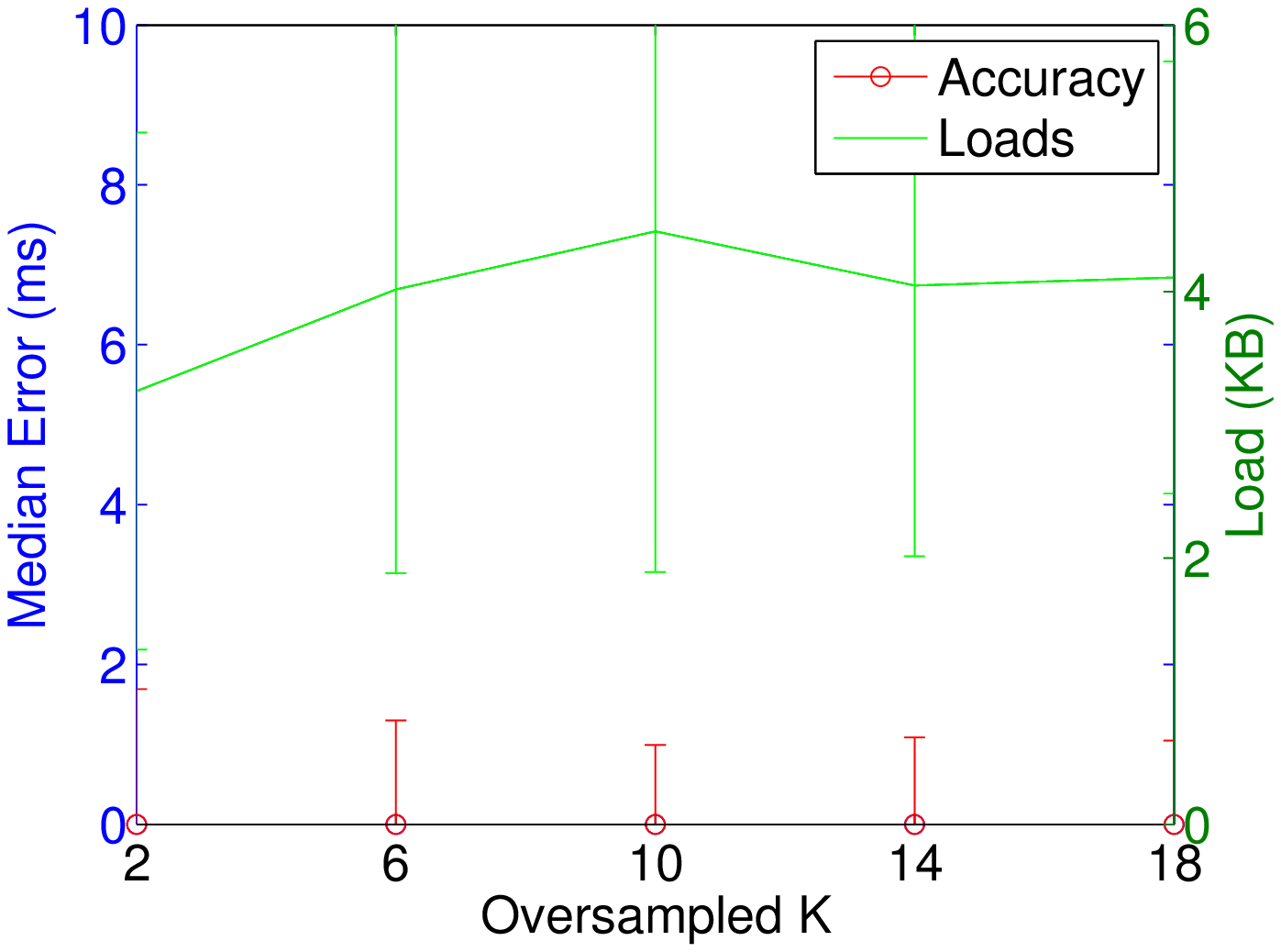}
          \label{fig:B}
         }
     \caption{Over-sampled number of neighbors.}
     \label{fig:Over-samplednumber}
\end{figure}

\subsubsection{Returned Nodes For Next-Hop Probe $m$}

Fig.~\ref{fig:returnedNodes} plots the median errors and loads of
HybridNN with increasing returned nodes for next-hop probes for
HybridNN. For all data sets, HybridNN is accurate when the size of
estimated nearest candidate neighbors for direct probes exceeds 2.
Moreover, the loads of HybridNN increase slowly as the increment
of relaxed probes. This is because we also add neighbors with
higher uncertain coordinates, weakening the increased overhead of
relaxed probes. Besides, the search process typically terminates
at 3 to 5 hops as we found during experiments, therefore the
measurement overhead is  mostly bounded below 3 KB.

\begin{figure}[htbp]
     \centering
     \subfigure[DNS1143.]
        {
          \setlength{\epsfxsize}{.44\hsize}
          \epsffile{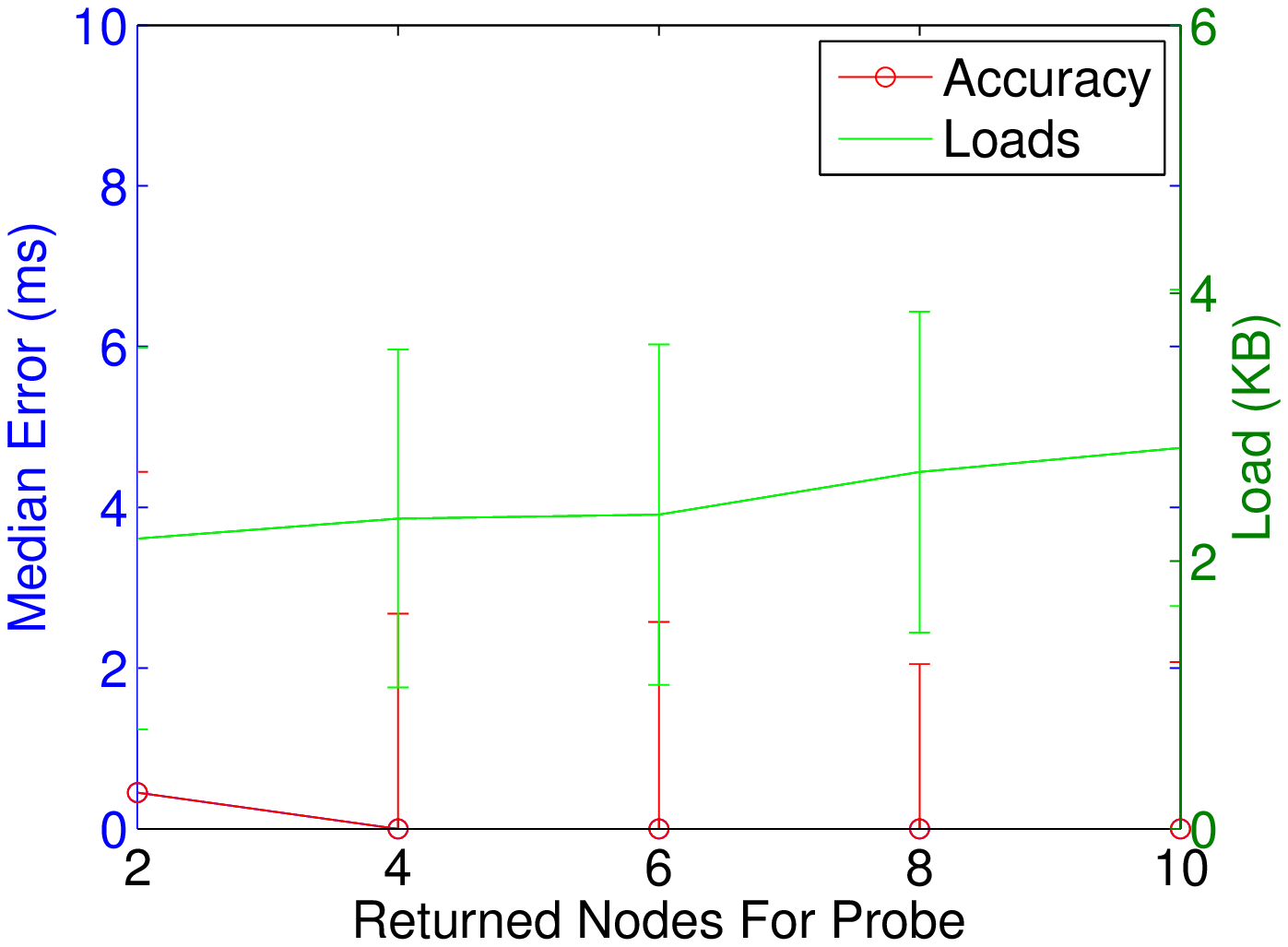}
          \label{fig:B}
         }
               \subfigure[DNS2500.]
        {
          \setlength{\epsfxsize}{.44\hsize}
          \epsffile{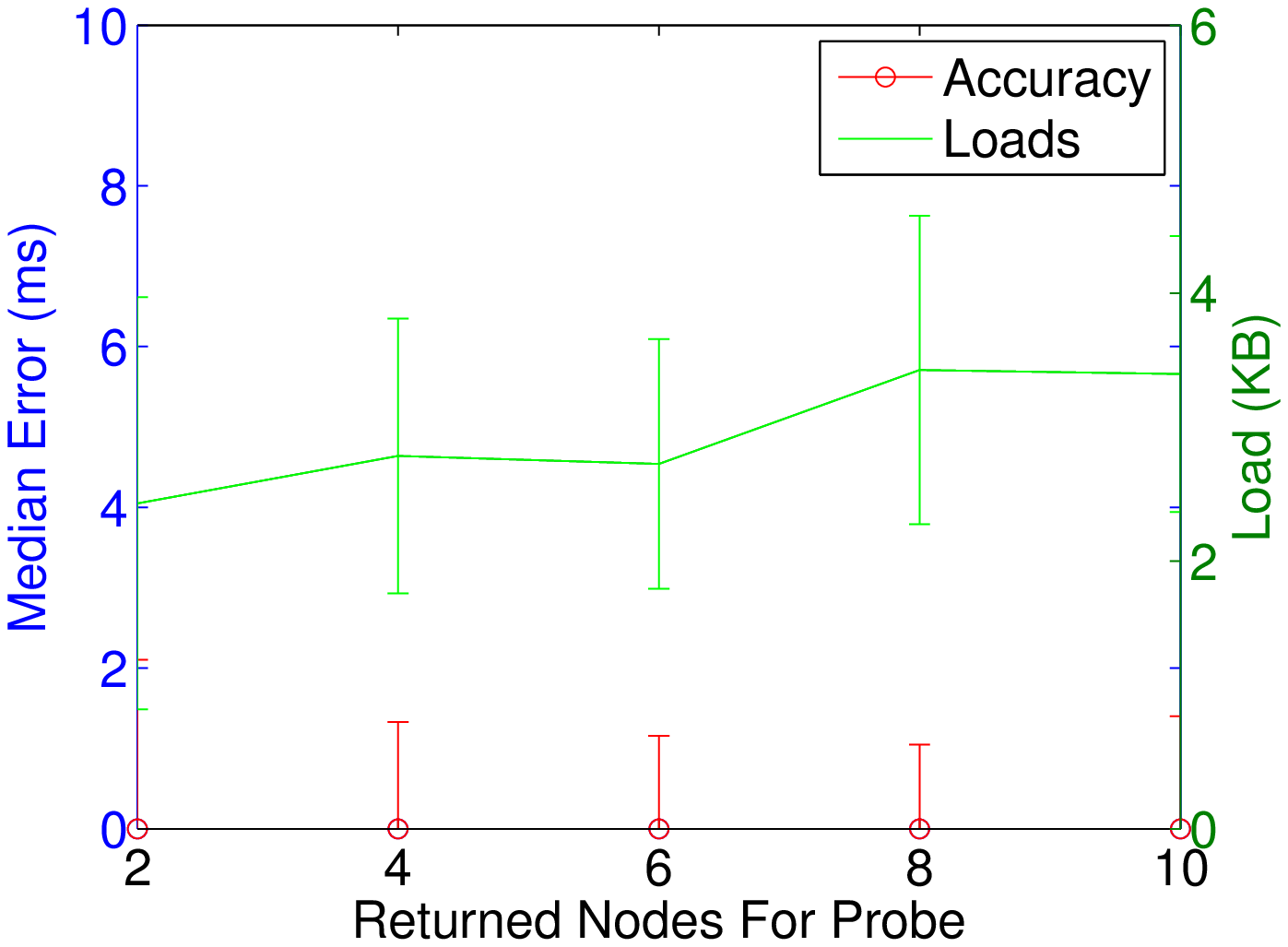}
          \label{fig:B}
         }
               \subfigure[DNS3997.]
        {
          \setlength{\epsfxsize}{.44\hsize}
          \epsffile{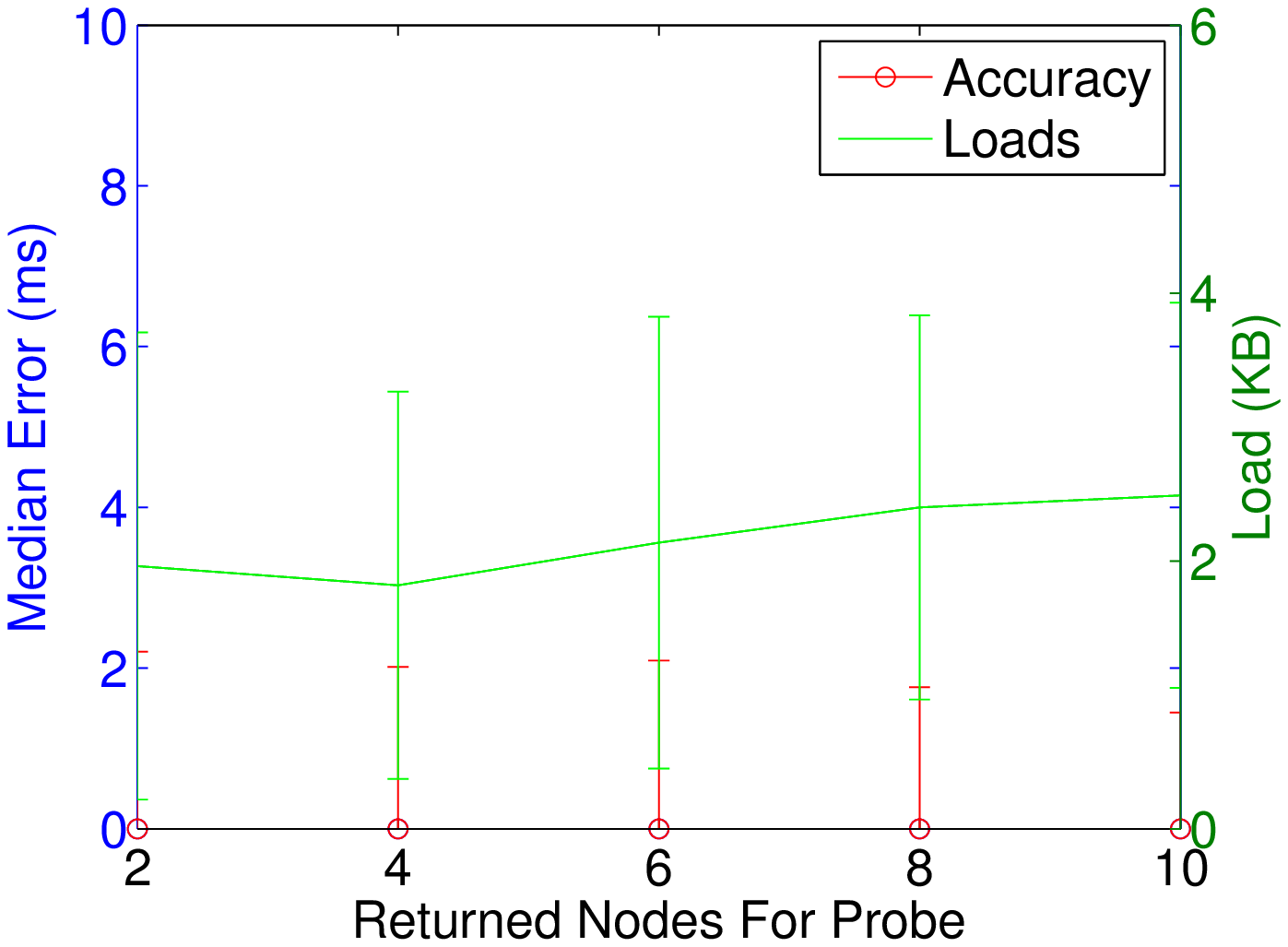}
          \label{fig:B}
         }
               \subfigure[Host479.]
        {
          \setlength{\epsfxsize}{.44\hsize}
          \epsffile{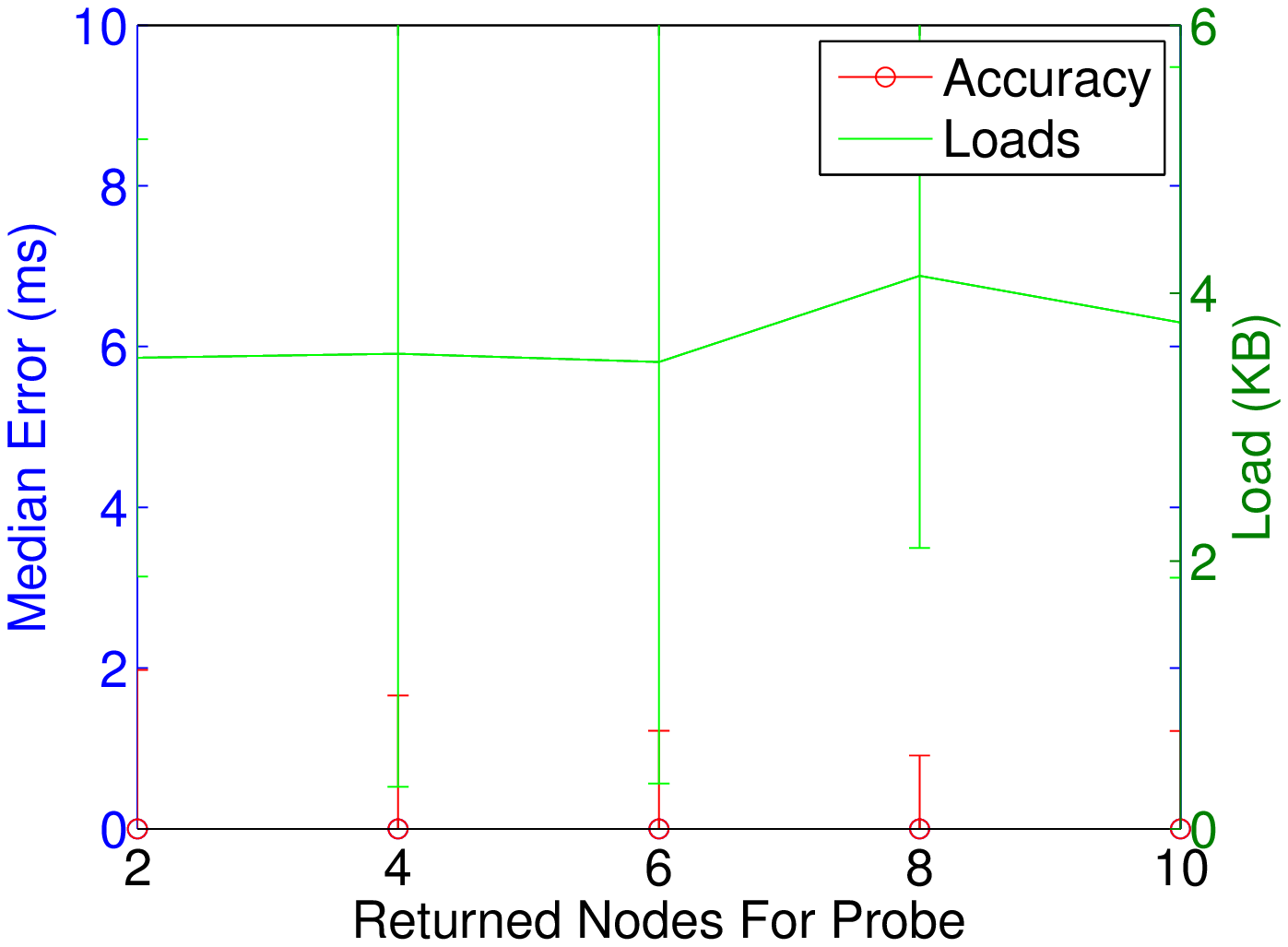}
          \label{fig:B}
         }
     \caption{Returned Nodes For Next-Hop Probes.}
     \label{fig:returnedNodes}
\end{figure}

\section{PlanetLab Experiments}
\label{planetlabExp}

We have implemented a prototype DNNS query system in Java using
the asynchronous communication library. We implemented both
HybridNN and Meridian. The core DNNS logic consists of around
5,000 lines of codes comprising three main modules: (1) prober
module, which uses the kernel-level ping for delay measurements,
to allievate application level perturbations caused by high loads
of PlanetLab nodes; (2) neighborhood management module, which
finds and maintains neighbors on the concentric rings; (3) DNNS
module, which utilizes the HybridNN or Meridian algorithm.

\co{Recall that the completion time for a DNNS query consists of
three components: (i) \textbf{query submission delay}, which is
the period of submitting a DNNS query to a service node by a
client; (ii) \textbf{query handling delay}, which is the period of
locating the nearest service node to the client; (iii)
\textbf{query answering delay}, which is the period of sending the
DNNS result to the client from a service node. Studying the
optimal policy of in-advance DNNS probing is our future work.}

Our objective is to compare the accuracy and efficiency of DNNS
queries with related nearest server location methods using
real-world deployments. To that end, we choose 173 servers
distributed globally on the PlanetLab as the service nodes. Then
we select another 412 servers on the PlanetLab as the target
machines. Our experiments last one week from 05-05-2011 to
12-05-2011.

\co{ , whose locations are depicted in
Fig~\ref{fig:Map533749-CommunityWalk}
  \begin{figure}[tp]
  \leavevmode \centering \setlength{\epsfxsize}{0.8\hsize}
  \epsffile{location.eps}
  \caption{The geographical distributions of selected machines on the PlanetLab.}
  \label{fig:Map533749-CommunityWalk}
\end{figure}}

We compare HybridNN with Meridian and iPlane
\cite{DBLP:conf/nsdi/MadhyasthaKAKV09}. We choose the same
parameter configurations for HybridNN and Meridian as in the
Simulation section (Sec \ref{simSetup}). For iPlane, we query
iPlane to obtain the delays between service nodes and target
machines, then we compute the nearest service node for each target
machine.

Besides, in order to compare the found nearest servers to the
ground-truth nearest servers, we compute the ground-truth nearest
servers using direct probes  (denoted as \emph{Direct}).
Specifically, since pairwise delays between PlanetLab machines
keep varying due to routing dynamics, we first use the median
delay of any node pairs to summarize the long-term delay trend.
Then we select the service node that has the lowest median delay
value to the target.

\subsection{Accuracy}

First we compare the accuracy of different methods with the
absolute error metric and the relative error metric defined in Sec
\ref{simSetup}. The results are shown in
Fig~\ref{PL-deployment}(a) and (b). HybridNN has significantly
lower absolute errors and relative errors than Meridian. iPlane is
similar with HybridNN, but incurs higher errors. The inaccuracy of
iPlane is caused by the mismatch of the estimated routing paths
and the real-world ones. The inaccuracy of Meridian shows that
Meridian is easily trapped at local minimum far away from the
optimal solutions.

On the other hand, HybridNN and iPlane are much accurate, which
implies that hybridNN can avoid bad local minima in most cases.
Nevertheless, HybridNN and iPlane also have around 3\% of DNNS
queries with relative errors above 10. we find that HybridNN
incurs such high errors occur at the early stage, where nodes do
not have enough neighbors in their concentric rings.

\co{ Specifically, over 95\% of cases from HybridNN and iPlane
have relative errors smaller than one. However, HybridNN incurs
smaller errors than iPlane. On the other hand, Meridian only has
around 30\% of DNNS queries with relative errors below one.
Similarly, HybridNN in over 85\% of cases has absolute errors
below 20 ms, but Meridian only has 10\% of DNNS queries satisfying
such absolute errors. Besides, Meridian also has about 10\% of
DNNS queries with relative errors above 10, indicating that

}


\subsection{Completion Time}

Next, we evaluate the completion time of individual DNNS queries
for HybridNN and Meridian. Empirically, we have found that both
HybridNN and Meridian complete DNNS queries within three search
hops, which is consistent with the simulation results in Fig
\ref{fig:searchHopDist}. However, the overall query time for DNNS
searches depends on not only the number of search hops, but also
the completion time of message exchanges and delay probes.

Fig~\ref{PL-deployment}(d) plots the distributions of query time
of HybridNN and Meridian. Around 85\% of the DNNS queries in
HybridNN are similar with those of Meridian. Therefore, query time
for HybridNN and Meridian are similar in most cases. However,
around 20\% of the queries take much large time to answer in
Meridian, and 10\% have query time larger than 15 seconds, while
the hybrid measurement approach of HybridNN can avoid large query
latencies.

\co{

Fig \ref{PL-deployment}(c) shows the distribution of search hops
of Meridian and HybridNN. Around 20\% of Meridian searches
complete in one hop, and the rest of Meridian searches complete in
two hops. On the other hand, around 52\% and 45\% of HybridNN
searches complete in two and three hops, and nearly 3\% searches
require one hop. Therefore, HybridNN and Meridian complete in very
few search hops, consistent with the simulation results in Fig
\ref{fig:searchHopDist}.

 Nevertheless, the overall query time for DNNS searches
depends on not only the number of search hops, but also the
completion time of message exchanges and delay probes.
Fig~\ref{PL-deployment}(d) plots the distributions of query time
of HybridNN and Meridian. Around 85\% of the DNNS queries in
HybridNN are similar with those of Meridian. Therefore, query time
for HybridNN and Meridian are similar in most cases. However,
 around 20\% of the queries take much large time to answer in Meridian, and
 10\% have query time larger than 15 seconds, while the hybrid
 measurement approach of HybridNN can avoid large query
 latencies.}

\subsection{Query Overhead}
Next, to quantify the bandwidth overhead of the DNNS queries of
HybridNN and Meridian, we define the load of a DNNS query as the
total size of the transmitted packets during the DNNS process. We
plot the CDFs of the loads for HybridNN and Meridian in
Fig~\ref{PL-deployment}(d). The load of HybridNN is significantly
lower than that of Meridian. In more than 95\%  of the cases the
load of HybridNN is less than 2KBytes, while in more than 50\% of
the cases the load of Meridian is more than 10 KBytes, which is
due to the large size of the candidate neighbor set for DNNS
queries. Therefore, the delay estimation of HybridNN substantially
reduces the measurement overhead.

\subsection{Control Overhead}

To measure the efficiency of HybridNN and Meridian. We collected
the bandwidth overhead of the neighborhood management in HybridNN
and Meridian for each service node every two minutes, as shown in
Fig~\ref{PL-deployment}(e). The maintenance overhead of Meridian
includes both the gossip process and the ring maintenance costs,
while the maintenance of HybridNN includes the gossip messages,
$K$ nearest neighbor search messages and the $K$ farthest neighbor
search messages. The average maintenance overhead of HybridNN is 2
KBytes per minute, and for Meridian is over 20 KBytes  per minute.
Since the time interval of ring maintenance for both HybridNN and
Meridian is identical, the all-pair probes between nodes in the
same ring is the main cause of the control overhead in Meridian.
On the other hand, as HybridNN uses the coordinate distances to
update the rings, it does not need to do all-pair probes between
nodes in a ring.

\begin{figure*}[tp]
     \centering
               \subfigure[Absolute Error.]
        {
          \setlength{\epsfxsize}{.17\hsize}
          \epsffile{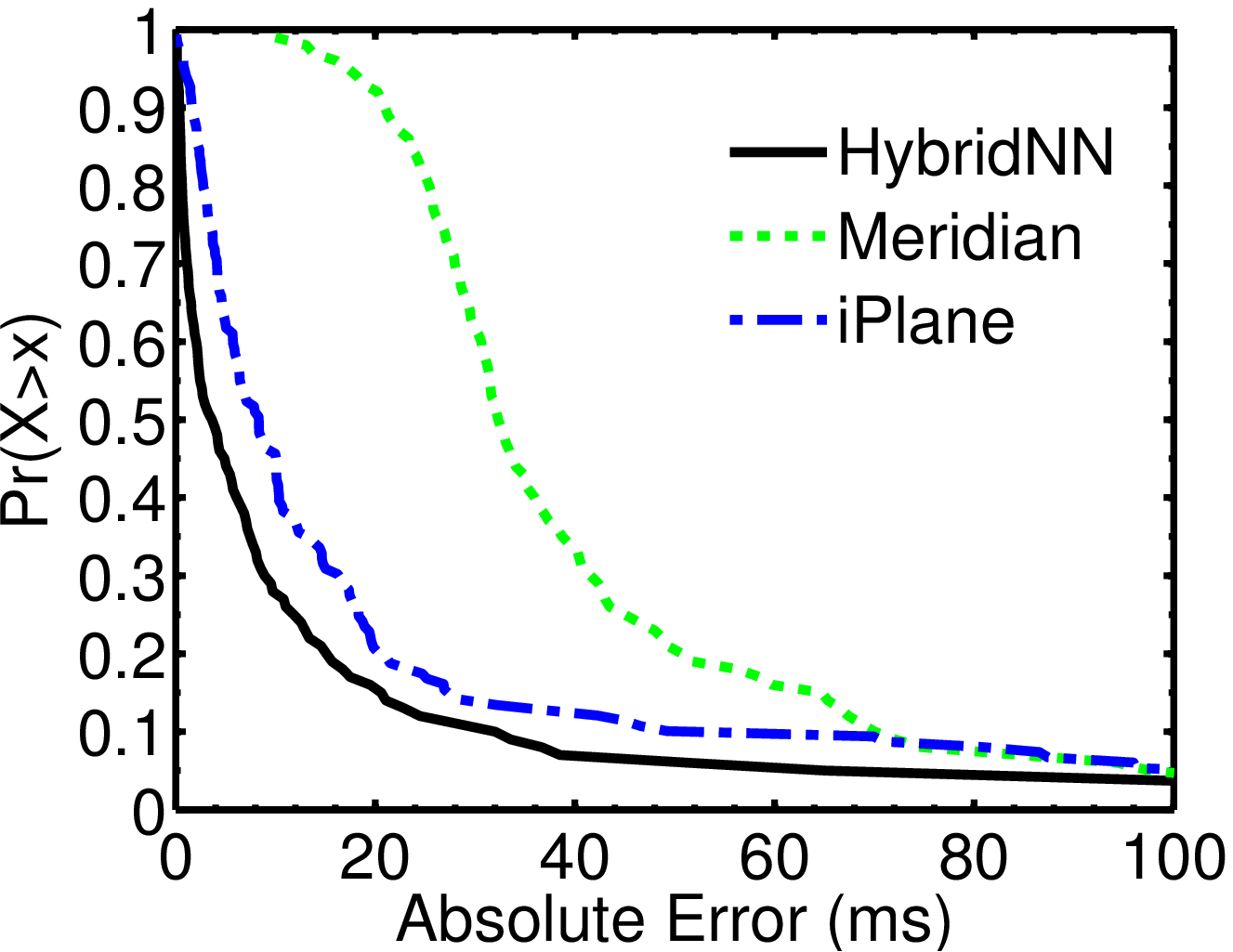}
         }
         \label{fig:1-new-tmpclosest_PL}%
         \subfigure[Relative Error.]
        {
          \setlength{\epsfxsize}{.17\hsize}
          \epsffile{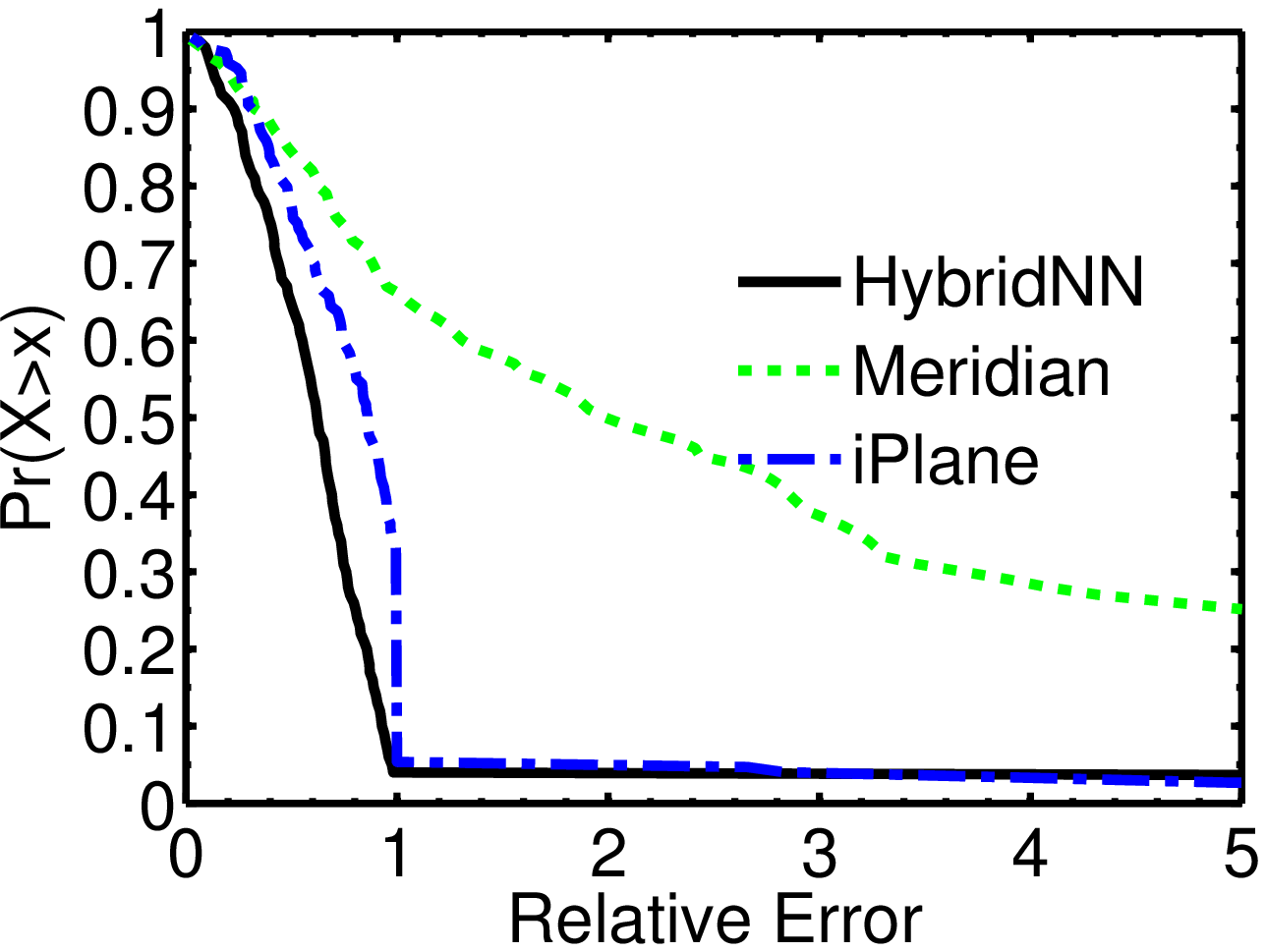}
         }
          \subfigure[Query time.]
        {
          \setlength{\epsfxsize}{.17\hsize}
          \epsffile{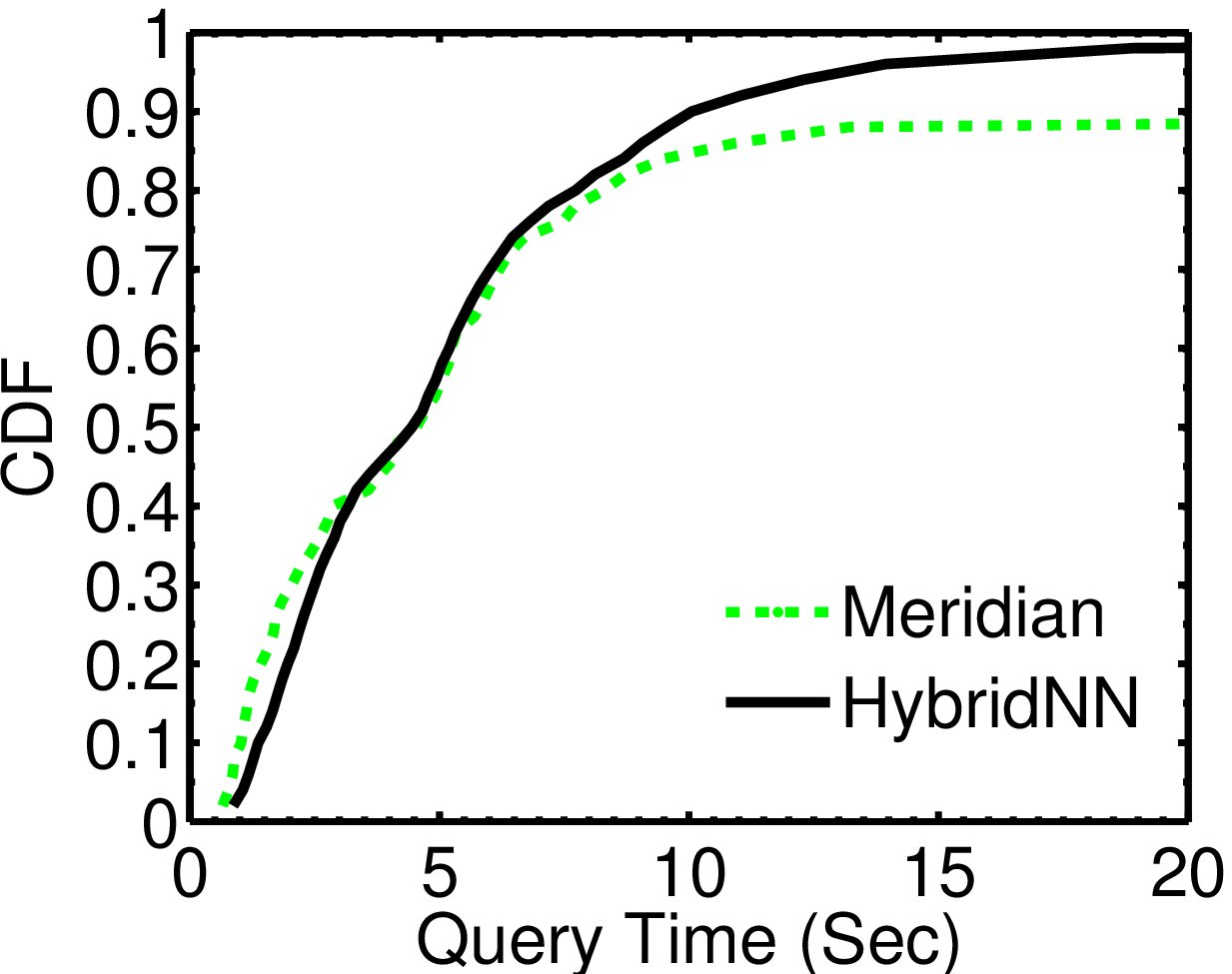}
         }
\label{fig:1_queryTime_PL}%
               \subfigure[Query load.]
        {
          \setlength{\epsfxsize}{.17\hsize}
          \epsffile{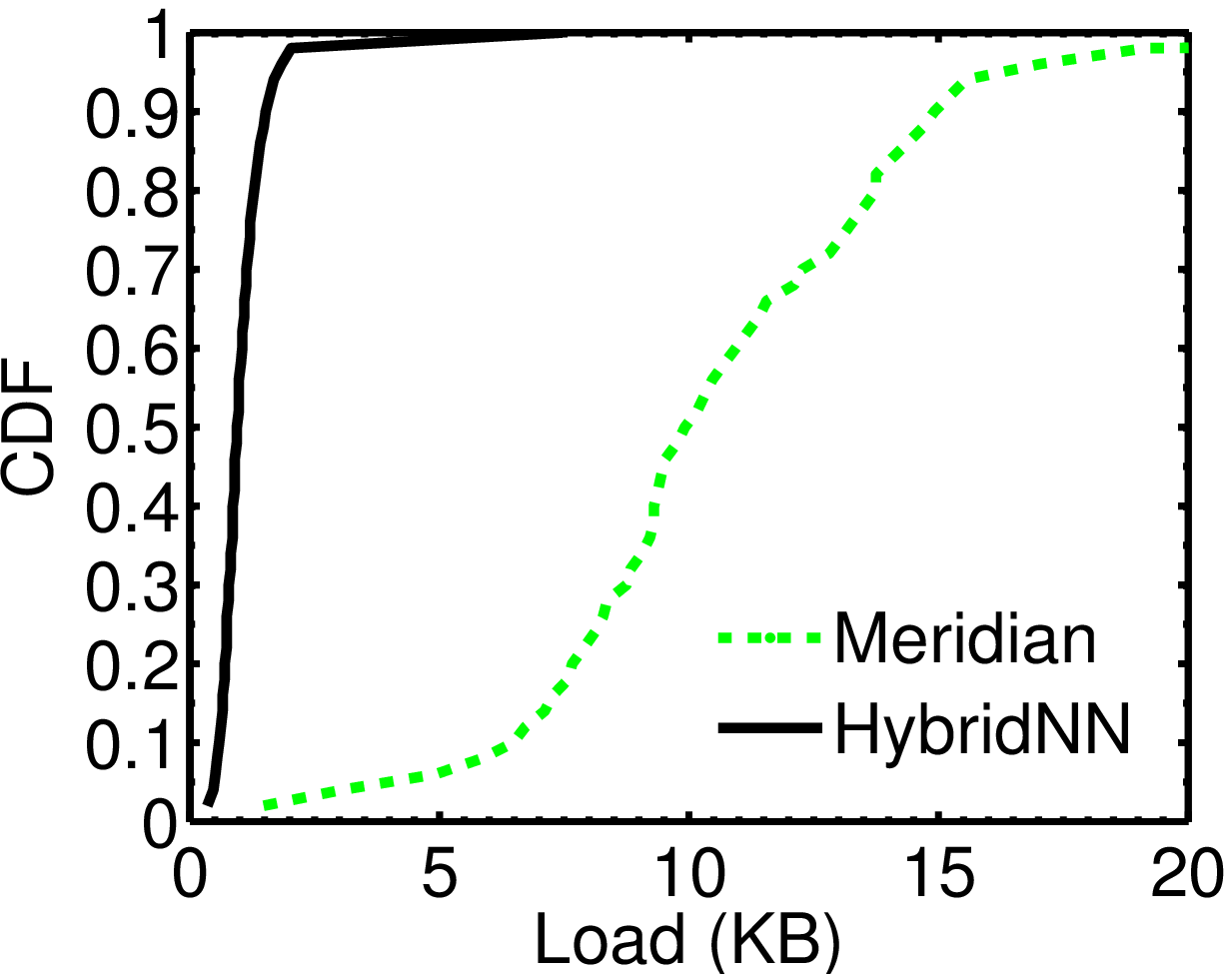}
         }
\label{fig:1_loads_PL}%
         \subfigure[Control overhead.]
        {
          \setlength{\epsfxsize}{.17\hsize}
          \epsffile{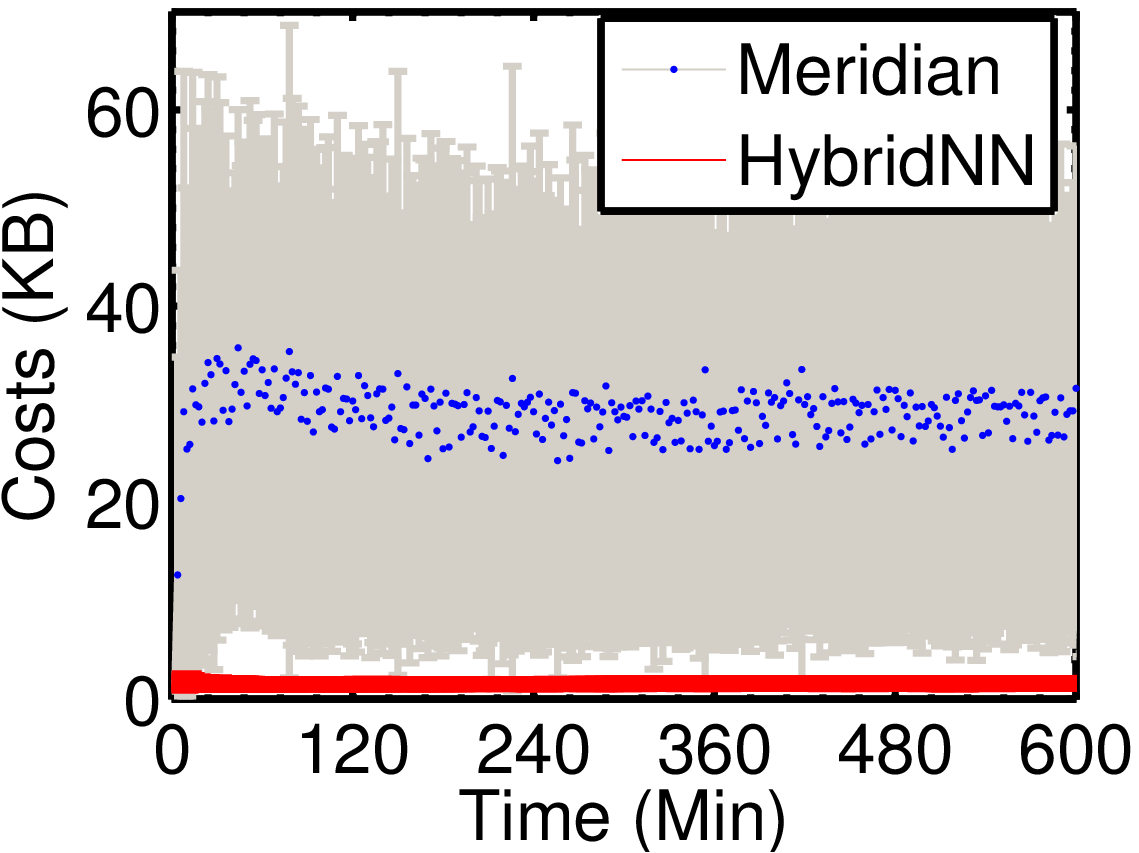}
         }
\label{fig:averagedCost_PL}%
     \caption{Performance comparison on the PlanetLab.}
 \label{PL-deployment}
\end{figure*}

\co{

    \subfigure[Search Hops.]
        {
          \setlength{\epsfxsize}{.17\hsize}
          \epsffile{figures/PLNew/SearchHopsoutput.eps}
         }
}



\section{Conclusion and Future Work}

We have addressed the problem of designing an accurate and
efficient DNNS algorithm in a comprehensive way. We first
formulate the DNNS problem to account for both symmetric and
asymmetric delay metrics for latency optimizations. Given the
generalized delay metrics, we proposed to use the relaxed
inframetric for modelling the delay space as a foundation for
designing new DNNS algorithms with strong theoretical guarantees
concerning search overhead and accuracy of the search results.

Next we apply all the insights gained to design a new DNNS
algorithm called HybrirdNN. HybridNN locates nearest neighbors for
any target using low bandwidth costs. For locating closer server
to any target, HybridNN maximizes the diversity in the neighbor
set, by discovering neighbors within each delay range through a
light-weight neighbor sampling process. Next, in order to reduce
the measurement costs of locating closer servers, HybridNN
combines network coordinate based delay estimation and direct
probes for fast and efficient nearest neighbor determination.
Although the symmetric coordinate distances may deviate from the
asymmetric delays, HybridNN is able to locate the nearest neighbor
to the target at each search step, since we use direct probes to
replace erroneous delay estimations. Finally, HybridNN terminates
the search process conservatively in order to obtain better
approximations of nearest neighbors. We confirmed the efficiency
and effectiveness of HybridNN with extensive simulation and a
prototype deployment on the PlanetLab. HybridNN can locate
approximately closest neighbors quickly with low measurement
costs.

As future work, we plan to continue two lines of research. First,
currently we use the revised Vivaldi to estimate delays, which
mismatches the asymmetric delay metric due to the symmetry of the
coordinate distances. We plan to extend Vivaldi to asymmetric
delay metrics. Second, we plan to study in-advance DNNS probing in
order to hide the waiting time of on-demand DNNS queries for more
practical latency-optimizations.

\bibliographystyle{IEEEtran}
\bibliography{IEEEabrv,ProblemStateRevisedFull}  


\appendix

\section{Proof of Lemma \ref{grid_growth}}

\textit{Lemma \ref{grid_growth}: Given a $\rho$-inframetric with
growth ${\gamma _{g}}\geq1$, for any $x\geq\rho$, $r>0$ and any
node $P$, the volume of a ball $B_P(r)$ is at most $x^{\alpha}$
smaller than that of the ball $B_P(xr)$, where ${\log _\rho
}{\gamma _{g}}  \le \alpha  \le 2{\log _\rho }{\gamma _{g}}$.}

\begin{proof}
First, according to the definition of the growth, it follows:
\[\left| {{B_P}\left( {xr} \right)} \right| \le {\gamma _{g}} \left|
{{B_P}\left( {\frac{x}{\rho }r} \right)} \right|\] Then, by
recursively calling $\left\lceil {{{\log }_\rho }x} \right\rceil $
times the growth definition, until $\frac{x}{{{\rho ^{\left\lceil
{{{\log }_\rho }x} \right\rceil }}}} < 1$, then
\[\begin{array}{l}
 \left| {{B_P}\left( {xr} \right)} \right| \le {{\gamma _{g}} ^{\left\lceil {{{\log }_\rho }x} \right\rceil }}\left| {{B_P}\left( r \right)} \right| = {x^{{{\log }_x}{{\gamma _{g}} ^{\left\lceil {{{\log }_\rho }x} \right\rceil }}}}\left| {{B_P}\left( r \right)} \right| \\
  = {x^\alpha }\left| {{B_P}\left( r \right)} \right|,\alpha  = {\log _x}{\gamma _{g}}  \times \left\lceil {{{\log }_\rho }x} \right\rceil  \\
 \end{array}\]
Therefore, by the definition of the ceiling function, we can
calculate the lower bound of $\alpha$ as:
\[
\alpha  \ge {\log _x}{\gamma _{g}}  \times {\log _\rho }x = {\log
_\rho }{\gamma _{g}}
\]
On the other hand, due to $x \ge \rho$, ${\gamma _{g}} >1$, we get
\[{\log _\rho }{\gamma _{g}}  = \frac{{\log {\gamma _{g}} }}{{\log \rho }} \geq \frac{{\log {\gamma _{g}} }}{{\log x}} = {\log _x}{\gamma _{g}} \]
thus we can compute the upper bound of $\alpha$ as:
\[\begin{array}{l}
 \alpha  \le {\log _x}{\gamma _{g}}  \times \left( {{{\log }_\rho }x + 1} \right) \\
  = {\log _\rho }{\gamma _{g}}  + {\log _x}{\gamma _{g}}  \\
  \le {\log _\rho }{\gamma _{g}}  + {\log _\rho }{\gamma _{g}}  \\
  = 2{\log _\rho }{\gamma _{g}}  \\
 \end{array}\]
this concludes the proof.
\end{proof}

\section{Proof of Lemma \ref{sandwich}}

\textit{Lemma \ref{sandwich}: (Sandwich lemma) For any pair of
node $p$ and $q$, and $d_{pq} \leq r$, then
\[{B_q}\left( r \right) \subseteq {B_p}\left( {\rho r} \right) \subseteq {B_q}\left( {{\rho ^2}r} \right)\]}

\begin{proof}
(1)For any node $i$ satisfying ${d_{qi}} \le r$, i.e., $i \in
{B_q}\left( r \right)$, by the definition of the inframetric
model, ${d_{pi}} \le \rho \max \left\{ {{d_{pq}},{d_{qi}}}
\right\} \le \rho r$,thus $i \in {B_p}\left( {\rho r} \right)$,
that is,
\[{B_q}\left( r \right) \subseteq {B_p}\left( {\rho r} \right)\]

(2) For any node $j$ satisfying $j \in {B_p}\left( {\rho r}
\right)$, by the definition of the inframetric model, it follows
\[
{d_{qj}} \le \rho \left\{ {{d_{pq}},{d_{pj}}} \right\} \le {\rho
^2}r
\]

Summing up (1) and (2) conclude the proof.
\end{proof}

\section{Proof of Theorem  \ref{thm:SamplingGrowth}}

\textit{Theorem  \ref{thm:SamplingGrowth}: (Sampling efficiency in
the growth dimension) For a $\rho$-inframetric model with growth
${\gamma _g} \geq 1$, for a service node $P$, and a DNNS target
$T$ satisfying ${d_{PT}} \le r$, when selecting $3{\left(
{\frac{{{\rho ^2}}}{\beta }} \right)^\alpha }$ nodes uniformly at
random from ${B_P}\left( {\rho r} \right)$ with replacement, with
probability of at least 95\%, one of these nodes will lie in
${B_T}\left( {\beta r} \right)$, where ${\log _\rho }{\gamma _g}
\le \alpha  \le 2{\log _\rho }{\gamma _g}$ and $\beta<1$.}

\begin{proof}
since ${B_T}\left( {\beta r} \right) \subset {B_T}\left( r \right)
\subseteq {B_P}\left( {\rho r} \right)$  by the sandwich lemma
\ref{sandwich}, all nodes covered by ${B_T}\left( {\beta r}
\right)$ are also covered by ${B_P}\left( {\rho r} \right)$.
Therefore, we only need to sample enough nodes in ${B_P}\left(
{\rho r} \right)$ in order to sample a node located in
${B_T}\left( {\beta r} \right)$.

Furthermore, for the pair of nodes $P$ and $T$ satisfying
${d_{PT}} \le r$, it follows
\[\left| {{B_P}\left( {\rho r} \right)} \right| \le \left| {{B_T}\left( {{\rho ^2}r} \right)} \right| = \left| {{B_T}\left( {\frac{{{\rho ^2}}}{\beta }\beta r} \right)} \right|\]
Since  we know $\rho>1$, then $\frac{{{\rho ^2}}}{\beta } > {\rho
^2} > \rho$, therefore the preconditions of lemma
\ref{grid_growth} hold, by lemma \ref{grid_growth}, we can show
the relation between the ball ${{B_P}\left( {\rho r} \right)}$ and
the ball ${B_T}\left( {\beta r} \right)$ where $\beta <1$,
\[\left| {{B_P}\left( {\rho r} \right)} \right| \le \left| {{B_T}\left( {\frac{{{\rho ^2}}}{\beta }\beta r} \right)} \right| \le {\left( {\frac{{{\rho ^2}}}{\beta }} \right)^\alpha }\left| {{B_T}\left( {\beta r} \right)} \right|\]
where ${\log _\rho }{\gamma _g}  \le \alpha  \le 2{\log _\rho
}{\gamma _g}$. Therefore, the probability of uniformly sampling a
node from ${{B_P}\left( {\rho r} \right)}$ which lies in the ball
${B_T}\left( {\beta r} \right)$ is:
\[\frac{{\left| {{B_T}\left( {\beta r} \right)} \right|}}{{\left| {{B_P}\left( {\rho r} \right)} \right|}} \ge \frac{{\left| {{B_T}\left( {\beta r} \right)} \right|}}{{{{\left( {\frac{{{\rho ^2}}}{\beta }} \right)}^\alpha }\left| {{B_T}\left( {\beta r} \right)} \right|}} = \frac{1}{{{{\left( {\frac{{{\rho ^2}}}{\beta }} \right)}^\alpha }}}\]
Consequently, the probability that $3{\left( {\frac{{{\rho
^2}}}{\beta }} \right)^\alpha }$ samples are not in the ball
${B_T}\left( {\beta r} \right)$ is at most
\[{\left( {1 - \frac{1}{{{{\left( {\frac{{{\rho ^2}}}{\beta }} \right)}^\alpha }}}} \right)^{3{{\left( {\frac{{{\rho ^2}}}{\beta }} \right)}^\alpha }}} \le {\left( {\frac{1}{e}} \right)^3} \le 0.05\]
Thus, with probability more than 95\% we succeed in locating a
node lying in the ball ${B_T}\left( {\beta r} \right)$ with
$3{\left( {\frac{{{\rho ^2}}}{\beta }} \right)^\alpha }$ samples.
\end{proof}

\section{Proof of Corollary \ref{corrolaryGrowthDN2S}}

\begin{corollary}
For a relaxed inframetric model with growth ${\gamma _g}$,
according to the DNNS process in Definition
\ref{samplingProcedure}, the found nearest neighbor is a
$\frac{1}{\beta }$-approximation, and the number of search steps
is smaller than ${{{\log }_{\frac{1}{\beta }}}\Delta }$, where
$\Delta$ is the ratio of the maximum delay to the minimum delay of
all pairwise delays.
\end{corollary}

\begin{proof}
If a DNNS request is forwarded from node $P$ to node $Q$, the
\textit{progress} is said to be $\frac{d_{PT}}{d_{QT}}$. According
to the DNNS search process, by Theorem \ref{thm:SamplingGrowth},
the progress is at least $\frac{1}{\beta}$ at every node $P$,
therefore in at most ${{{\log }_{\frac{1}{\beta }}}\Delta }$
steps, we reach some node $v$ satisfying $d_{vT} <
{{\frac{1}{\beta }}d_*}$, which terminates the DNNS query process
as we can not find suitable next-hop neighbors, where $d_*$ is the
minimum delay to target $T$. Therefore, the found nearest neighbor
$v$ is $\frac{1}{\beta }$-approximation.
\end{proof}

\end{document}